\documentclass{lmcs}
\pdfoutput=1

\usepackage{lastpage}
\lmcsdoi{21}{2}{19}
\lmcsheading{}{\pageref{LastPage}}{}{}%
{May~14,~2024}{Jun.~05,~2025}{}

\usepackage[utf8]{inputenc}
\usepackage{mypackages}
\usepackage{mymacros}

\title{Stochastic Window Mean-Payoff Games}

\author[L. Doyen]{Laurent Doyen\lmcsorcid{0000-0003-3714-6145}}[a]
\author[P. Gaba]{Pranshu Gaba\lmcsorcid{0009-0000-8012-780X}}[b]
\author[S. Guha]{Shibashis Guha\lmcsorcid{0000-0002-9814-6651}}[b]

\address{CNRS \& LMF, ENS Paris-Saclay, Gif-sur-Yvette, France}
\email{ldoyen@lmf.cnrs.fr}

\address{Tata Institute of Fundamental Research, Mumbai, India}
\email{pranshu.gaba@tifr.res.in, shibashis@tifr.res.in}
\thanks{This work is partially supported by the Indian Science and Engineering Research Board (SERB) grant SRG/2021/000466 and by the Indo-French Centre for the Promotion of Advanced Research (IFCPAR)}


\usetikzlibrary{automata,positioning,arrows,shapes.geometric}
\tikzset{
    ->,
    >=latex,
    node distance=1.7cm,
    every state/.style={thick, fill=gray!20, minimum size=3pt, inner sep=3pt},
    random/.style={diamond, thick, fill=gray!20, inner sep=2.5pt},
    square/.style={regular polygon,regular polygon sides=4, thick, fill=gray!20, inner sep=1.5pt},
    initial text=$ $,
}

\theoremstyle{definition}
\newtheorem{construction}[thm]{Construction}

\begin{document}

\begin{abstract}
    Stochastic two-player games model systems with an environment that is both adversarial and stochastic.
    The adversarial part of the environment is modeled by a player (\(\PlayerAdversary\)) who tries to prevent the system (\(\PlayerMain\)) from achieving its objective. 
    We consider finitary versions of the traditional mean-payoff objective, replacing the long-run average of the payoffs by payoff average computed over a finite sliding window.
    Two variants have been considered: in one variant, the maximum window length is fixed and given, while in the other, it is not fixed but is required to be bounded. 
    For both variants, we present complexity bounds and algorithmic solutions for computing strategies for \(\PlayerMain\) to ensure that the objective is satisfied with positive probability, with probability~$1$, or with probability at least~$p$, regardless of the strategy of \(\PlayerAdversary\).
    The solution crucially relies on a reduction to the special case of non-stochastic two-player games.
    We give a general characterization of prefix-independent objectives for which this reduction holds.
    The memory requirement for both players in stochastic games is also the same as in non-stochastic games by our reduction.
    Moreover, for non-stochastic games, we improve upon the upper bound for the memory requirement of \(\PlayerMain\) and upon the lower bound for the memory requirement of \(\PlayerAdversary\).
\end{abstract}

\maketitle

\section{Introduction}%
We consider two-player turn-based stochastic games played on graphs. 
Games are a central model in computer science, in particular for the verification and synthesis of reactive systems~\cite{LNCS2500,CH12,FV97}. 
A stochastic game is played by two players on a graph with stochastic transitions,  where the players represent the system and the adversarial environment, while the objective represents the functional requirement that the synthesized  system aims to satisfy with a probability $p$ as large as possible. 
The vertices of the graph are partitioned into system, environment, and probabilistic vertices.
A stochastic game is played in infinitely many rounds, starting by initially placing a token on some vertex. 
In every round, if the token is on a system or an environment vertex, then the owner of the vertex chooses a successor vertex; if the token is on a probabilistic vertex, then the successor vertex is chosen according to a given probability distribution. 
The token moves to the successor vertex, from where the next round starts.
The outcome is an infinite sequence of vertices, which is winning for the system if it satisfies the given objective.
The associated \emph{quantitative satisfaction problem} is to decide, given a threshold $p$, whether the system can win with probability at least~$p$. 
The \emph{almost-sure problem} is the special case where $p=1$, and the \emph{positive problem} is to decide whether the system can win with positive probability.
The almost-sure and the positive problems are referred to as the \emph{qualitative satisfaction problems}.
When the answer to these decision problems is ${\sf Yes}$, it is useful to construct a winning strategy for the system that can be used as a model for an implementation that ensures the objective be satisfied with the given probability.

Traditional objectives in stochastic games are $\omega$-regular such as reachability, safety, and parity objectives~\cite{CH12}, or quantitative such as mean-payoff objectives~\cite{EM79,ZP96}.
For example, a parity objective may specify that every request of the environment  is eventually granted by the system, and a mean-payoff objective may specify  the long-run average power consumption of the system.
A well-known drawback of parity and mean-payoff objectives is that only the long-run behaviour of the system may be specified~\cite{AH94,CHH09A,HTWZ15}, allowing weird transient behaviour: for example, a system may grant all its requests but with an unbounded response time; or a system with long-run average power consumption below some threshold may exhibit arbitrarily long (but finite) sequences with average power consumption above the threshold.
This limitation has led to considering finitary versions of those objectives~\cite{CHH09A,KPV09,CDRR15}, where the sequences of undesired transient behaviours must be of fixed or bounded length.
Window mean-payoff objectives~\cite{CDRR15} are quantitative finitary objectives
that strengthen the traditional mean-payoff objective: the satisfaction of a window mean-payoff objective implies the satisfaction of the standard mean-payoff objective.
Given a length \(\WindowLength \geq 1\), the fixed window mean-payoff objective (\(\FWMPL\)) is satisfied if except for a finite prefix, from every point in the play, there exists a window of length at most \(\WindowLength\) starting from that point such that the mean payoff of the window is at least a given threshold.
In the bounded window mean-payoff objective (\(\BWMP\)), it is sufficient that there exists some length \(\WindowLength\) for which the \(\FWMPL\) objective is satisfied.

\paragraph*{Contributions.}
We present algorithmic solutions for stochastic games with 
window mean-payoff objectives, and show that the positive and almost-sure problems are solvable in polynomial time for \(\FWMPL\) (Theorem~\ref{thm:posfwmp_algorithm_correctness}), 
and in \(\NP \cap \coNP\) for \(\BWMP\) (Theorem~\ref{thm:quantitative_bwmp}).
The complexity result for the almost-sure problem entails that the quantitative satisfaction problem is in \(\NP \cap \coNP\) (for both the fixed and bounded version), using standard techniques for solving stochastic games with prefix-independent objectives~\cite{CHH09}.
Note that the \(\NP \cap \coNP\) bound for the quantitative satisfaction problem matches the special case of reachability objectives in simple stochastic games~\cite{Con92}, and thus would require a major breakthrough to be improved.

As a consequence, using the \(\FWMPL\) objective instead of the standard mean-payoff objective provides a stronger guarantee on the system, and even with a polynomial complexity for the positive and the almost-sure problems (which is not known for mean-payoff objectives), and at no additional computational cost for the  quantitative satisfaction problem. 
The solution relies on a reduction to non-stochastic two-player games (stochastic games without probabilistic vertices). 
The key result is to show that in order to win positively from some vertex of the game graph, it is necessary to win from some vertex of the non-stochastic game obtained by transforming all probabilistic vertices into adversarial vertices.
This condition, which we call the sure-almost-sure ({\sf SAS}) property (Definition~\ref{def:prop}), was used to solve finitary Streett objectives~\cite{CHH09}.
We follow a similar approach and generalize it to prefix-independent objectives satisfying the {\sf SAS} property (Theorem~\ref{thm:twop_to_stochastic}). 
The bounds on the memory requirement of optimal strategies of \(\PlayerMain\) can also be derived from the key result, and are the same as optimal bounds for non-stochastic games.
For \(\FWMPL\) and \(\BWMP\) objectives in particular, we show that the memory requirement of \(\PlayerAdversary\) is also no more than the optimal memory required by winning strategies in non-stochastic games.

As solving a stochastic game with a prefix-independent objective \(\Objective\) reduces to solving non-stochastic games with objective \(\Objective\) and showing that \(\Objective\) satisfies the \(\SAS\) property, we show that the \(\FWMPL\) and \(\BWMP\) objectives satisfy the \(\SAS\) property (Lemma~\ref{lem:fwmp-sas}, Lemma~\ref{lem:bwmp-sas}) and rely on the solution of non-stochastic games with these objectives \cite{CDRR15} to complete the reduction.

We improve the memory bounds for optimal strategies of both players in non-stochastic games. 
It is stated in~\cite{CDRR15} that \(\abs{\Vertices} \cdot \WindowLength\) memory suffices for both players, where \(|V|\) and \(\WindowLength\) are the number of vertices and the window length respectively. 
In \cite[Theorem 2]{BHR16a}, the bound is loosened to \(\bigO(w_{\max} \cdot \WindowLength^2)\) and \(\bigO(w_{\max} \cdot \WindowLength^2 \cdot \abs{\Vertices})\) for \(\PlayerMain\) and \(\PlayerAdversary\) respectively, where \(w_{\max}\) is the maximum absolute payoff in the graph, as the original tighter bounds~\cite{CDRR15} were stated without proof.
Since the payoffs are given in binary, this is exponential in the size of the input.
In contrast, we tighten the bounds stated in~\cite{CDRR15}.
We show that for \(\PlayerMain\), memory \(\WindowLength\) suffices (Theorem~\ref{thm:MemoryPlayer1}), and improve the bound on memory of \(\PlayerAdversary\) strategies as follows.
We compute the set \(W\) of vertices from which \(\PlayerAdversary\) can ensure that the mean payoff remains negative for \(\WindowLength\) steps, as well as the vertices from which \(\PlayerAdversary\) can ensure that the game reaches \(W\).
These vertices are identified recursively on successive subgames of the original input game.
If $k$ is the number of recursive calls, then we show that \(k \cdot \WindowLength\) memory suffices for \(\PlayerAdversary\) to play optimally (Theorem~\ref{thm:memoryPlayer2}).
Note that \(k \le \abs{\Vertices}\). 
We also provide a lower bound on the memory size for \(\PlayerAdversary\).
Given \(\WindowLength \ge 2\), for every \(k \ge 1\), we  construct a graph with a set \(V\) of vertices such that \(\PlayerAdversary\) requires at least \(k+1=\frac{1}{2}(|V|-\WindowLength+3)\) memory to play optimally (Theorem~\ref{thm:player-two-memory-lower-bound-g-k-l}).
This is an improvement over the result in~\cite{CDRR15} which showed that memoryless strategies do not suffice for \(\PlayerAdversary\). 
Our result is quite counterintuitive since given an open window (a window in which every prefix has a total weight less than \(0\)) that needs to be kept open for another \(j \le \WindowLength\) steps from a vertex \(v\), one would conjecture that it is sufficient for a \(\PlayerAdversary\)~winning strategy to choose an edge from \(v\) that leads to the minimum payoff over paths of length \(j\).
Thus for every \(j\), \(\PlayerAdversary\) should choose a fixed edge and hence memory of size \(\WindowLength\) should suffice.
However, we show that this is not the case.

To the best of our knowledge, this work leads to the first study of stochastic games with finitary quantitative objectives.

\paragraph*{Related work.}
Window mean-payoff objectives were first introduced in~\cite{CDRR15} for non-stochastic games, where solving \(\FWMPL\) was shown to be in \(\PTime\) and \(\BWMP\) in \(\NP \cap \coNP\).
These have also been studied for Markov decision processes (MDPs) in~\cite{BDOR20,BGR19}.
In~\cite{BDOR20}, a threshold probability problem has been studied, while in~\cite{BGR19}, the authors studied the problem of maximising the expected value of the window mean-payoff.
Positive, almost-sure, and quantitative satisfaction of \(\BWMP\) in MDPs are in \(\NP \intersection \coNP\) \cite{BDOR20}, the same as in non-stochastic games.

Parity objectives can be viewed as a special case of mean-payoff~\cite{Jurdzinski98}. 
A bounded-window parity objective has been studied in~\cite{Horn07,CHH09A} where a play satisfies the objective if from some point on, there exists a bound \(\WindowLength\) such that from every state with an odd priority a smaller even priority occurs within at most \(\WindowLength\) steps.
Non-stochastic games with bounded window parity objectives can be solved in \(\PTime\)~\cite{Horn07,CHH09A}.
Stochastic games with bounded window parity objectives have been studied in~\cite{CHH09} where the positive and almost-sure problems are in \(\PTime\) while the quantitative satisfaction problem is in \(\NP \cap \coNP\).
A fixed version of the window parity objective has been studied for two-player games and shown to be \(\PSpace\)-complete~\cite{WZ16}.
Another window parity objective has been studied in~\cite{BHR16} for which both the fixed and the bounded variants have been shown to be in \(\PTime\) for non-stochastic games.
The threshold problem for this objective has also been studied in the context of MDPs, and both fixed and bounded variants are in \(\PTime\)~\cite{BDOR20}.
Finally, synthesis for \emph{bounded} eventuality properties in LTL is 2-\(\EXPTime\)-complete~\cite{KPV09}.

\paragraph*{Outline.}
In Section~\ref{sec:preliminaries}, we provide the necessary technical preliminaries.
In Section~\ref{sec:window_mean_payoff}, we define window mean-payoff objectives and state relevant decision problems for stochastic games.
In Section~\ref{sec:memory-non-stochastic-games}, we give improved bounds on the memory requirements of the players' strategies for non-stochastic games with window mean-payoff objectives.
In Section~\ref{sec:reducing_stochastic_games_to_two_player_games}, we define a property of prefix-independent objectives that allows one to solve stochastic games by reducing them to non-stochastic games.
Finally, as an application, in Section~\ref{sec:reducing_stochastic_window_mean_payoff_games}, we provide solutions to problems in stochastic games with window mean-payoff objectives.

\section{Preliminaries}%
\label{sec:preliminaries}

\paragraph*{Probability distributions.}
A \emph{probability distribution} over a finite nonempty set \(A\) is a function  \(\Prob \colon A \to [0,1] \) such that \(\sum_{a \in A} \Prob(a) = 1\).  
The \emph{support} of a probability distribution \(\Prob\) over \(A\), denoted by \(\Support{\Prob}\), is the set of all elements $a$ in \(A\) such that \(\Prob(a) > 0\).
We denote by \(\DistributionSet{A}\) the set of all probability distributions over \(A\).
For the algorithmic and complexity results,  we assume that probabilities are given as rational numbers.

\paragraph*{Stochastic games.}
We consider two-player turn-based zero-sum stochastic games (or simply, stochastic games in the sequel). 
The two players are referred to as \(\PlayerMain\) and \(\PlayerAdversary\).
A \emph{stochastic game} is a weighted directed graph \(\Game = \tuple{(\Vertices, \Edges), (\VerticesMain, \VerticesAdversary, \VerticesRandom),  \ProbabilityFunction, \PayoffFunction}\), where:
\begin{itemize}
    \item \((\Vertices, \Edges)\) is a directed graph with 
    a set \(\Vertices\) of vertices and a set  \(\Edges \subseteq \Vertices \times \Vertices\) of directed edges such that
    for all vertices \(\Vertex \in \Vertices\), 
    the set \(\OutNeighbours{\Vertex} = \{ \Vertex' \in \Vertices \suchthat (\Vertex, \Vertex') \in \Edges\}\) 
    of out-neighbours of \(\Vertex\) is nonempty, i.e., \(\OutNeighbours{\Vertex} \ne \emptyset\) (no deadlocks).
    A stochastic game is said to be finite if~\(\Vertices\) is a finite set, and infinite otherwise. Unless mentioned otherwise, stochastic games considered in this paper are finite;
    \item \((\VerticesMain, \VerticesAdversary, \VerticesRandom)\) is a partition of \(\Vertices\). The vertices in \(\VerticesMain\) belong to \(\PlayerMain\), the vertices in~\(\VerticesAdversary\) belong to \(\PlayerAdversary\), and the vertices in \(\VerticesRandom\) are called probabilistic vertices;
    \item \(\ProbabilityFunction \colon \VerticesRandom \to \DistributionSet{\Vertices}\) is a \emph{probability function} that describes the behaviour of probabilistic vertices in the game. 
    It maps the probabilistic vertices \(\Vertex \in \VerticesRandom\) to a probability distribution \(\ProbabilityFunction(\Vertex)\) over the set \(\OutNeighbours{\Vertex}\) of out-neighbours of \(\Vertex\) such that \(\ProbabilityFunction(\Vertex)(\Vertex') > 0\) for all \(\Vertex' \in \OutNeighbours{\Vertex}\) (i.e., all neighbours have nonzero probability);
    \item  \(\PayoffFunction \colon \Edges \to \Rationals\) 
    is a \emph{payoff function} assigning a rational payoff to every edge in the game.
\end{itemize}
\noindent 
Stochastic games are played in rounds. 
The game starts by initially placing a token on some vertex.
At the beginning of a round, if the token is on a vertex~\(\Vertex\),
and \(\Vertex \in \Vertices[i]\) for \(i \in \{\Main, \Adversary\}\), 
then \(\Player{i}\) chooses an out-neighbour \(\Vertex' \in \OutNeighbours{\Vertex}\);
otherwise \(\Vertex \in \VerticesRandom\), and an out-neighbour \(v' \in \OutNeighbours{\Vertex}\) is chosen with probability \(\ProbabilityFunction(\Vertex)(v')\). 
\(\PlayerMain\) receives from \(\PlayerAdversary\) the amount \(\PayoffFunction(\Vertex, \Vertex')\) given by the payoff function, 
and the token moves to \(v'\) for the next round. 
This continues ad infinitum resulting in an infinite sequence \(\Play = \Vertex[0] \Vertex[1] \Vertex[2] \dotsm \in \Vertices^\omega\) such that \((\Vertex[i], \Vertex[i+1]) \in \Edges\) for all \(i \ge 0\).

A stochastic game with \(\VerticesRandom = \emptyset\) is called a non-stochastic two-player game, 
a stochastic game with \(\VerticesAdversary = \emptyset\) is called a Markov decision process (MDP), 
a stochastic game with  \(\VerticesAdversary = \VerticesRandom = \emptyset\) is called a one-player game, 
and a stochastic game with \(\VerticesMain = \VerticesAdversary = \emptyset\) is called a Markov chain. 
We use pronouns ``she/her'' for \(\PlayerMain\) and ``he/him'' for \(\PlayerAdversary\). 
\figurename~\ref{fig:swmp_example} shows an example of a stochastic game. 
In figures, \(\PlayerMain\) vertices are shown as circles, \(\PlayerAdversary\) vertices as boxes, and probabilistic vertices as diamonds. 

\begin{figure}[t]
    \centering
    \begin{tikzpicture}
        \node[random, draw] (v1) {\(\Vertex[1]\)};
        \node[random, draw, left of=v1, xshift=-13mm] (v2) {\(\Vertex[2]\)};
        \node[random, draw, right of=v1, xshift=+13mm] (v3) {\(\Vertex[3]\)};
        \node[state, below left of=v2] (v4) {\(\Vertex[4]\)};
        \node[state, below right of=v2] (v5) {\(\Vertex[5]\)};
        \node[square, draw, below left of=v3] (v6) {\(\Vertex[6]\)};
        \node[state, below right of=v3] (v7) {\(\Vertex[7]\)};
        \node[state, left of=v2, xshift=-13mm] (v8) {\(\Vertex[8]\)};
        \node[state, left of=v8] (v9) {\(\Vertex[9]\)};
        \draw 
              (v1) edge[above, pos=0.3] node{\(\EdgeValues{0}{.3}\)} (v2)
              (v1) edge[above, pos=0.3] node{\(\EdgeValues{0}{.7}\)} (v3)
              (v2) edge[] node[above, xshift=-2mm]{\(\EdgeValues{0}{.1}\)} (v4)
              (v2) edge[auto] node{\(\EdgeValues{0}{.9}\)} (v5)
              (v3) edge[bend right, above left] node[yshift=-2mm]{\(\EdgeValues{0}{.2}\)} (v6)
              (v3) edge[auto] node{\(\EdgeValues{0}{.8}\)} (v7)
              (v4) edge[loop left] node{\(\EdgeValues{0}{}\)} (v4)
              (v5) edge[loop right] node{\(\EdgeValues{-1}{}\)} (v5)
              (v6) edge[bend right, below right] node{\(\EdgeValues{-1}{}\)}  (v3)
              (v6) edge[loop left] node{\(\EdgeValues{0}{}\)} (v6)
              (v7) edge[loop right] node{\(\EdgeValues{0}{}\)}(v7)
              (v8) edge[above, pos=0.3] node{\(\EdgeValues{0}{}\)} (v2)
              (v8) edge[bend right, above] node{\(\EdgeValues{-1}{}\)} (v9)
              (v9) edge[bend right, below] node{\(\EdgeValues{0}{}\)} (v8)
        ;
    \end{tikzpicture}
    \caption{A stochastic game. 
    \(\PlayerMain\) vertices are denoted by circles, \(\PlayerAdversary\) vertices are denoted by boxes, and probabilistic vertices are denoted by diamonds.
    The payoff for each edge is shown in red and probability distribution out of probabilistic vertices is shown in blue.}
    \label{fig:swmp_example}
\end{figure}

\paragraph*{Plays and prefixes.}
A \emph{play} in \(\Game\) is an infinite sequence \(\Play = \Vertex[0] \Vertex[1] \cdots \in \Vertices^{\omega}\) of vertices  such that \((\Vertex[i], \Vertex[i+1]) \in \Edges\) for all \(i \ge 0\). 
We denote by \(\occ(\Play)\) the set of vertices in \(\Vertices\) that occur at least once in \(\Play\), and by  \(\inf(\Play)\) the set of vertices in \(\Vertices\) that occur infinitely often in \(\Play\).
For \(i < j\), we denote by \(\PlayInfix{i}{j}\) the \emph{infix} \(\Vertex[i] \Vertex[i+1] \dotsm \Vertex[j]\) of \(\Play\). 
Its length is \( \abs{\PlayInfix{i}{j}} = j - i\), the number of edges. 
We denote by \(\PlayPrefix{j}\) the finite \emph{prefix} \(\Vertex[0] \Vertex[1] \dotsm \Vertex[j]\) of \(\Play\), and by \(\PlaySuffix{i}\) the infinite \emph{suffix} \(\Vertex[i] \Vertex[i+1] \ldots \) of \(\Play\). 
We denote by \(\PlaySet{\Game}\) and \(\PrefixSet{\Game}\) the set of all plays and the set of all prefixes in \(\Game\) respectively; the symbol \(\Game\) is omitted when it can easily be derived from the context. 
We denote by \(\Last{\Prefix}\) the last vertex of a prefix \(\Prefix \in \PrefixSet{\Game}\). 
We denote by \(\PrefixSet[i]{\Game}\) (\(i \in \{\Main, \Adversary\}\)) the set of all prefixes \(\Prefix\) such that \(\Last{\Prefix} \in \Vertices[i]\). 
The \emph{cone} at \(\Prefix\) is the set \(\Cone{\Prefix} \define \{\Play \in \PlaySet{\Game} \suchthat \Prefix \text{ is a prefix of } \Play \}\), the set of all plays having \(\Prefix\) as a prefix.

\paragraph*{Objectives.}
An \emph{objective} \(\Objective\) is a Borel-measurable subset of \(\PlaySet{\Game}\)~\cite{BK08}.
A play \(\Play \in \PlaySet{\Game}\) \emph{satisfies} an objective \(\Objective\) if \(\Play \in \Objective\). 
In a (zero-sum) stochastic game \(\Game\) with objective \(\Objective\), the objective of \(\PlayerMain\) is \(\Objective\), and the objective of \(\PlayerAdversary\) is the complement set \(\ComplementObjective{\Objective} = \PlaySet{\Game} \setminus \Objective\). 
Given \(T \subseteq\Vertices\), we define some qualitative objectives:
\begin{itemize}
    \item the \emph{reachability} objective \(\ReachObj_{\Game}(T) \define \{ \Play \in \PlaySet{\Game} \suchthat T \intersection \occ(\Play) \ne \emptyset \}\), the set of all plays that visit a vertex in \(T\),
    \item the dual \emph{safety} objective \(\SafeObj_{\Game}(T) \define \{ \Play \in \PlaySet{\Game} \suchthat \occ(\Play) \subseteq T \}\),
    the set of all plays that never visit a vertex outside \(T\),
    \item the \emph{\Buchi} objective \(\textsf{\Buchi}_{\Game}(T) \define \{ \Play \in \PlaySet{\Game} \suchthat T \intersection \inf(\Play) \ne \emptyset \}\), the set of all plays that visit a vertex in~\(T\) infinitely often, and
    \item the dual \emph{co\Buchi} objective \(\textsf{co\Buchi}_{\Game}(T) \define \{ \Play \in \PlaySet{\Game} \suchthat \inf(\Play) \subseteq T \}\), the set of all plays that eventually only visit vertices in \(T\). 
\end{itemize}
\noindent 
An objective \(\Objective\) is \emph{closed under suffixes} if for all plays \(\Play\) satisfying \(\Objective\), 
all suffixes of \(\Play\) also satisfy \(\Objective\), that is,  \(\PlaySuffix{j} \in \Objective\) for all \(j \ge 0\). 
An objective \(\Objective\) is \emph{closed under prefixes} if for all plays \(\Play \) satisfying \(\Objective\), for all prefixes \(\Prefix\) such that the concatenation \(\Prefix \cdot \Play\) is a play in \(\Game\), i.e., \(\Prefix \cdot \Play \in \PlaySet{\Game}\), we have that \(\Prefix \cdot \Play \in \Objective\).
An objective \(\Objective\) is \emph{prefix-independent} if it is closed under both prefixes and suffixes.
An objective \(\Objective\) is closed under suffixes if and only if the complement objective \(\ObjectiveBar\) is closed under prefixes. 
Thus, an objective \(\Objective\) is prefix-independent if and only if its complement \(\ObjectiveBar\) is prefix-independent. 
The reachability objective is closed under prefixes, the safety objective is closed under suffixes, and the \Buchi\ and co\Buchi\ objectives are closed under both prefixes and suffixes.

\paragraph*{Strategies.}
A \emph{strategy} for \(\PlayerI \in \{\Main, \Adversary\}\) in a game \(\Game\) is a function  \(\Strategy[i]: \PrefixSet[i]{\Game} \to \DistributionSet{\Vertices}\) that maps prefixes ending in a vertex \(\Vertex \in \Vertices[i]\) to a probability distribution over the out-neighbours of \(\Vertex\). 
That is, a strategy prescribes a randomized move for the player, taking into account the history seen so far. 
A strategy \(\Strategy[\iLabel]\) is \emph{deterministic} (or \emph{pure}) if for all prefixes \(\Prefix \in \PrefixSet[i]{\Game}\), the support \(\Support{\Strategy[\iLabel](\Prefix)}\) is a singleton, that is, \(\Strategy[\iLabel](\Prefix)\) is a single vertex with probability~\(1\).
A deterministic strategy \(\Strategy[\iLabel]\) can be viewed as a function \(\Strategy[i]: \PrefixSet[i]{\Game} \to \Vertices\).
Unless mentioned otherwise, the strategies considered in this paper are deterministic.

The set of all strategies of Player \(i \in \{\Main, \Adversary\}\) in the game \(\Game\) is denoted by \(\StrategySet[\Game]{i}\), or \(\StrategySet{i}\) when \(\Game\) is clear from the context.
Strategies can be realized as the output of a (possibly infinite-state) Mealy machine.
A \emph{Mealy machine} is a transition system with transitions labeled by a pair of symbols: one from the input alphabet and one from an output alphabet. 
For each state \(q\) of the Mealy machine and every letter \(a\) of the input alphabet, there is exactly one transition defined from state \(q\) on reading the letter \(a\). 
Formally, a Mealy machine \(M\) is a tuple \((Q, q_0, \Sigma_i, \Sigma_o, \Delta, \delta)\) 
where 
\(Q\) is the set of states of \(M\) (the memory of the induced strategy), \(q_0 \in Q\) is the initial state,
\(\Sigma_i\) is the input alphabet, 
\(\Sigma_o\) is the output alphabet,
\(\Delta \colon Q \times \Sigma_i \to Q\) is a transition function that reads the current state of~\(M\) and an input letter and returns the next state of  \(M\), and \(\delta \colon Q \times \Sigma_i \to \Sigma_o \) is an output function that reads the current state of
\(M\) and an input letter and returns an output letter.

The transition function \(\Delta\) can be extended to a function \(\hat{\Delta} \colon Q \times \Sigma_i^+ \to Q\) that reads words and can be defined inductively by \(\hat{\Delta}(q, a) = \Delta(q, a)\) and \(\hat{\Delta}(q, x \cdot a) = \Delta(\hat{\Delta}(q, x), a) \), for \(q \in Q\), \(x \in \Sigma_i^+\), and \(a \in \Sigma_i\).
The output function~\(\delta\) can be also be similarly extended to a function  \(\hat{\delta} \colon Q \times \Sigma_i^+ \to \Sigma_o\) on words and can be defined inductively by \(\hat{\delta}(q, a) = \delta(q, a)\) and \(\hat{\delta}(q, x \cdot a) = \delta(\hat{\Delta}(q, x), a) \), for \(q \in Q\), \(x \in \Sigma_i^+\), and \(a \in \Sigma_i\).

A player's strategy can be defined by a Mealy machine whose input and output alphabets are \(\Vertices\) and \(\Vertices \union \{\epsilon\}\) respectively. 
For \(i \in \{\Main, \Adversary\}\), a strategy \(\Strategy[i]\) of \(\PlayerI\) can be defined by a Mealy machine \((Q, q_0, V, V \union \{\epsilon\}, \Delta, \delta)\) as follows: 
Given a prefix \(\Prefix \in \PrefixSet[i]{\Game}\) ending in a \(\PlayerIDash\) vertex, the strategy \(\Strategy[i]\) defined by a Mealy machine is \(\Strategy[i](\Prefix) = \hat{\delta}(q_0, \Prefix)\).
Intuitively, in each turn, if the token is on a vertex~\(v\) that belongs to \(\PlayerI\) for \(i \in \{\Main, \Adversary\}\), then~\(v\) is given as input to the Mealy machine, and the Mealy machine outputs the successor vertex of~\(v\) that \(\PlayerI\) must choose. 
Otherwise, the token is on a vertex \(v\) that either belongs to \(\PlayerI\)'s opponent or is a probabilistic vertex, in which case, the Mealy machine outputs the symbol \(\epsilon\) to denote that \(\PlayerI\) cannot decide the successor vertex of \(v\).
The \emph{memory size} of a strategy \(\Strategy[i]\) is the smallest number of states a Mealy machine defining \(\Strategy[i]\) can have. 
A strategy $\Strategy[\iLabel]$ is \emph{memoryless} if $\Strategy[\iLabel](\Prefix)$ only depends on the last element of the prefix~\(\Prefix\), that is for all prefixes \(\Prefix, \Prefix' \in \PrefixSet[\iLabel]{\Game}\) if \(\Last{\Prefix} = \Last{\Prefix'}\), then \(\Strategy[\iLabel](\Prefix) = \Strategy[\iLabel](\Prefix')\).
Memoryless strategies can be defined by Mealy machines with only one state. 

A strategy \emph{profile} \(\Profile = (\Strategy[\Main], \Strategy[\Adversary])\) is a pair of strategies \(\Strategy[\Main] \in \StrategySet[\Game]{\Main}\) and \(\Strategy[\Adversary] \in \StrategySet[\Game]{\Adversary}\). 
A play \(\Play = \Vertex[0] \Vertex[1] \dotsm\) is \emph{consistent} with a strategy \(\Strategy[i] \in \StrategySet{i}\) (\(i \in \{\Main, \Adversary\}\)) if for all \(j \geq 0\), we have that if \(\Vertex[j] \in \Vertices[i]\), then \(\Vertex[j+1] = \Strategy[i](\PlayPrefix{j})\). 
A play \(\Play\) is an \emph{outcome} of a profile \(\Profile = (\Strategy[\Main], \Strategy[\Adversary])\) if it is consistent with both \(\Strategy[\Main]\) and \(\Strategy[\Adversary]\). 
We denote by \(\Pr{\Strategy[\Main], \Strategy[\Adversary]}{\Game, \Vertex}{\Objective}\) the probability that an outcome of the profile \((\Strategy[\Main], \Strategy[\Adversary])\) in \(\Game\) with initial vertex \(\Vertex\) satisfies \(\Objective\). 
First, we define this probability measure over cones inductively as follows. 
If \(\abs{\Prefix} = 0\), then \(\Prefix \) is just a vertex \(\Vertex[0]\), and 
\(\Pr{\Strategy[\Main], \Strategy[\Adversary]}{\Game, \Vertex}{\Cone{\Prefix}}\) is $1$ if $\Vertex = \Vertex[0]$, and $0$ otherwise. 
For the inductive case \(\abs{\Prefix} > 0\), 
there exist \(\Prefix' \in \PrefixSet{\Game}\) and \(\Vertex' \in \Vertices\) such that \(\Prefix = \Prefix' \cdot \Vertex'\), and we have
\begin{equation*}
    \Pr{\Strategy[\Main], \Strategy[\Adversary]}{\Game, \Vertex}{\Cone{\Prefix' \cdot \Vertex'}} = 
    \begin{cases}
        \Pr{\Strategy[\Main], \Strategy[\Adversary]}{\Game, \Vertex}{\Cone{\Prefix'}}  \cdot \ProbabilityFunction(\Last{\Prefix'})(\Vertex') & \text{if } \Last{\Prefix'} \in \VerticesRandom,\\
        \Pr{\Strategy[\Main], \Strategy[\Adversary]}{\Game, \Vertex}{\Cone{\Prefix'}}   
        & \text{if } \Last{\Prefix'} \in \Vertices[i] \text{ and } \Strategy[i](\Prefix') = \Vertex', \\
        0 & \text{otherwise.}\\
    \end{cases}
\end{equation*}
It is sufficient to define \(\Pr{\Strategy[\Main], \Strategy[\Adversary]}{\Game, \Vertex}{\Objective}\) on cones in \(\Game\) since a measure defined on cones extends to a unique measure on \(\PlaySet{\Game}\) by Carath\'{e}odory's extension theorem~\cite{Billingsley86}.

\paragraph*{Non-stochastic two-player games.}
A stochastic game without probabilistic vertices (that is, with \(\VerticesRandom = \emptyset\)) is called a \emph{non-stochastic two-player game}
(or simply, non-stochastic game in the sequel). 
In a non-stochastic game \(\Game\) with objective $\Objective$, a strategy 
\(\Strategy[\iLabel]\) is \emph{winning} for \(\PlayerI\) (\(\iLabel \in \{\Main, \Adversary\}\)) if
every play in \(\Game\) consistent with \(\Strategy[\iLabel]\) satisfies the objective \(\Objective\).
A vertex \(\Vertex \in \Vertices\) is \emph{winning} for \(\Player{\iLabel}\) in \(\Game\) if \(\Player{\iLabel}\) has a winning strategy in \(\Game\) when the initial vertex is \(\Vertex\).
The set of vertices in \(\Vertices\) that are winning for \(\Player{\iLabel}\) in \(\Game\) is the \emph{winning region} of \(\Player{\iLabel}\) in \(\Game\), denoted \(\TwoPlayerWinningRegion[\Game]{\iLabel}{\Objective}\). 
If a vertex $v$ belongs to the winning region of Player $i$ ($i \in \{1,2\}$), then Player $i$ is said to play \emph{optimally} from $v$ if they follow a winning strategy.
By fixing a strategy \(\Strategy[i]\) of \(\Player{i}\) in a non-stochastic game \(\Game\) we obtain a (possibly infinite) one-player game \(\Game^{\Strategy[i]}\) with vertices \(\Vertices \times Q\), where \(Q\) is the set of states of a Mealy machine defining \(\Strategy[\iLabel]\).

\paragraph*{Subgames.}
Given a stochastic game \(\Game = \tuple{(\Vertices, \Edges), (\VerticesMain, \VerticesAdversary, \VerticesRandom),  \ProbabilityFunction, \PayoffFunction}\), a subset \(\Vertices' \subseteq \Vertices\) of vertices \emph{induces} a subgame if $(i)$
every vertex  \(\Vertex' \in \Vertices'\) has an outgoing edge in $\Vertices'$, that is
\(\OutNeighbours{\Vertex'} \cap \Vertices' \neq \emptyset\), and $(ii)$
every probabilistic vertex \(\Vertex' \in \VerticesRandom \intersection \Vertices'\) has all 
outgoing edges in~$\Vertices'$, that is \(\OutNeighbours{\Vertex'} \subseteq \Vertices'\).
The induced \emph{subgame} is 
\(\tuple{(\Vertices',  \Edges'),
(\VerticesMain \intersection \Vertices', \VerticesAdversary \intersection \Vertices', \VerticesRandom \intersection \Vertices'), \ProbabilityFunction',
\PayoffFunction'}\), where \(\Edges' = \Edges \intersection (\Vertices' \times \Vertices')\), and \(\ProbabilityFunction'\) and \(\PayoffFunction'\) are restrictions of \(\ProbabilityFunction\) and \(\PayoffFunction\) respectively to \((\Vertices', \Edges')\).
We denote this subgame by \(\Game \restriction \Vertices'\).
Let \(\Objective\) be an objective in the stochastic game \(\Game\). We define the restriction of \(\Objective\) to a subgame~\(\Game'\) of \(\Game\) 
to be the set of all plays in \(\Game'\) satisfying \(\Objective\), that is, the set \(\PlaySet{\Game'} \intersection \Objective\).

\paragraph*{Satisfaction probability.}
A strategy \(\Strategy[\Main]\) of \(\PlayerMain\) is \emph{winning} with probability \(p\) from an initial vertex \(\Vertex\) in \(\Game\) for objective \(\Objective\) if \(\Pr{\Strategy[\Main], \Strategy[\Adversary]}{\Game, \Vertex}{\Objective} \ge p\) for all strategies \(\Strategy[\Adversary]\) of \(\PlayerAdversary\).
A strategy \(\Strategy[\Main]\) of \(\PlayerMain\) is \emph{positive} winning (resp., \emph{almost-sure} winning) from \(v\) for \(\PlayerMain\) in~\(\Game\) with objective \(\Objective\) if 
\(\Pr{\Strategy[\Main], \Strategy[\Adversary]}{\Game, \Vertex}{\Objective} > 0\) (resp., \(\Pr{\Strategy[\Main], \Strategy[\Adversary]}{\Game, \Vertex}{\Objective} = 1\))
for all strategies \(\Strategy[\Adversary]\) of \(\PlayerAdversary\).
We refer to positive and almost-sure winning as \emph{qualitative} satisfaction of~\(\Objective\), while for arbitrary $p \in [0,1]$, we call it \emph{quantitative} satisfaction. 
We denote by \(\PosWinningRegion[\Game]{\Main}{\Objective}\) (resp., by \(\ASWinningRegion[\Game]{\Main}{\Objective}\)) the positive (resp., almost-sure) winning region of \(\PlayerMain\), i.e., the set of all vertices in \(\Game\) from which \(\PlayerMain\) has a positive (resp., almost-sure) winning strategy for \(\Game\) with objective~\(\Objective\). 
If a vertex $v$ belongs to the positive (resp., almost-sure) winning region of \(\PlayerMain\), then \(\PlayerMain\) is said to play \emph{optimally} from $v$ if she follows a positive (resp., almost-sure) winning strategy from $v$.
We omit analogous definitions for \(\PlayerAdversary\).

\paragraph*{Positive attractors and traps.}
The \(\PlayerIDash\)~\emph{positive attractor} (\(i \in \{\Main, \Adversary\}\)) to \(T \subseteq \Vertices\), denoted \(\PosAttr{\iLabel}{T}\), is the set  \(\PosWinningRegion[\Game]{\iLabel}{\ReachObj(T)}\) of vertices in \(\Vertices\) from which \(\PlayerI\) can ensure that the token reaches a vertex in \(T\) with positive probability.
It can be computed as the least fixed point  of the operator $\lambda x. \PosPre_{\iLabel}(x) \cup T$ where $\PosPre_1(x) = \{\Vertex \in \VerticesMain \cup \VerticesRandom \mid \OutNeighbours{\Vertex} \cap x \neq \emptyset\} \cup  \{\Vertex \in \VerticesAdversary \mid \OutNeighbours{\Vertex} \subseteq x \}$  is the positive predecessor operator for \(\PlayerMain\),  and $\PosPre_2$ is defined analogously for \(\PlayerAdversary\).
Intuitively $\PosPre_{\iLabel}(x)$ is the set of vertices from which \(\PlayerI\) has a strategy to ensure with positive probability that the vertex in the next round is in $x$.
It is possible to compute the positive attractor in \(\bigO(\abs{\Edges})\) time~\cite{CH08}.
It is easy to derive from the computation of \(\PosAttr{\iLabel}{T}\) a memoryless strategy for \(\Player{\iLabel}\) that ensures the positive satisfaction of $\ReachObj(T)$ from vertices in \(\PosAttr{\iLabel}{T}\).
We call such a strategy a \emph{positive-attractor strategy} of \(\Player{\iLabel}\). 
Given a set \(T\), we denote the standard notion of an attractor to \(T\) from the literature by \(\Attr{\iLabel}{T}\).
In non-stochastic games, a positive-attractor to the set \(T\) is the same as a standard attractor to \(T\).

A \emph{trap} for \(\PlayerMain\) is a set \(T \subseteq \Vertices\) such that  for every vertex \(\Vertex \in T\), if $\Vertex \in \VerticesMain \cup \VerticesRandom$, then $\OutNeighbours{\Vertex} \subseteq T$, and if $\Vertex \in \VerticesAdversary$, then $\OutNeighbours{\Vertex} \cap T \neq \emptyset$, that is $\PosPre_{\Main}(\Vertices \setminus T) = \emptyset$. 
In other words, from every vertex \(\Vertex \in T\), \(\PlayerAdversary\) can ensure (with probability~$1$) that the game  never leaves~\(T\), moreover using a memoryless strategy.
A trap for \(\PlayerAdversary\) can be defined analogously. 

\begin{rem}\label{rem:trap-pref-ind}
    Let \(\Game\) be a non-stochastic game with objective \(\Objective\) for \(\PlayerMain\). 
    If \(\Objective\) is closed under suffixes, then the winning region of \(\PlayerMain\) is a trap for \(\PlayerAdversary\). 
    As a corollary, if~\(\Objective\) is prefix-independent, then the winning region of \(\PlayerMain\) is a trap for \(\PlayerAdversary\) and the winning region of \(\PlayerAdversary\) is a trap for \(\PlayerMain\).
\end{rem}

\section{Window mean payoff}%
\label{sec:window_mean_payoff}
We consider two types of window mean-payoff objectives, introduced in~\cite{CDRR15}:  
\((i)\) \emph{fixed window mean-payoff} objective (\(\FWMPL\)) in which a window length \(\WindowLength \geq 1\) is given,  
and \((ii)\) \emph{bounded window mean-payoff} objective (\(\BWMP\)) in which for every play, we need a bound on window lengths. 
We define these objectives below.

For a play \(\Play\) in a stochastic game \(\Game\), the \emph{total payoff} of an infix \(\PlayInfix{i}{i+n} = \Vertex[i] \Vertex[i+1] \dotsm \Vertex[i+n]\) is the sum of the payoffs of the edges in the infix and  is defined as \(\mathsf{TP}(\PlayInfix{i}{i+n}) = \sum_{k=i}^{i+n-1} \PayoffFunction(\Vertex[k], \Vertex[k+1])\). 
The \emph{mean payoff} of an infix \(\PlayInfix{i}{i+n} \)  is the average of the payoffs of the edges in the infix and  is defined as \(\mathsf{MP}(\PlayInfix{i}{i+n}) = \frac{1}{n} \mathsf{TP}(\PlayInfix{i}{i+n}) = \sum_{k=i}^{i+n-1} \frac{1}{n} \PayoffFunction(\Vertex[k], \Vertex[k+1])\). 
The mean payoff of a play $\pi$ is defined as \(\mathsf{MP}(\Play) = \displaystyle{\liminf_{n \rightarrow \infty}}\;\mathsf{MP}(\PlayInfix{0}{n})\).

Given a window length \(\WindowLength \geq 1\), a play \(\Play = \Vertex[0] \Vertex[1] \dotsm \)  in \(\Game\) satisfies the \emph{fixed window mean-payoff objective} \(\FWMPL[\Game]\) if  from every position after some point, it is possible to start an infix of length at most $\WindowLength$ with a nonnegative mean payoff.
Formally, 
\begin{equation*}
    \FWMPL[\Game] = \{ \Play \in \PlaySet{\Game} \mid \exists k \geq 0 \cdot \forall i \ge k \cdot \exists j \in \PositiveSet{\WindowLength}: \mathsf{MP}(\Play(i, i + j)) \ge 0\}. 
\end{equation*} 
We omit the subscript \(\Game\) when it is clear from the context. 
In this definition, there is no loss of generality in considering mean-payoff threshold $0$ rather than some \(\Threshold \in \Rationals\): consider the game \(\Game'\) obtained by subtracting \(\Threshold\) from every edge payoff in \(\Game\), and the mean payoff of any infix of a play in \(\Game\) is at least \(\Threshold\) if and only if its mean payoff in \(\Game'\) is nonnegative. 
Moreover, observe that the mean payoff of an infix is nonnegative if and only if the total-payoff of the infix is nonnegative.

Note that when \(\WindowLength = 1\), the \(\FWMP(1)\) and \(\overline{\FWMP(1)}\) (i.e., the complement of \(\FWMP(1)\)) objectives reduce to co\Buchi\ and \Buchi\ objectives respectively. 
To see this, let \(T\) be the set of all vertices \(v \in \Vertices\) such that either \(v \in \VerticesMain\) and all out-edges of \(v\) have a negative payoff, or \(v \in \VerticesAdversary\) and at least one out-edge of \(v\) has a negative payoff. 
Then, a play satisfies the \(\overline{\FWMP(1)}\) objective if and only if it satisfies the \textsf{\Buchi}\((T)\) objective. 
The following properties of \(\FWMPL\) have been observed in~\cite{CDRR15}.
The fixed window mean-payoff objective provides a robust and conservative approximation of the 
traditional mean-payoff objective, defined as the set of plays with nonnegative 
mean payoff: for all window lengths \(\WindowLength\), if a play $\Play$ satisfies \(\FWMPL[\Game]\), then it has a nonnegative mean payoff.
Since $\WindowLength \leq \WindowLength'$ implies $\FWMPObj[\Game]{\WindowLength} \subseteq \FWMPObj[\Game]{\WindowLength'}$, more precise approximations of mean payoff can be obtained by increasing the window length. 
In all plays satisfying \(\FWMPL\), there exists a suffix that can be decomposed into infixes of length at most $\WindowLength$, each with a nonnegative mean payoff. 
Such a desirable robust property is not guaranteed by the classical mean-payoff objective, where infixes of unbounded lengths may have negative mean payoff.

As defined in~\cite{CDRR15}, given a play \(\Play = \Vertex[0] \Vertex[1] \cdots\) and \(0 \le i<  j\), we say that the window \(\Play(i, j)\) is \emph{open} if the total-payoff of \(\Play(i,k)\) is negative for all \(i < k \le j\). 
Otherwise, the window is \emph{closed}.
Given \(j > 0\), we say a window is open at \(j\) if there exists an open window \(\Play(i, j)\) for some \(i < j\). 
The window starting at position \(i\) \emph{closes} at position \(j\) if \(j\) is the first position after \(i\) such that the total payoff of \(\Play(i, j)\) is nonnegative. 
If the window starting at \(i\) closes at \(j\), then for all \(i \le k < j\), the windows \(\Play(k, j)\) are closed. 
This property is called the \emph{inductive property of windows}.
A play \(\Play\) satisfies \(\FWMPL\) if and only if, from some point on, every window in \(\Play\) closes within at most \(\WindowLength\) steps.

We also consider the bounded window mean payoff objective \(\BWMP[\Game]\). 
We omit the subscript \(\Game\) when it is clear from the context.
A play \(\Play\) satisfies the \(\BWMP\) objective if there exists a window length \(\WindowLength \ge 1\) for which \(\Play\) satisfies \(\FWMPL\). Formally,
\begin{align*}
    \BWMP[\Game] 
    & = \{ \Play \in \PlaySet{\Game} \suchthat \exists \WindowLength \ge 1 : \Play \in \FWMPL \}.
\end{align*}
Equivalently, a play \(\Play\) does not satisfy \(\BWMP\) if and only if for every suffix of \(\Play\), for all \(\WindowLength \ge 1\), the suffix contains an open window of length \(\WindowLength\).
Note that both \(\FWMPL\) for all \(\WindowLength \ge 1\) and \(\BWMP\) are prefix-independent objectives.

\paragraph*{Decision problems.}
Given a game \(\Game\), an initial vertex \(\Vertex \in \Vertices\), a rational threshold  \(p \in [0, 1]\), and an objective \(\Objective\)  (that is either \(\FWMPL\) for a given window length \(\WindowLength \ge 1\), or \(\BWMP\)),
consider the problem of deciding:
\begin{itemize} 
    \item  \emph{Positive satisfaction of \(\Objective\)}:
    if \(\PlayerMain\) positively wins \(\Objective\) from \(\Vertex\), i.e., if \(\Vertex \in \PosWinningRegion[\Game]{\Main}{\Objective}\).
    \item  \emph{Almost-sure satisfaction of \(\Objective\)}:
    if \(\PlayerMain\) almost-surely wins \(\Objective\) from \(\Vertex\), i.e.,
    if \(\Vertex \in \ASWinningRegion[\Game]{\Main}{\Objective}\).
    \item  \emph{Quantitative satisfaction  of \(\Objective\)}
    (also known as \emph{quantitative value problem}~\cite{CHH09}):
    if \(\PlayerMain\) wins \(\Objective\) from \(\Vertex\) with probability at least \(p\), i.e., if \(\sup_{\Strategy[\Main] \in \StrategySet{\Main}} \inf_{\Strategy[\Adversary] \in \StrategySet{\Adversary}} \Pr{\Strategy[\Main], \Strategy[\Adversary]}{\Game, \Vertex}{\Objective} \ge p\).
\end{itemize}
\noindent 
Note that these three problems coincide for non-stochastic games.
As considered in previous works~\cite{CDRR15,BGR19,BDOR20}, the window length $\WindowLength$ is usually small (typically $\WindowLength \leq \abs{V}$), and therefore we assume that $\WindowLength$ is given in unary (while the payoff on the edges is given in binary).

\paragraph*{Determinacy.}
From determinacy of  Blackwell games~\cite{Mar98}, stochastic games with window mean-payoff objectives as defined above are determined, i.e., the largest probability with which \(\PlayerMain\) is winning and the largest probability with which \(\PlayerAdversary\) is winning add up to \(1\).

\paragraph*{Algorithms for non-stochastic window mean-payoff games.}
To compute the positive and almost-sure winning regions for \(\PlayerMain\) for \(\FWMPL\), we recall intermediate objectives defined in~\cite{CDRR15}.
The \emph{good window} objective \(\GW_{\Game}(\WindowLength)\) consists of all plays \(\Play\) in \(\Game\) such that the window opened at the first position in the play closes in at most \(\WindowLength\) steps:
\[\GW_{\Game}(\WindowLength) = \{\Play \in \PlaySet{\Game} \mid \exists j \in \PositiveSet{\WindowLength} : \mathsf{MP}(\Play(0,j)) \ge 0 \}. \]
The \emph{direct fixed window mean-payoff} objective \(\DirFWMP_{\Game}(\WindowLength)\) consists of all plays \(\Play\) in \(\Game\) such that from every position in \(\Play\), the window closes in at most \(\WindowLength\) steps:
\[\DirFWMP_{\Game}(\WindowLength) = \{\Play \in \PlaySet{\Game} \mid \forall i \ge 0 :  \Play(i, \infty) \in \GW_{\Game}(\WindowLength) \}. \]
The \(\FWMPL[\Game]\) objective can be expressed in terms of \(\DirFWMP_{\Game}(\WindowLength)\):
\[\FWMPL[\Game] = \{\Play \in \PlaySet{\Game} \mid \exists k \ge 0 :  \PlaySuffix{k} \in \DirFWMP_{\Game}(\WindowLength) \}. \]

We refer to Algorithm 1, 2, and 3 from \cite{CDRR15} shown below with the same numbering. 
They compute the winning regions for \(\PlayerMain\) for the \(\FWMPL\), \(\DirFWMP(\WindowLength)\), and \(\GW(\WindowLength)\) objectives in non-stochastic games respectively.
\cite[Algorithm~2 and Algorithm~3]{CDRR15} contain subtle errors for which the fixes are known~\cite{BHR16a,Hautem18}.
In fact, a related objective that is a combination of the good window and reachability objectives was studied in~\cite{Hautem18} and~\cite{BHR16a} from which the correct algorithms can be derived.
For completeness, we include below counterexamples for the versions in~\cite{CDRR15}, along with the correct algorithms and brief explanations of correctness.

\begin{algorithm}[t]
    \caption{\(\TwoPlayerFWMP(\Game, \WindowLength)\) \cite[Algorithm~1]{CDRR15}}\label{alg:fixed_wmp}
    \begin{algorithmic}[1]
        \Require{\(\Game = \tuple{(\Vertices, \Edges), (\VerticesMain, \VerticesAdversary, \emptyset), \PayoffFunction}\), the non-stochastic game, and \(\WindowLength \geq 1\), the window length}
        \Ensure{The set of vertices in \(\Vertices\) from which \(\PlayerMain\) wins \(\FWMPL\)}
        \State \(W_{d} \assign \TwoPlayerDirectWMP(\Game, \WindowLength)\) \label{alg_line:fixed_wmp:direct}
        \If{\(W_d = \emptyset\)}
            \State \Return{\(\emptyset\)} \label{alg_line:fixed_wmp:termination}
        \Else
            \State \(A \assign \Attr{\Main}{W_d}\) \label{alg_line:fixed_wmp:attractor}
            \State \Return \(A \union \TwoPlayerFWMP(\Game \restriction (V \setminus A), \WindowLength)\)
        \EndIf
    \end{algorithmic}
\end{algorithm}
\begin{algorithm}[t]
    \caption{\(\TwoPlayerDirectWMP(\Game, \WindowLength)\)}\label{alg:direct_wmp}
    \begin{algorithmic}[1]
        \Require{\(\Game = \tuple{(\Vertices, \Edges), (\VerticesMain, \VerticesAdversary, \emptyset), \PayoffFunction}\) the non-stochastic game, and \(\WindowLength \geq 1\), the window length}
        \Ensure{The set of vertices in \(V\) from which \(\PlayerMain\) wins \(\DirFWMP(\WindowLength)\)}
        \State \(W_{gw} \assign \GoodWin(\Game, \WindowLength) \)\label{alg_line:direct_wmp:gw}
        \If{\(W_{gw} = \Vertices\) or \(W_{gw} = \emptyset\)}
            \State \Return \(W_{gw}\)\label{alg:direct_wmp:gw_call}
        \Else
            \State \(A \assign \Attr{\Adversary}{\Vertices \setminus W_{gw}}\) \label{alg:direct_wmp:attr}
            \State \Return \(\TwoPlayerDirectWMP(\Game \restriction (W_{gw} \setminus A), \WindowLength)\)
        \EndIf
    \end{algorithmic}
\end{algorithm}

\begin{algorithm}[t]
    \caption{\(\GoodWin(\Game, \WindowLength)\)  }\label{alg:good_win}
    \begin{algorithmic}[1]
    \Require{\(\Game = \tuple{(\Vertices, \Edges), (\VerticesMain, \VerticesAdversary, \emptyset), \PayoffFunction}\) the non-stochastic game, and \(\WindowLength \geq 1\), the window length}
    \Ensure{The set of vertices in \(V\) from which \(\PlayerMain\) wins \(\GW(\WindowLength)\)}
        \ForAll {\(\Vertex \in \Vertices\)}
        \State \(C_0(\Vertex) \assign 0 \)\label{alg_line:goodwin_c_0}
        \ForAll{ \(i \in \{1, \ldots, \WindowLength\}\) }
            \State \(C_{i}(\Vertex) \assign -\infty\)
        \EndFor
    \EndFor
    \ForAll{\(i \in \{1, \ldots, \WindowLength\}\)}
        \ForAll {\(\Vertex \in \VerticesMain\)}
        \State \(C_i(\Vertex) \assign \max_{(\Vertex,\Vertex') \in \Edges}\{\max\{\PayoffFunction(\Vertex,\Vertex'), \PayoffFunction(\Vertex,\Vertex') + C_{i-1}(\Vertex')\}\}\)
        \Comment{In~\cite{CDRR15}, \( \PayoffFunction(\Vertex,\Vertex') + C_{i-1}(\Vertex')\) was used instead of \(\max\{\PayoffFunction(\Vertex,\Vertex'), \PayoffFunction(\Vertex,\Vertex') + C_{i-1}(\Vertex')\}\). }%
        \label{alg_line:goodwin_c_i}\label{alg_line:goodwin_c_i_max}
        \EndFor
        \ForAll {\(\Vertex \in \VerticesAdversary\)}
            \State  \(C_i(\Vertex) \assign \min_{(\Vertex,\Vertex') \in \Edges}\{\max\{\PayoffFunction(\Vertex,\Vertex'), \PayoffFunction(\Vertex,\Vertex') + C_{i-1}(\Vertex')\}\}\)
            \Comment{In~\cite{CDRR15},
            \( \PayoffFunction(\Vertex,\Vertex') + C_{i-1}(\Vertex')\) was used
            instead of \(\max\{\PayoffFunction(\Vertex,\Vertex'), \PayoffFunction(\Vertex,\Vertex') + C_{i-1}(\Vertex')\}\). }
            \label{alg_line:goodwin_c_i_min}
        \EndFor
    \EndFor
     \State \(W_{gw} \assign \{\Vertex \in \Vertices \suchthat C_\WindowLength(\Vertex) \ge 0 \}\)  \label{alg_line:final_condition} \Comment{In~\cite{CDRR15}, \(W_{gw}\) was defined as \(\{\Vertex \in \Vertices \mid \exists i \in \{1, 2, \ldots \WindowLength\}, \ C_i(\Vertex) \ge 0 \}\) instead because \(C_i(v)\) had a different definition in~\cite{CDRR15}.
     }
    \State \Return{\(W_{gw}\)}
    \end{algorithmic}
\end{algorithm}

\begin{figure}[t]
    \begin{subfigure}[b]{0.45\textwidth}
        \centering
        \scalebox{0.8}{
            \begin{tikzpicture}
                \node[square, draw] (s1) {\(v_1\)};
                \node[state, above right of=s1, xshift=3mm, yshift=-3mm] (s2) {\(v_2\)};
                \node[state, right of=s2] (s3) {\(v_3\)};
                \node[state, below right of=s1, xshift=3mm, yshift=3mm] (s4) {\(v_4\)};
                \node[state, right of=s4] (s5) {\(v_5\)};
                \draw (s1) edge[left, above] node{\(\EdgeValues{1}{}\)} (s2)
                      (s1) edge[right, above] node{\(\EdgeValues{1}{}\)} (s4)
                      (s2) edge[bend left, above] node{\(\EdgeValues{1}{}\)} (s3)
                      (s3) edge[bend left, below] node{\(\EdgeValues{-1}{}\)} (s2)
                      (s4) edge[bend left, above] node{\(\EdgeValues{-1}{}\)} (s5)
                      (s5) edge[bend left, below] node{\(\EdgeValues{-1}{}\)} (s4);
            \end{tikzpicture}
        }
        \caption{A counterexample for computing winning region for \(\DirFWMP(\WindowLength)\) \cite[Algorithm~2]{CDRR15} with \(\WindowLength = 2\).}
        \label{fig:dirfwmp_counterexample}
    \end{subfigure}
    \hfill
    \begin{subfigure}[b]{0.45\textwidth}
        \centering
        \scalebox{0.8}{
            \begin{tikzpicture}[node distance=1.5cm]
                \node[state, draw] (v1) {\(\Vertex[1]\)};
                \node[square, draw, right of=v1] (v2) {\(\Vertex[2]\)};
                \node[state, right of=v2] (v3) {\(\Vertex[3]\)};
                \node[state, right of=v3] (v4) {\(\Vertex[4]\)};
                \node[state, below of=v2] (v5) {\(\Vertex[5]\)};
                \node[state, right of=v5] (v6) {\(\Vertex[6]\)};
                \node[state, right of=v6] (v7) {\(\Vertex[7]\)};
                \node[state, right of=v7] (v8) {\(\Vertex[8]\)};
                \draw 
                      (v1) edge[auto] node{\(\EdgeValues{-1}{}\)} (v2)
                      (v2) edge[auto] node{\(\EdgeValues{-1}{}\)} (v3)
                      (v2) edge[left] node{\(\EdgeValues{2}{}\)} (v5)
                      (v3) edge[auto] node{\(\EdgeValues{3}{}\)} (v4)
                      (v4) edge[loop right] node{\(\EdgeValues{0}{}\)} (v4)
                      (v5) edge[auto] node{\(\EdgeValues{-2}{}\)} (v6)
                      (v6) edge[auto] node{\(\EdgeValues{0}{}\)} (v7)
                      (v7) edge[auto] node{\(\EdgeValues{2}{}\)} (v8)
                      (v8) edge[loop right] node{\(\EdgeValues{0}{}\)} (v8)
                      ;
            \end{tikzpicture}
        }
        \caption{A counterexample for computing winning region for \(\GW(\WindowLength)\)  \cite[Algorithm~3]{CDRR15} with \(\WindowLength = 3\).
        }
        \label{fig:gw_counterexample}
    \end{subfigure}
    \caption{Counterexamples for algorithms in~\cite{CDRR15}}
\end{figure}

\paragraph*{Description of algorithms from~\cite{CDRR15}.}
Algorithm~\ref{alg:fixed_wmp}~\cite[Algorithm~1]{CDRR15}  computes the winning region of \(\PlayerMain\) for the \(\FWMPL\) objective. 
First (Line~\ref{alg_line:fixed_wmp:direct}) it computes the winning region $W_d$ for \(\PlayerMain\) for the \(\DirFWMP(\WindowLength)\) objective (using Algorithm~\ref{alg:direct_wmp}). 
If \(W_d\) is empty, then it is easy to show that the winning region for \(\PlayerMain\) for objective \(\FWMPL\) is also empty, and the algorithm terminates (Line~\ref{alg_line:fixed_wmp:termination}).
Otherwise, all vertices in the \(\PlayerMainDash\) attractor of $W_d$ (Line~\ref{alg_line:fixed_wmp:attractor}) are also winning (as \(\FWMPL\) is closed under prefixes), and the remaining states (i.e., the complement of the attractor) induce a smaller subgame, which can be solved (recursively) by the same algorithm. 

Algorithm~\ref{alg:direct_wmp} computes the winning region of \(\PlayerMain\) for the \(\DirFWMP(\WindowLength)\) objective. 
It does so by first computing the region $V\setminus W_{gw}$ from which \(\PlayerMain\) cannot win the good window objective \(\GW(\WindowLength)\) (Line~\ref{alg:direct_wmp:gw_call}). If \(\PlayerMain\) does not win the \(\GW(\WindowLength)\) objective, then she does not win the \(\DirFWMP(\WindowLength)\) objective either, and thus, all vertices in \(V \setminus W_{gw}\) are losing for \(\PlayerMain\).
If \(V \setminus W_{gw}\) is empty, that is, if $W_{gw} = V$, then \(\PlayerMain\) wins the \(\GW(\WindowLength)\) objective from every vertex, and it is easy to see that \(\PlayerMain\) also wins 
the \(\DirFWMP(\WindowLength)\) objective from every vertex.
Otherwise, the \(\PlayerAdversaryDash\) attractor to $V\setminus W_{gw}$ is also losing for \(\PlayerMain\). The remaining states (i.e., the complement of \(A\)) induce a smaller subgame, which can be solved (recursively) by the same algorithm.

Algorithm~\ref{alg:good_win} computes the winning region of \(\PlayerMain\) for the good window  objective \(\GW(\WindowLength)\), that is, the 
set of vertices from which \(\PlayerMain\) can close the window within at most \(\WindowLength\) steps. The algorithm uses dynamic programming to compute, for all \(v \in V\) and all lengths \(i \in \{ 1, \ldots, \WindowLength\}\), the largest payoff \(C_i(v)\) that  \(\PlayerMain\) can ensure from \(v\) within at most \(i\) steps. 
The winning region for \(\GW(\WindowLength)\) for \(\PlayerMain\) consists of all vertices \(v\) such that \(C_{\WindowLength}(v) \ge 0\).

\paragraph*{Correctness of Algorithm~\ref{alg:direct_wmp}.}
We show the correctness of Algorithm~\ref{alg:direct_wmp}, that is, we show that this algorithm correctly computes \(\TwoPlayerWinningRegion[\Game]{\Main}{\DirFWMP(\WindowLength)}\), the winning region for \(\PlayerMain\) for the \(\DirFWMP(\WindowLength)\) objective.
The proof makes use of the fact that \(\DirFWMP(\WindowLength) \subseteq \GW(\WindowLength)\), that is, if \(\PlayerMain\) does not win \(\GW(\WindowLength)\) from a vertex \(v \in V\), then she also does not win \(\DirFWMP(\WindowLength)\) from \(v\).

The algorithm successively finds vertices that are losing for \(\PlayerMain\) for the \(\GW(\WindowLength)\) objective, removes them, and recurses on the rest of the game graph. 
In Line~\ref{alg_line:direct_wmp:gw}, we have \(W_{gw} = \TwoPlayerWinningRegion[\Game]{\Main}{\GW(\WindowLength)}\), the winning region for \(\PlayerMain\) for the \(\GW(\WindowLength)\) objective.
\begin{itemize}
    \item  If \(W_{gw} = \emptyset\), then \(\PlayerAdversary\) wins \(\overline{\GW(\WindowLength)}\) from all vertices in \(V\), and therefore, \(\PlayerAdversary\) also wins \(\overline{\DirFWMP(\WindowLength)}\) from all vertices in \(V\).
    Hence, \(\TwoPlayerWinningRegion[\Game]{\Main}{\DirFWMP(\WindowLength)} = \emptyset\).
    \item Otherwise, if \(W_{gw} = V\), then \(\PlayerMain\) wins \(\GW(\WindowLength)\) from all vertices in \(G\). For all vertices \(v \in V\), starting from \(v\), \(\PlayerMain\) can ensure that the window starting at \(v\) closes in at most \(\WindowLength\) steps. 
    When the window starting at \(v\) closes, suppose the token is on a vertex \(v'\).
    By the inductive property of windows, all windows that opened after \(v\) are also closed by the time the token reaches \(v'\).
    Now, since \(v' \in \TwoPlayerWinningRegion[\Game]{\Main}{\GW(\WindowLength)}\), \(\PlayerMain\) can ensure that the window starting at \(v'\) also closes in at most \(\WindowLength\) steps. 
    In this manner, \(\PlayerMain\) closes every window in at most \(\WindowLength\) steps, resulting in an outcome that is winning for the \(\DirFWMP(\WindowLength)\) objective.
    We get that \(\TwoPlayerWinningRegion[\Game]{\Main}{\DirFWMP(\WindowLength)} = V\).
    \item Finally, suppose \(\emptyset \subsetneq W_{gw} \subsetneq V\). 
    Starting from \(V \setminus W_{gw}\), \(\PlayerAdversary\) wins the \(\overline{\GW(\WindowLength)}\) objective, and hence, also the  \(\overline{\DirFWMP(\WindowLength)}\) objective.
    Therefore, no vertex in \(V \setminus W_{gw}\) belongs to \(\TwoPlayerWinningRegion[\Game]{\Main}{\DirFWMP(\WindowLength)}\).
    Moreover, consider the \(\PlayerAdversary\) attractor \(A\) to \(V \setminus W_{gw}\) (Line~\ref{alg:direct_wmp:attr}). 
    Starting from a vertex in this attractor, \(\PlayerAdversary\) can follow a memoryless strategy to eventually reach \(V \setminus W_{gw}\). 
    Once the token reaches \(V \setminus W_{gw}\), \(\PlayerAdversary\) can ensure that a window remains open for \(\WindowLength\) steps, resulting in \(\PlayerMain\) losing. 
    Hence, no vertex in the attractor belongs to \(\TwoPlayerWinningRegion[\Game]{\Main}{\DirFWMP(\WindowLength)}\) either.
    The complement of this \(\PlayerAdversary\) attractor is a trap for \(\PlayerAdversary\) and induces a subgame.
    For all vertices \(v\) in this subgame, if \(\PlayerMain\) wins from \(v\) in the subgame, then she also wins from \(v\) in the original game as she can mimic a winning strategy from the subgame while also ensuring that the token never leaves the subgame. 
    Conversely, for all vertices \(v\) in the subgame, if \(\PlayerMain\) does not win from \(v\) in the subgame, then she also does not win from \(v\) in the original game. This is because if the token remains in the subgame forever, then \(\PlayerAdversary\) wins, and if the token ever leaves the subgame, then as discussed above, the token enters the \(\PlayerAdversary\) attractor, and \(\PlayerAdversary\) wins.
    Thus, for all vertices \(v\) in the subgame, \(\PlayerMain\) wins from \(v\) in the subgame if and only if she wins from \(v\) in the original game.
    Hence, the winning region for \(\PlayerMain\) in the subgame is equal to the winning region for \(\PlayerMain\) in the original game, and hence, the algorithm recurses on this subgame.
\end{itemize}

\paragraph*{Correctness of Algorithm~\ref{alg:good_win}.}
The following characterization of \(C_i(v)\) holds:
\begin{itemize}
    \item There exists a strategy \(\Strategy[\Main]\) of \(\PlayerMain\) such that for all strategies \(\Strategy[\Adversary]\) of \(\PlayerAdversary\), there exists \(1 \le j \le i\) such that in the outcome of the strategy profile \((\Strategy[\Main], \Strategy[\Adversary])\) with initial vertex \(v\), the total payoff in the first \(j\) steps of the outcome is at least \(C_i(v)\);
    \item  For all strategies \(\Strategy[\Main]\) of \(\PlayerMain\), there exists a  strategy \(\Strategy[\Adversary]\) of \(\PlayerAdversary\) such that for all \(1 \le j \le i\), the total payoff of the first \(j\) steps in the outcome of \((\Strategy[\Main], \Strategy[\Adversary])\) with initial vertex \(v\) is at most \(C_i(v)\).
\end{itemize}
\noindent 
A monotonicity property follows from the above characterization, namely that for all \(v \in \Vertices\), if $1 \leq i \leq j \leq \WindowLength$, then $C_i(v) \leq C_j(v)$, which can also easily be established from  Algorithm~\ref{alg:good_win}.
Note that monotonicity does not hold for $i=0$ as $C_0(v) = 0$ for all \(v \in V\), but we may have $C_1(v) < 0$
(e.g., if all outgoing edges from $v$ have negative weight). It also follows from this characterization that  $C_i(v) \geq 0$ if
\(\PlayerMain\) wins from $v$ for objective $\GW(i)$.

We now show the correctness of Algorithm~\ref{alg:good_win} to compute $C_i(\cdot)$, by induction on \(i \in \{ 1, \ldots, \WindowLength\}\). 
The base case \(i = 1\) holds since $C_0(v) =  0$ for all \(v \in V\) and the maximum possible payoff from $v$ in one step is \(C_1(v) = \max_{(v,v')\in \Edges}\{\PayoffFunction(v, v')\}\) if \(v \in \VerticesMain\) is a vertex of \(\PlayerMain\), and \(C_1(v) = \min_{(v,v')\in \Edges}\{\PayoffFunction(v,v')\}\) if \(v \in \VerticesAdversary\) is a vertex of \(\PlayerAdversary\). 
For the induction step \(i \ge 2\), assume that 
 \(C_{i-1}(v)\) is correctly computed by the algorithm
for all \(v \in V\), as the maximum payoff that \(\PlayerMain\) can ensure from $v$ in at least $1$ and at most $i-1$ steps.
Then, from a vertex \(v\) and if the edge $(v,v')$ is chosen (either by  \(\PlayerMain\) or by  \(\PlayerAdversary\)), 
the maximum payoff that \(\PlayerMain\) can ensure in at least $1$ and at most $i$ steps is either
$\PayoffFunction(v, v')$ (in $1$ step) or $\PayoffFunction(v, v') + C_{i-1}(v')$ (in at least $1+1=2$ steps and at most $1+i-1=i$ steps), whichever is greater.
Hence if $v \in \VerticesMain$ is a vertex of \(\PlayerMain\), then $C_i(v)$ is the maximum such value across the out-neighbours $v'$ of $v$, and if $v \in \VerticesAdversary$ is a vertex of \(\PlayerAdversary\), then $C_i(v)$ is the minimum, as in Line~\ref{alg_line:goodwin_c_i_max} and  Line~\ref{alg_line:goodwin_c_i_min} of the algorithm.
Finally by the characterization of $C_i(v)$, \(\PlayerMain\) wins from $v$ for the \(\GW(\WindowLength)\) objective if $C_{\WindowLength}(v) \geq 0$ (Line~\ref{alg_line:final_condition}), which by the monotonicity property, is equivalent to the condition $\exists 1 \leq i \leq \WindowLength: C_i(v) \geq 0$ used in \cite{CDRR15}.

\paragraph*{Counterexample for~\cite[Algorithm~2]{CDRR15}.}
The version of Algorithm~2 in~\cite{CDRR15} does not compute (and does not remove) the  \(\PlayerAdversaryDash\) attractor \(\Attr{\Adversary}{\Vertices \setminus W_{gw}}\) to the winning region of \(\PlayerAdversary\) for the good-window objective (Line~\ref{alg:direct_wmp:attr}).
However, it is easy to see that \(\PlayerAdversary\) can spoil the good-window objective from \(A\), not only from $\Vertices \setminus W_{gw}$. 
This may lead to incorrectly classifying some losing states as being winning (for \(\PlayerMain\)). 
Consider the non-stochastic game shown in \figurename~\ref{fig:dirfwmp_counterexample} with \(\WindowLength = 2\). 
The vertices \(\Vertex[4]\) and \(\Vertex[5]\) are losing for \(\PlayerMain\) for \(\DirFWMP(\WindowLength)\), and since \(\Vertex[1] \in \VerticesAdversary\), the vertex \(\Vertex[1]\) is also losing for \(\PlayerMain\). The remaining vertices \(\{\Vertex[2], \Vertex[3]\}\) are winning for \(\PlayerMain\). 
After computing the winning region of \(\PlayerAdversary\) for the good-window objective, which is  \(\{\Vertex[4], \Vertex[5]\}\), the winning region in the
subgame induced by \(V \setminus \{\Vertex[4], \Vertex[5]\} =  \{\Vertex[1], \Vertex[2], \Vertex[3]\}\) is \(\{\Vertex[1], \Vertex[2], \Vertex[3]\}\),
which is returned as the winning region for \(\DirFWMP(\WindowLength)\), instead of 
\(\{\Vertex[2], \Vertex[3]\}\).

\paragraph*{Counterexample for~\cite[Algorithm~3]{CDRR15}.}
The version of Algorithm~3 in~\cite{CDRR15} differs at Line~\ref{alg_line:goodwin_c_i_max} and  Line~\ref{alg_line:goodwin_c_i_min}, and we show that it does not compute the winning region for the good-window objective.

Consider the non-stochastic game \(\Game\) shown in \figurename~\ref{fig:gw_counterexample} with \(\WindowLength = 3\). 
The vertex \(\Vertex[1]\) is winning for \(\GW(3)\) since the window closes in three steps irrespective of the successor chosen by \(\PlayerAdversary\) from \(\Vertex[2]\). 
If \(\PlayerAdversary\) chooses \(\Vertex[3]\) from \(\Vertex[2]\), the window closes in three steps, whereas if \(\PlayerAdversary\) chooses \(\Vertex[5]\) from \(\Vertex[2]\), the window closes in two steps.
However, Algorithm~3 in~\cite{CDRR15} does not include \(\Vertex[1]\) in the winning region for \(\GW(3)\). 
This is because for every vertex \(\Vertex\) in the game, it computes for all \(1 \le i \le \WindowLength\), the value \(C_i(\Vertex)\) as the best payoff that \(\PlayerMain\) can ensure from \(v\) in \emph{exactly} \(i\) steps, instead of at most $i$ steps. 
The algorithm includes a vertex \(\Vertex\) in the winning region for \(\GW(\WindowLength)\) if at least one of the \(C_i(\Vertex)\) is nonnegative. 
In our example, the best payoff that \(\PlayerMain\) can ensure from \(\Vertex[1]\) in exactly one step is \(-1\), in exactly two steps is \(-2\) (corresponding to the prefix \(\Vertex[1] \Vertex[2] \Vertex[3])\), and in exactly three steps is \(-1\) (corresponding to the prefix \(\Vertex[1]\Vertex[2]\Vertex[5]\Vertex[6])\).
Thus, for all \(1 \le i \le 3\), the value \(C_i(\Vertex[1])\) is negative. 
This example shows that it is possible for \(\PlayerMain\) to ensure a nonnegative payoff in at most~\(\WindowLength\) steps despite the worst payoff in exactly \(i\) steps being negative for all \(i \in \{1, \ldots, \WindowLength\}\).

\section{Memory requirement for non-stochastic window mean-payoff games}
\label{sec:memory-non-stochastic-games}
The memory requirement  for  winning strategies 
of both \(\PlayerMain\) and \(\PlayerAdversary\)
in non-stochastic games with objective \(\FWMPL\) 
is claimed to be $\bigO(\abs{\Vertices} \cdot \WindowLength)$ without proof~\cite[Lemma~7]{CDRR15}. 
Further, the bounds are ``correctly stated'' as
\(\bigO(w_{\max} \cdot \WindowLength^2)\) and \(\bigO(w_{\max} \cdot \WindowLength^2 \cdot \abs{\Vertices})\) for \(\PlayerMain\) and \(\PlayerAdversary\) respectively, where \(w_{\max}\) is the maximum absolute payoff in the graph~\cite[Theorem 2]{BHR16a}.
We improve upon these bounds and show that memory of size~\(\WindowLength\) suffices for a winning strategy of \(\PlayerMain\). 
Furthermore, a formal argument for memory requirement for \(\PlayerAdversary\) strategies is missing in~\cite{CDRR15} which we provide here.

We show constructions of Mealy machines \(M_{\Main}^{\TwoP}\) and \(M_\Adversary^{\TwoP}\) (with at most \(\WindowLength\) and \(\abs{\Vertices} \cdot \WindowLength\) states respectively) that define winning strategies  \(\Strategy[\Main]^{\TwoP}\) and \(\Strategy[\Adversary]^{\TwoP}\) of \(\PlayerMain\) and \(\PlayerAdversary\) respectively, showing upper bounds on the memory requirements for both players.
We also present a family of games with arbitrarily large state space
where \(\PlayerAdversary\) is winning and all his winning strategies   require at least \(\frac{1}{2} (\abs{\Vertices} - \WindowLength)+3\) memory,
while it was only known that memoryless strategies are not sufficient for \(\PlayerAdversary\)~\cite{CDRR15}.

The paper~\cite{CDRR15} also has results on the analysis of the \(\BWMP\) objective. 
It has been shown that solving \(\BWMP\) for non-stochastic games is in \(\NP \intersection \coNP\),
and memoryless strategies suffice for \(\PlayerMain\), whereas \(\PlayerAdversary\) may need infinite memory strategies to play optimally. 

\vspace{3em}
\subsection{Memory requirement for Player 1 for FWMP objective}

\paragraph*{Upper bound on memory requirement for Player 1.}
We show that memory of size \(\WindowLength\) suffices for winning strategies of \(\PlayerMain\) for the \(\DirFWMP(\WindowLength)\) objective (Lemma~\ref{lem:DirFWMP-memory}), which is in turn used to show that the same memory also works for the \(\FWMPL\) objective (Theorem~\ref{thm:MemoryPlayer1}).
\begin{lem}%
\label{lem:DirFWMP-memory}
    If \(\PlayerMain\) wins in a non-stochastic game with objective \(\DirFWMP(\WindowLength)\), then \(\PlayerMain\) has a winning strategy with memory of size $\WindowLength$.
\end{lem}

\begin{proof}
    Given a game \(\Game\), let $W_d$ be the winning region of \(\PlayerMain\) in \(\Game\) for objective \(\DirFWMP(\WindowLength)\). 
    Note that the region \(W_d\) is a trap for \(\PlayerAdversary\) in \(\Game\),
    as the objective \(\DirFWMP(\WindowLength)\) is closed under suffixes.
    Every vertex in \(W_d\) is moreover winning for \(\PlayerMain\) with objective \(\GW(\WindowLength)\), by definition.
    
    A winning strategy of \(\PlayerMain\) is to play for the objective
    \(\GW(\WindowLength)\) until the window closes (which \(\PlayerMain\)
    can ensure within at most $\WindowLength$ steps), and then to restart
    with the same strategy, playing for \(\GW(\WindowLength)\) and so on. 
    Using memory space $Q = \PositiveSet{\WindowLength}$, we may store
    the number of steps remaining before the window must close to satisfy 
    \(\GW(\WindowLength)\), and reset
    the memory to $q_0 =  \WindowLength$ whenever the window closes.
    However, the window may close any time within $\WindowLength$ steps, 
    and the difficulty is to detect when this happens: how to update the memory $q = i$, 
    given the next visited vertex $v$, but independently of the history? 
    Intuitively, the memory should be updated to $q = i-1$ if the window did not close yet
    upon reaching $v$, and to $q = \WindowLength$ if it did, but that depends
    on which path was followed to reach $v$ (not just on $v$), 
    which is not stored in the memory space.
    
    The crux is to show that it is not always necessary for \(\PlayerMain\) to be able to infer
    when the window closes. 
    Given the current memory state $q = i$,
    and the next visited vertex $v$, the memory update is as follows:
    if \(C_i(v) \geq 0\) (that is, \(\PlayerMain\) can ensure the window
    from $v$ will close within $i$ steps), then we update to $q = i-1$ (decrement) although the window may or may not have closed upon reaching $v$; otherwise \(C_i(v) < 0\)
    and we update to $q = \WindowLength-1$ (reset to $\WindowLength$ and decrement) and we show that in this case the window did close. 
    Intuitively, updating to $q = i-1$ is safe even if the window
    did close, because the strategy of \(\PlayerMain\) will anyway ensure the 
    (upcoming) window is closed within $i-1 < \WindowLength$ steps.
    For the case \(C_i(v) < 0\), we want the Mealy machine to be in state \(\WindowLength\) when \(v\) is being read.
    However, there is an additional difficulty to this.
    Assume that vertex \(v'\) is read by the Mealy machine before reading \(v\).
    The Mealy machine is thus in state \(i+1\) and \(C_{i+1}(v') \geq 0\).
    Now \(C_i(v) < 0\) denotes that an open window is closed on the edge \((v',v)\), and the state of the Mealy machine should be reset to \(\WindowLength\).
    However, if \(v'\) is a \(\PlayerAdversary\) vertex, since the vertex chosen by a \(\PlayerAdversary\) strategy from \(v'\) is not known to \(\PlayerMain\) (the output on the transition of the Mealy machine is thus \(\epsilon\)), the state of the Mealy machine is updated from \(i+1\) to \(i\)  while reading $v'$.
    The state is then updated to \(\WindowLength-1\) after reading \(v\) to simulate that the window was already closed upon reaching \(v\).
    
    For $v \in \VerticesMain$, we define the next vertex chosen by the strategy as 
    $$D_i(v) = \arg \max_{(\Vertex,\Vertex') \in \Edges}\{\max\{\PayoffFunction(\Vertex,\Vertex'), \PayoffFunction(\Vertex,\Vertex') + C_{i-1}(\Vertex')\}\},$$ 
     the out-neighbour from $v$ that maximizes 
    the expression of Algorithm~\ref{alg:good_win} for the \(\GW(\WindowLength)\)
    objective, Line~\ref{alg_line:goodwin_c_i}.
    If there is more than one such out-neighbour, we choose one arbitrarily.
    Example~\ref{example:Player1_strat} illustrates how computing \(C_i(v)\) and checking if it is nonnegative is useful in constructing a winning strategy for \(\PlayerMain\).
    We see in Construction~\ref{con:Player1-Mealy-Machine} a formal description of a Mealy machine with \(\WindowLength\) states defining a winning strategy of \(\PlayerMain\) for the \(\DirFWMP(\WindowLength)\) objective. 
    This concludes the proof of Lemma~\ref{lem:DirFWMP-memory}.
\end{proof}

\begin{exa}%
\label{example:Player1_strat}
    In this example, we show why checking if \(C_i(v)\) is negative or not is sometimes necessary.
    \figurename~\ref{fig:window_close_detection_is_necessary} shows a fragment of a game where if \(\PlayerMain\) does not know if the window has closed, then she may choose a vertex that causes her to lose the \(\DirFWMP(\WindowLength)\) objective for \(\WindowLength = 4\). 
    \begin{figure}[ht]
        \centering
        \scalebox{0.8}{
        \begin{tikzpicture}
            \node[state] (v) {\(v\)};
            \node[square, draw, below left of=v, yshift=+3mm] (u5) {\(u_5\)};
            \node[square, draw, left of=u5, xshift=-5mm] (u4) {\(u_4\)};
            \node[state, left of=u4] (u3) {\(u_3\)};
            \node[state, above of=u4] (u1) {\(u_1\)};
            \node[state, right of=u1, xshift=+5mm] (u2) {\(u_2\)};
            \node[left of=u3, xshift=+3mm] (invisible1) {};
            \node[left of=u1, xshift=+3mm] (invisible7) {};
            \node[square, draw, below right of=u4, xshift=+4mm, yshift=+2mm] (u6) {\(u_6\)};
            \node[right of=u6,  xshift=-2mm, yshift=-2mm] (invisible2) {};
            \node[state, above right of=v, yshift=-5mm, xshift=+5mm] (v1) {\(v_1\)};
            \node[state, right of=v1] (v2) {\(v_2\)};
            \node[state, right of=v2] (v3) {\(v_3\)};
            \node[state, below right of=v, yshift=+5mm, xshift=+5mm] (v4) {\(v_4\)};
            \node[state, right of=v4] (v5) {\(v_5\)};
            \node[state, right of=v5] (v6) {\(v_6\)};
            \node[right of=v3, xshift=-3mm] (invisible3) {};
            \node[right of=v6, xshift=-3mm, yshift=+4mm] (invisible4) {};
            \node[right of=v6, xshift=-3mm, yshift=-4mm] (invisible5) {};
            \node[right of=u5, xshift=-3mm, yshift=-4mm] (invisible6) {};
            \draw 
                  (invisible1) edge[dashed, above, ] node{\(+5\)} (u3)
                  (invisible7) edge[dashed, above, ] node{\(+5\)} (u1)
                  (u1) edge[below] node{\(-2\)} (u2)
                  (u2) edge[below] node{\(0\)} (v)
                  (u3) edge[above] node{\(-2\)} (u4)
                  (u4) edge[above] node{\(+3\)} (u5)
                  (u4) edge[below left] node{\(+1\)} (u6)
                  (u5) edge[below] node{\(-4\)} (v)
                  (u6) edge[dashed, below left] (invisible2)
                  (v) edge[above] node{\(0\)} (v1)
                  (v1) edge[above] node{\(0\)} (v2)
                  (v2) edge[above] node{\(+4\)} (v3)
                  (v) edge[below left] node[yshift=+1mm]{\(+3\)} (v4)
                  (v4) edge[below] node{\(0\)} (v5)
                  (v5) edge[below] node{\(0\)} (v6)
                  (v3) edge[dashed, below] (invisible3)
                  (v6) edge[dashed, below] (invisible4)
                  (v6) edge[dashed, below] (invisible5)
                  (u5) edge[dashed, below] (invisible6)
                ;
        \end{tikzpicture}
        }
        \caption{%
            The successor of \(v\) that \(\PlayerMain\) should choose depends on how many more steps she has to close the window.
            If \(\PlayerMain\) does not detect that the window is closed on \((u_4, u_5)\), then she chooses \(v_4\) from~\(v\). Otherwise, she chooses~\(v_1\) from~\(v\). 
            Computing \(C_2(u_5)\) shows that it is negative and this implies that the window starting at \(u_3\) must have closed at \(u_5\).
        }
        \label{fig:window_close_detection_is_necessary}
    \end{figure}
    Suppose it is the case that all windows in the play have closed when the token reaches~\(u_1\) and~\(u_3\). 
    If the token reaches \(v\) along \(u_1u_2\), then \(\PlayerMain\) must move the token from~\(v\) to~\(D_2(v) = v_4\) as this closes the window starting at \(u_1\). 
    If \(\PlayerMain\) moves the token from \(v\) to~\(v_1\) instead, then this results in an open window \(u_1 u_2 v v_1v_2\) of length \(4\) which is not desirable for \(\PlayerMain\).
    
    On the other hand, if the token reaches \(v\) along \(u_3 u_4 u_5\), then since the window starting at~\(u_3\) closes at \(u_5\), we have that \(\PlayerMain\) must choose a successor of \(v\) such that the window starting at~\(u_5\) closes in at most three steps from \(v\).
    Hence, if \(\PlayerMain\) moves the token from~\(v\) to \(D_3(v)\), that is, \(v_1\), then the window starting at \(u_5\) closes in at most \(4\) steps. 
    However, suppose \(\PlayerMain\) does not detect that the window starting at \(u_3\) closes at~\(u_5\).
    Although the total payoff along \(u_3 u_4 u_5 v v_4\) is nonnegative (which implies that the window starting at \(u_3\) is closed at \(v_4\)), one cannot use the inductive property of windows to claim that all subsequent windows are closed at \(v_4\). 
    The inductive property of windows does not hold since the window starting at \(u_3\) closes at~\(u_5\) and this gives no information about when the window starting at~\(u_5\) closes. 
    If \(\PlayerMain\) plays from \(v\) as if she has only one more step to close the window, then she moves the token to \(D_1(v) = v_4\) and this results in an open window \(u_5 v v_4 v_5 v_6\) of length \(4\) which is undesirable for \(\PlayerMain\).
    
    Thus, if \(\PlayerMain\) never detects window closings, then this may result in open windows of length \(\WindowLength\) in the outcome. 
    Since \(C_2(u_5)\) is negative and \(u_5\) belongs to the winning region for \(\DirFWMP(\WindowLength)\), this implies that the window starting at \(u_3\) must have closed at \(u_5\) and \(\PlayerMain\) cannot continue playing as if the window did not close at \(u_5\). 
    Computing \(C_i(v)\) in general helps detect those window closings where \(\PlayerMain\) cannot continue on as if the window did not close. 
    If \(C_i(v) \ge 0\), then even if a window closes along a path when \(v\) is reached, it is safe not to detect it, and we can still construct a winning strategy if one exists.
    \qed
\end{exa}

\begin{construction}%
\label{con:Player1-Mealy-Machine}
    We construct a Mealy machine \(M_{\dsf} = (Q_{\dsf}, q_0, \Vertices, \Vertices \union \{\epsilon\}, \Delta_{\dsf}, \delta_{\dsf})\) with \(\WindowLength\) states that defines a winning strategy \(\Strategy[\dsf]\) of \(\PlayerMain\) for 
    the \(\DirFWMP(\WindowLength)\) objective, where: 
    \begin{itemize}
        \item the memory \(Q_{\dsf} = \{1, \ldots, \WindowLength\}\) 
        stores a counter (modulo $\WindowLength$), and we assume arithmetic modulo $\WindowLength$ 
        (that is, $\WindowLength+1 = 1$, $1-1=\WindowLength$, etc.);
        \item the initial state is \(q_0 = \WindowLength\) (although we show that an arbitrary
        initial state also induces a winning strategy);
        \item the input alphabet is \(\Vertices\), as the Mealy machine reads vertices of the game,
        \item the output alphabet is \(\Vertices \union \{\epsilon\}\), as the Mealy machine either outputs a vertex (upon reading a vertex of \(\PlayerMain\)) or \(\epsilon\) (upon reading a vertex of \(\PlayerAdversary\)); 
        \item The transition function 
        \(\Delta_{\dsf} \colon Q_{\dsf} \times \Vertices \to Q_{\dsf}\) is defined as follows: 
        \[
            \Delta_{\dsf}(i, v) = 
            \begin{cases}
                i-1 \pmod{\WindowLength} & \text{ if } C_i(v) \ge 0 \quad (\emph{decrement})\\
                \WindowLength - 1        & \text{ if } C_i(v) < 0   \quad(\emph{reset and decrement})\\
            \end{cases}
        \]
        \item The output function  \(\delta_{\dsf} \colon \{1, \ldots, \WindowLength\} \times \Vertices \to \Vertices \union \{\epsilon\} \) is defined as follows:
        \[
            \delta_{\dsf}(i, v) = 
            \begin{cases}
                \epsilon & \text{ if } v \in \VerticesAdversary\\
                D_i(v) & \text{ if } v \in \VerticesMain \text{ and } C_i(v) \ge 0\\
                D_{\WindowLength}(v) & \text{ if } v \in \VerticesMain \text{ and } C_i(v) < 0\\
            \end{cases}
        \]
    \end{itemize}
    Note that if $C_i(v) < 0$, then $\delta_{\dsf}(i, v) = \delta_{\dsf}(\WindowLength, v)$,
    that is, on a reset the strategy plays as if the counter was equal to $\WindowLength$.
    
    We establish the correctness of the construction as follows.
    We show that the strategy \(\Strategy[\dsf]\) of \(\PlayerMain\) defined
    by $M_{\dsf}$ is winning, that is for all 
    strategies \(\Strategy[\Adversary]\) of \(\PlayerAdversary\), the outcome \(\Play\) 
    of the strategy profile \((\Strategy[\dsf], \Strategy[\Adversary])\)
    satisfies the objective \(\DirFWMP(\WindowLength)\).
    
    We show that every window in \(\Play\) closes within at most~\(\WindowLength\) steps. 
    We split the outcome \(\Play\) into segments (where the last vertex of each segment 
    is the same as the first vertex of the next segment) such that for each segment, 
    the state of the Mealy machine is updated to \(\WindowLength - 1\) upon reading 
    the last vertex of the segment (thus also upon reading the first vertex 
    of each segment), but the Mealy machine is never in 
    state \(\WindowLength - 1\) in between.
    Note that the initial memory state of the Mealy machine is $\WindowLength$, thus the first segment
    starts at the beginning of the outcome, and the segments cover the whole
    outcome.
    Note also that the length of each segment (i.e., the number of transitions) is at most 
    \(\WindowLength\) since either the memory state 
    is either updated to \(\WindowLength - 1\), or decreased by \(1\) (modulo \(\WindowLength\)).
    
    For all segments in the outcome \(\Play\), we show that all windows that open 
    in the segment are closed by (or before) the end of the segment, from 
    which we can conclude that the objective \(\DirFWMP(\WindowLength)\) 
    is satisfied.
    
    Consider a segment \(v_{\WindowLength} v_{\WindowLength - 1} v_{\WindowLength -2}
    \cdots v_{p+1} v_{p}\) (where $p \geq 0$ since each segment has at most 
    \(\WindowLength\) transitions), and the sequence of memory states along the segment: 
    $$
        q_{\WindowLength} \xrightarrow{v_{\WindowLength}} q_{\WindowLength - 1}  \xrightarrow{v_{\WindowLength-1}} q_{\WindowLength - 2} \ \cdots \ 
        q_{p+1} \xrightarrow{v_{p+1}} q_{p}  \xrightarrow{v_{p}} q_{p-1} 
    $$
    which can be written as:
    $$
        x \xrightarrow{v_{\WindowLength}} \WindowLength - 1  \xrightarrow{v_{\WindowLength-1}} \WindowLength - 2 \ \cdots \ 
        p+1 \xrightarrow{v_{p+1}} y  \xrightarrow{v_{p}} \WindowLength - 1 
    $$
    where $q_{\WindowLength - 1} = \WindowLength - 1$ and $q_{p-1} = \WindowLength - 1$
    by the definition of segments, and $q_i = i$ for $i=p+1, \ldots, \WindowLength-2$ since the 
    counter is decremented whenever it is not reset to $\WindowLength - 1$ (and thus
    we also have $C_i(v_i) \geq 0$ for $i=p+1, \ldots, \WindowLength-1$).
    We discuss the possible values of $q_{p} = y$ (and $q_{\WindowLength} = x$ 
    for which the situation is similar). There are two possibilities:
    either $y = \WindowLength$ and the counter is decremented upon reading $v_{p}$ 
    (and thus $p=0$ and $C_p(v_{p}) = 0$), 
    or $y = p < \WindowLength$ and the counter is reset upon reading $v_{p}$ 
    (and thus $C_p(v_{p}) < 0$). It follows that in both cases $C_p(v_p) \leq 0$.
    
    Moreover at the beginning of the segment (considering $q_{\WindowLength} = x$), 
    the strategy chooses 
    $D_{\WindowLength}(v_{\WindowLength})$ upon reading $v_{\WindowLength}$
    (if $v_{\WindowLength} \in \VerticesMain$ is a player-1 vertex), 
    as either $x = \WindowLength$ and the output on a decrement is
    $D_x(v_{\WindowLength}) = D_{\WindowLength}(v_{\WindowLength})$, 
    or $x < \WindowLength$ and $C_x(v_{\WindowLength}) < 0$ (and the output
    on a reset is $D_{\WindowLength}(v_{\WindowLength})$ by definition).
    Hence we have $v_{i-1} = D_{i}(v_{i})$ 
    whenever $v_{i} \in \VerticesMain$ is a player-1 vertex,
    for all $i=p+1,\dots,\WindowLength$.
    
    We now show by induction on $i$ that 
    \(C_i(v_i) \le \mathsf{TP}(v_i \cdots v_{p+1} v_{p})\) 
    for all \(i \in\{p+1, \ldots, \WindowLength -1\}\),
    which implies, since $C_i(v_i) \geq 0$, 
    that $\mathsf{TP}(v_i \cdots v_{p+1} v_{p}) \geq 0$,
    and thus all windows in the segment close within $\WindowLength$ steps.
    
    For the base case \(i = p + 1\), since \(C_{p}(v_{p}) \leq 0\), 
    the $\max$ subexpression at Line~\ref{alg_line:goodwin_c_i_max}
    and Line~\ref{alg_line:goodwin_c_i_min} of Algorithm~\ref{alg:good_win} simplifies to $w(v_{p+1}, v_p)$, and accordingly we get either $C_{p+1}(v_{p+1}) = w(v_{p+1}, v_p)$ if 
    $v_{p+1} \in \VerticesMain$ is a player-1 vertex,
    or $C_{p+1}(v_{p+1}) \leq w(v_{p+1}, v_p)$ if 
    $v_{p+1} \in \VerticesAdversary$ is a player-2 vertex,
    which establishes the base case 
    $C_{p+1}(v_{p+1}) \leq w(v_{p+1}, v_p) = \mathsf{TP}(v_{p+1}v_{p})$.
    
    For the induction step, let \(i \in \{p + 2, \ldots, \WindowLength - 1\}\). 
    Since \(C_{i-1}(v_{i-1}) \ge 0\), we have \(C_{i}(v_i) \le w(v_{i},v_{i-1}) + C_{i-1}(v_{i-1})\) by a similar argument as in the base case.
    By the induction hypothesis, we get \(C_{i}(v_i) \le \PayoffFunction(v_i, v_{i-1}) + \mathsf{TP}(v_{i-1} \cdots v_{p}) = \mathsf{TP}(v_i \cdots v_p)\).
    
    We show that the result holds no matter which is the initial memory state of
    the Mealy machine. It suffices to remark that with initial state $q_0 = i$
    instead of $q_0 = \WindowLength$, the prefix of the outcome until the first
    reset occurs (and the first segment starts) can be considered as a truncated
    segment (thus of length at most $i \leq \WindowLength$) where the same argument 
    can be used to show that all windows close within the length of the segment.
    \qed
\end{construction}

\begin{thm}%
\label{thm:MemoryPlayer1}
    If \(\PlayerMain\) wins in a non-stochastic game \(\Game\) with objective \(\FWMPL\),
    then \(\PlayerMain\) has a winning strategy with memory of size $\WindowLength$.
\end{thm}

\begin{proof}
    Since \(\FWMPL\) is a prefix-independent objective, we have that the winning region  \(\TwoPlayerWinningRegion[\Game]{\Main}{\FWMPL}\)  of \(\PlayerMain\) is a trap for \(\PlayerAdversary\)  (Remark~\ref{rem:trap-pref-ind}), and induces a subgame, say \(\Game_0\).
    We construct a winning strategy \(\Strategy[\Main]^{\TwoP}\) for \(\PlayerMain\) in~\(\Game_0\), with memory of size \(\WindowLength\). 
    Let there be $k+1$ calls to the subroutine \(\TwoPlayerDirectWMP\) from Algorithm~\ref{alg:fixed_wmp}.
    We denote by $(W_i)_{i \in \{1, \dots, k\}}$ the nonempty $W_d$ returned by the $i^{\text{th}}$ call to the subroutine, and let \(A_i = \Attr{\Main}{W_i}\).
    The \(A_i\)'s are pairwise disjoint, and their union is \(\bigcup_{i=1}^{k} A_i = \TwoPlayerWinningRegion[\Game]{\Main}{\FWMPL}\). 
    For \(i \in \{1, \ldots, k\}\), inductively define \(\Game_i\) to be the subgame induced by the complement of \(A_i\) in \(\Game_{i-1}\).
    Since \(\DirFWMP(\WindowLength)\) is closed under suffixes, for all \(i \in \{1, \ldots, k\}\), we have that \(W_i\) is a trap for \(\PlayerAdversary\) in \(\Game_i\) (Remark~\ref{rem:trap-pref-ind}).
    
    Let \(W = \bigcup_{i=1}^{k} W_i\) be the union of the regions \(W_i\) 
    over all subgames \(\Game_i\), and let \(A = \bigcup_{i=1}^{k} (A_i \setminus W_i)\) be the union of the regions \(A_i \setminus W_i\) over all subgames \(\Game_i\), for \(i \in \PositiveSet{k}\).
    Note that \(W \intersection A = \emptyset\) and \(W \union A = \TwoPlayerWinningRegion[\Game]{\Main}{\FWMPL}\).
    
    We construct a strategy \(\Strategy[\Main]^{\TwoP}\) that plays according to 
    the (memoryless) attractor strategy in $A$, 
    and according to the winning strategy \(\Strategy[\dsf]\)
    for \(\DirFWMP(\WindowLength)\) objective (defined in Construction~\ref{con:Player1-Mealy-Machine}) in $W$. Formally, define the Mealy machine \(M_{\Main}^{\TwoP}\) with \(\WindowLength\) states that defines \(\Strategy[\Main]^{\TwoP}\). 
    The Mealy machine \(M_{\Main}^{\TwoP}\) is given by the tuple \((Q_{\Main}^{\TwoP},\WindowLength, \Vertices, \Vertices \union \{\epsilon\}, \Delta_{\Main}^{\TwoP}, \delta_{\Main}^{\TwoP})\), where
    
    \begin{itemize}
        \item  the memory \(Q_{\Main}^{\TwoP} = \{1, \ldots, \WindowLength\}\) of the Mealy machine stores a counter (modulo $\WindowLength$);
        \item the initial state is \(q_0 = \WindowLength\);
        \item the input alphabet is \(\Vertices\), as the Mealy machine reads vertices of the game;
        \item the output alphabet is \(\Vertices \union \{\epsilon\}\), as the Mealy machine either outputs a vertex (upon reading a vertex of \(\PlayerMain\)) or \(\epsilon\) (upon reading a vertex of \(\PlayerAdversary\)); 
        \item The transition function 
        \(\Delta_\Main^{\TwoP} \colon Q_\Main^{\TwoP} \times \Vertices \to Q_{\Main}^{\TwoP}\) is defined as: 
        \[
            \Delta_1^{\TwoP}(i, v) = 
            \begin{cases}
                \WindowLength     & \text{ if } v \in A  \,\,\quad (\text{follow attractor strategy})\\
                \Delta_{\dsf}(i, v)    & \text{ if } v \in W  \quad (\text{follow $\Strategy[\dsf]$ for objective \(\DirFWMP(\WindowLength)\)})\\
            \end{cases}
        \]
        \item The output function  \(\delta_\Main^{\TwoP} \colon \{1, \ldots, \WindowLength\} \times \Vertices \to \Vertices \union \{\epsilon\} \) is defined as follows. Here, \(A(v)\) is the output of a (memoryless) attractor strategy to reach the set \(W\).
        \[
            \delta_1^{\TwoP}(i, v) = 
            \begin{cases}
                \epsilon & \text{ if } v \in \VerticesAdversary\\
                A(v)     & \text{ if } v \in A \intersection \VerticesMain \\
                \delta_{\dsf}(i,v) & \text{ if } v \in W \intersection \VerticesMain\\
            \end{cases}
        \]
    \end{itemize}
    \noindent 
    We establish the correctness of the construction as follows.
    We show that the strategy \(\Strategy[\Main]^{\TwoP}\)  of \(\PlayerMain\) defined
    by \(M_{\Main}^{\TwoP}\) is winning, that is for all 
    strategies \(\Strategy[\Adversary]\) of \(\PlayerAdversary\), the outcome \(\Play\) 
    of the strategy profile \((\Strategy[\Main]^{\TwoP}, \Strategy[\Adversary])\)
    satisfies the objective \(\FWMPL\).
    
    The crux is to show that one of the sets \(W_i\) for some \(i \in \{1, \ldots, k\}\)
    is never left from some point on. Intuitively, given the token is in $A_i$
    for some \(i \in \{1, \ldots, k\}\) (thus in $\Game_i$), 
    following \(\Strategy[\Main]^{\TwoP}\) the token will either remain in
    $A_i$, or leave the subgame $\Game_i$, thus entering $A_j$ for a smaller index $j < i$.
    Repeating this argument (at most $k$ times, as the index is decreasing) shows
    that the token eventually remains in some \(W_i\) (\(i \in \{1, \ldots, k\}\)).
    From that point on, the strategy plays like $\Strategy[\dsf]$ (with some initial 
    memory state $i \in  \{1, \ldots, \WindowLength\}$) 
    which ensures objective \(\DirFWMP(\WindowLength)\) (proof of Lemma~\ref{lem:DirFWMP-memory}),
    and thus from the initial state the objective \(\FWMPL\) is satisfied.
\end{proof}

\begin{rem}
    In every play \(\Play\) consistent with \(\Strategy[\Main]^{\TwoP}\), eventually, all windows close in at most \(\WindowLength\) steps. 
    If \(\PlayerMain\) follows the strategy \(\Strategy[\Main]^{\TwoP}\), then irrespective of how \(\PlayerAdversary\)'s choices are made (whether they are deterministic or randomized),  the outcome always satisfies the \(\FWMPL\) objective.
    The proof of Theorem~\ref{thm:MemoryPlayer1} thus continues to hold even if the strategy \(\Strategy[\Adversary]\) of \(\PlayerAdversary\) is not deterministic. 
    Since the constructed strategy \(\Strategy[\Main]^{\TwoP}\) is a deterministic strategy,
    we have that deterministic strategies suffice for the \(\FWMPL\) objective for \(\PlayerMain\), and memory of size \(\WindowLength\) suffices.
\end{rem}

\paragraph*{Lower bound on memory requirement for Player 1.}
In~\cite{CDRR15}, the authors show an example of game with \(\WindowLength =4\) where \(\PlayerMain\) requires memory at least 3. 
While it is not difficult to generalize this to arbitrary \(\WindowLength\), we state it here for completeness. 
\begin{thm} \label{thm:memory_lb_player1}
    There exists a family of games \(\{\Game_{\WindowLength}\}_{\WindowLength \ge 2}\) with objective \(\FWMPL\) for \(\PlayerMain\) such that every winning strategy of \(\PlayerMain\) in \(\Game_{\WindowLength}\) requires at least \(\WindowLength - 1\) memory. 
\end{thm}
\begin{proof}
    We describe the game \(\Game_\WindowLength\) with objective \(\FWMPL\) in \figurename~\ref{fig:family-of-graphs-player-main-memory-lower-bound}.
    \begin{figure}[t]
        \centering
        \scalebox{0.8}{
            \begin{tikzpicture}
                \node[square, draw] (u0) {\(u_0\)};
                \node[state, right of=u0] (u1) {\(u_1\)};
                \node[state, right of=u1] (u2) {\(u_2\)};
                \node[state, right of=u2, xshift=+1mm] (ul1) {\(u_{\WindowLength-2}\)};
                \node[state, right of=ul1] (ul) {\(u_{\WindowLength -1}\)};
                \node[state, right of=ul, yshift=+6mm] (v11) {};
                \node[state, right of=ul, yshift=0mm] (v21) {};
                \node[state, right of=ul, yshift=-6mm] (v31) {};
                \node[state, right of=ul, yshift=-12mm] (vl1) {};
                \node[state, right of=v11] (v12) {};
                \node[state, right of=v21] (v22) {};
                \node[state, right of=v31] (v32) {};
                \node[state, right of=vl1] (vl2) {};
                \node[state, right of=v12] (v1l) {};
                \node[state, right of=v22] (v2l) {};
                \node[state, right of=v32] (v3l) {};
                \node[state, right of=vl2] (vll) {};
                \node[state, right of=v3l, xshift=+5mm] (vlll) {\(v\)};
                \draw 
                      (u0) edge[below] node{\(-1\)} (u1)
                      (u1) edge[below] node{\(0\)} (u2)
                      (u2) edge[below, dashed] node{} (ul1)
                      (ul1) edge[below] node{\(0\)} (ul)
                      (u0) edge[below, bend left=40, pos=0.6] node{\(-2\)} (u2)
                      (u0) edge[below, bend left=40, pos=0.6] node{\(-(\WindowLength-2)\)} (ul1)
                      (u0) edge[below, bend left=40, pos=0.6] node{\(-(\WindowLength-1)\)} (ul)
                      (ul) edge[above, pos=0.75] node{\(+1\)} (v11)
                      (ul) edge[above, pos=0.75] node{\(0\)} (v21)
                      (ul) edge[above, pos=0.75] node{\(0\)} (v31)
                      (ul) edge[above, pos=0.75] node{\(0\)} (vl1)
                      (v11) edge[above] node{\(0\)} (v12)
                      (v21) edge[above] node{\(+2\)} (v22)
                      (v31) edge[above] node{\(0\)} (v32)
                      (vl1) edge[above] node{\(0\)} (vl2)
                      (v12) edge[above, dashed] node{} (v1l)
                      (v22) edge[above, dashed] node{} (v2l)
                      (v32) edge[above, dashed] node{} (v3l)
                      (vl2) edge[above, dashed] node{} (vll)
                      (v1l) edge[above, pos=0.3] node{\(0\)} (vlll)
                      (v2l) edge[above, pos=0.3] node{\(0\)} (vlll)
                      (v3l) edge[above, pos=0.3] node{\(0\)} (vlll)
                      (vll) edge[below] node{\(+(\WindowLength  -1)\)} (vlll)
                      (vlll) edge[below, bend right=45] node{\(0\)} (u0)
                    ;
            \end{tikzpicture}
        }
        \caption{Memory \(\WindowLength - 1\) is necessary for \(\PlayerMain\).}
        \label{fig:family-of-graphs-player-main-memory-lower-bound}
    \end{figure}
    The vertex \(u_0\) belongs to \(\PlayerAdversary\). 
    All other vertices in the game belong to \(\PlayerMain\).
    For each \(i \in \{1, \ldots, \WindowLength - 1\}\), there is an edge from \((u_0, u_i)\) with payoff \(-i\). 
    For all \(i \in \{1, \ldots, \WindowLength - 2\}\), there is an edge \((u_i, u_{i+1})\) with payoff \(0\).
    There are \(\WindowLength - 1\) disjoint paths from \(u_{\WindowLength -1 }\) to \(v\), each of length \(\WindowLength - 1\).
    The \(i^{\text{th}}\) edge on the \(i^{\text{th}}\) path has payoff \(+i\). All other edges in all paths from \(u_{\WindowLength-1}\) to \(v\) have payoff \(0\).
    Finally, there is an edge \((v, u_0)\) with payoff \(0\).
    
    If \(\PlayerAdversary\) moves the token from \(u_0\) to \(u_i\) for \(i \in \{1, \ldots, \WindowLength -1\}\), then \(\PlayerMain\) needs to ensure a payoff of at least \(+i\) from \(u_{\WindowLength - 1}\) in at most \(i\) steps to ensure that the window starting at \(u_0\) closes in at most \(\WindowLength\) steps. 
    When \(\PlayerAdversary\) moves the token from \(u_0\) to \(u_i\), we have that \(\PlayerMain\) must take the \(i^{\text{th}}\) path from \(u_{\WindowLength - 1}\) so the window starting at \(u_0\) closes in at most \(\WindowLength\) steps.
    If \(\PlayerMain\) chooses any other successor of \(u_{\WindowLength - 1}\), then the window starting at \(u_0\) remains open for more than \(\WindowLength\) steps.
    Since there are \(\WindowLength - 1\) different choices for \(\PlayerMain\) from \(u_{\WindowLength - 1}\), a Mealy machine defining a winning strategy of \(\PlayerMain\) requires at least \(\WindowLength - 1\) distinct states. 
    Thus, a winning strategy of \(\PlayerMain\) in the game \(\Game_\WindowLength\) requires at least \(\WindowLength -1\) memory. 
\end{proof}

\begin{rem}
    The game \(\Game_{\WindowLength}\) in Figure~\ref{fig:family-of-graphs-player-main-memory-lower-bound} shows that the memory requirement of \(\PlayerMain\) is \(\WindowLength- 1\) even if she uses randomized strategies.
    This is because after a certain point in the play, each time the token reaches \(u_{\WindowLength - 1}\), \(\PlayerMain\) must choose the correct successor with probability~1 in order to win the \(\FWMPL\) objective. 
    Thus, randomization does not improve the lower bound for the size of the memory required.
\end{rem}

\subsection{Memory requirement for Player 2 for FWMP objective}

\paragraph*{Upper bound on memory requirement for Player 2.}
Now we show that for \(\FWMPLBar\) objective, \(\PlayerAdversary\) has a winning strategy that requires at most $|V| \cdot \WindowLength$ memory.
This has been loosely stated in~\cite{CDRR15} without a formal proof. We use this result to show in Section~\ref{sec:reducing_stochastic_window_mean_payoff_games} that the same memory bound for \(\PlayerAdversary\) also suffices in stochastic games.
In fact, we show a stronger result that the memory required in the stochastic window mean-payoff games is no more than the optimal memory bounds for non-stochastic games.

\begin{thm}%
\label{thm:memoryPlayer2}
    Let \(\Game\) be a non-stochastic game with objective \(\FWMPLBar\) for \(\PlayerAdversary\). 
    Then, \(\PlayerAdversary\) has a winning strategy with memory size at most \(\abs{\Vertices} \cdot \WindowLength\).
\end{thm}
\begin{proof}
    Since \(\FWMPL\) is a prefix-independent objective, so is \(\FWMPLBar\). 
    We have that \(\TwoPlayerWinningRegion[\Game]{\Adversary}{\FWMPLBar}\) is a trap for \(\PlayerMain\)  (Remark~\ref{rem:trap-pref-ind}) and induces a subgame, say \(\GameH_0\), of \(\Game\).
    Let there be $k+1$ calls to the subroutine \(\GoodWin\) from Algorithm~\ref{alg:direct_wmp}, and let \(\GameH_i\) be the subgame corresponding to the \(i^{\text{th}}\) call of the subroutine.
    We denote by $(W_i)_{i=1}^k$ the complement of $W_{gw}$ in \(\GameH_i\), where \(W_{gw}\) is returned by the $i^{\text{th}}$ call to the subroutine, and let \(A_i = \Attr{\Adversary}{W_i}\).
    The \(A_i\)'s are pairwise disjoint, and their union is \(\bigcup_{i=1}^{k} A_i = \TwoPlayerWinningRegion[\Game]{\Adversary}{\FWMPLBar}\).
    
    We describe a winning strategy for the \(\FWMPLBar\) objective with memory \(k \cdot \WindowLength\), which is at most \(\abs{\Vertices} \cdot \WindowLength\). 
    The strategy is always in either {\em attractor mode} or {\em window-open mode}.
    When the game begins, it is in attractor mode.
    If the strategy is in attractor mode and the token is on a vertex \(v \in A_i \setminus W_i\) for some \(i \in \{1, \ldots, k\}\), then the attractor strategy is to eventually reach \(W_i\). 
    If the token reaches~\(W_i\), then the strategy switches to window-open mode. 
    Since all vertices in \(W_i\) are winning for \(\PlayerAdversary\) for the \(\overline{\GW(\WindowLength)}\) objective, he can keep the window open for \(\WindowLength\) more steps, provided that \(\PlayerMain\) does not move the token out of the subgame~\(\GameH_{i}\). 
    If, at some point, \(\PlayerMain\) moves the token out of the subgame \(\GameH_{i}\) to \(A_j\) for a smaller index \(j < i\), then the strategy switches back to attractor mode, this time trying to reach \(W_j\) in the bigger subgame \(\GameH_j\).
    Otherwise, if \(\PlayerAdversary\) keeps the window open for \(\WindowLength\) steps, then the strategy switches back to attractor mode until the token reaches a vertex in \(\bigcup_{i=1}^{k} W_i\).
    This strategy can be defined by a Mealy machine \(M_{\Adversary}^{\TwoP}\) with states \(\PositiveSet{k} \times \PositiveSet{\WindowLength}\), where the first component tracks the smallest subgame \(\GameH_i\) in which the window started to remain open, and the second component indicates how many more steps the window needs to be kept open for.
    A formal description of \(M_{\Adversary}^{\TwoP}\) is given in Construction~\ref{con:twop_mealy_machine_adversary}.
\end{proof}
    
\begin{construction}
\label{con:twop_mealy_machine_adversary}
    Let \(W = \bigcup_{i=1}^{k} W_i\), and \(A = \bigcup_{i=1}^k (A_i \setminus W_i)\).
    For \(1 \le i \le k\), let \(H_i\) denote the set of vertices in the subgame \(\GameH_i\).
    For all \(v \in \Vertices\), let \(\Gamma(v)\) denote the largest \(j \ge 1\) such that \(v \in H_j\). 
    That is, $\GameH_{\Gamma(v)}$ is the smallest subgame that \(v\) belongs to.
    Given \(j \ge 1\) and \(v \in A \intersection \VerticesAdversary \intersection (H_j \setminus H_{j+1})\), let \(A^j(v)\) denote a successor vertex that \(\PlayerAdversary\) can choose to eventually reach the \(W_j\) region in \(W \intersection (H_j \setminus H_{j+1})\). This is given by a memoryless attractor strategy. 
    Given \(i \in \PositiveSet{\WindowLength}\),  \(j \ge 1\) and \(v \in \VerticesAdversary \intersection H_j\), let \(D_i^j(v)\) denote the best successor that \(\PlayerAdversary\) should choose from \(v\) to ensure that the window remains open for \(i\) more steps. 
    These values can be computed for all \(i \in \PositiveSet{\WindowLength}\) and for all \(v \in H_j\) by running the \(\GoodWin\) algorithm on the subgame~\(\GameH_j\). 
    Recall from Line~\ref{alg_line:goodwin_c_i_min} in Algorithm~\ref{alg:good_win} that \(C_i(\Vertex) = \min_{(\Vertex,\Vertex') \in \Edges}\{\max\{\PayoffFunction(\Vertex,\Vertex'), \PayoffFunction(\Vertex,\Vertex') + C_{i-1}(\Vertex')\}\}\).
    We let \(D^j_i(\Vertex)\) be the successor vertex \(v' \in E(v)\) of \(v\) such that the value of \( \{\max\{\PayoffFunction(\Vertex,\Vertex'), \PayoffFunction(\Vertex,\Vertex') + C_{i-1}(\Vertex')\}\}\) is minimized.
    If there is more than one such \(\Vertex'\), we choose one arbitrarily.
    
    We construct a Mealy machine \(M_\Adversary^{\TwoP}\) that defines \(\Strategy[\Adversary]^{\TwoP}\),  a winning strategy of \(\PlayerAdversary\).
    The Mealy machine \(M_\Adversary^{\TwoP}\) is a tuple \((Q^{\TwoP}_{\Adversary}, (1, \WindowLength), \Vertices, \Vertices \union \{\epsilon\}, \Delta_\Adversary^{\TwoP}, \delta_\Adversary^{\TwoP})\) where 
    \begin{itemize}
        \item the set of states \(Q_\Adversary^{\TwoP}\) is the set \(\PositiveSet{k} \times \PositiveSet{\WindowLength}\),
        \item the initial state of the Mealy machine is \((1, \WindowLength)\) irrespective of where the game begins in~\(\GameH_1\).
        We could also instead have set the initial state of the Mealy machine as \((\Gamma(\VertexInitial), \WindowLength)\), where \(\VertexInitial\) is the initial vertex of the game. 
        However, to keep the initial state of the Mealy machine independent of the initial vertex of the game, we have the initial state of the Mealy machine as \((1, \WindowLength)\). 
        The transitions are defined such that the state of the Mealy machine is changed to \(\Gamma(v)\) in the very next step.
        \item the input alphabet is \(\Vertices\),  same as in \(M_\Main^{\TwoP}\),
        \item the output alphabet is \(\Vertices \union \{\epsilon\}\), same as in \(M_\Main^{\TwoP}\). 
    \end{itemize}
    \noindent 
    The transition function \(\Delta_\Adversary^{\TwoP} \colon Q_\Adversary^{\TwoP} \times \Vertices \to Q_\Adversary^{\TwoP}\) is defined as follows: 
    \[
        \Delta_\Adversary^{\TwoP}(q, v) = 
        \begin{cases}
            \Delta_\Adversary^{\TwoP}((1, \WindowLength), v)  & q = (j, i), v \in H_1 \setminus H_j, \text{for all } i \in \{1, \ldots , \WindowLength\}, j \in \{2, \ldots, k\}\\
            (\Gamma(v), \WindowLength) & q = (j, \WindowLength), v \in A \intersection H_j, \text{for all } j \in \{1, \ldots, k\} \\
            (\Gamma(v), \WindowLength - 1) & q = (j, \WindowLength), v \in W \intersection H_j, \text{for all } j \in \{1, \ldots, k\}\\
            (j, i - 1) & q= (j, i), v \in H_j, \text{for all } i \in \{2, \ldots , \WindowLength - 1\}, j \in \{1, \ldots, k\} \\
            (1, \WindowLength) & q = (j, 1), v \in H_j, \text{for all } j \in \{1, \ldots, k\}\\
        \end{cases}
    \]
    Suppose the Mealy machine is in state \(q \in Q_\Adversary^{\TwoP}\) and the token is in vertex \(v \in \Vertices\). 
    We describe the definition of \(\Delta_\Adversary^{\TwoP}(q, v)\). 
    \begin{itemize}
        \item  If  \(q = (j, i)\) for \(i \in \{1, \ldots, \WindowLength\}\) and \(j \in \{2, \ldots, k\}\), but \(v \in H_1 \setminus H_j\), then this means that \(\PlayerMain\) must have moved the token out of the subgame \(\GameH_j\).
        \(\PlayerAdversary\) may have been in the process of keeping the window open for \(\WindowLength\) steps in \(\GameH_j\), but the move by \(\PlayerMain\) may have closed the window. 
        The strategy switches to attractor mode, and the state of the Mealy machine changes to \(\Delta_\Adversary^{\TwoP}((1, \WindowLength), v)\). 
        Note that \(\PlayerMain\) can only move the token out of a subgame finitely many times, and once the token is in \(H_1 \setminus H_2\), then \(\PlayerMain\) can no longer move the token out of \(\GameH_1\), and \(\PlayerAdversary\) will be able to keep the window open for \(\WindowLength \) steps without resetting the strategy. Now, for the remaining cases, suppose that if \(q = (j, i)\), then \(v \in H_j\). 
        \item If \(q = (j, \WindowLength)\) for \(j \in \{1, \ldots, k\}\) and \(v \in A \intersection H_j\), then the strategy is in attractor mode. 
        Since \(v \in H_j\), we have that \(\Gamma(v) \ge j\), i.e. \(\GameH_{\Gamma(v)}\) is a subgame of \(\GameH_j\).
        When the game begins, or when \(\PlayerAdversary\) manages to keep the window open for \(\WindowLength\) steps, the first component resets to 1 according to the fifth type of transition in the definition of \(\Delta_\Adversary^{\TwoP}(q, v)\).
        In such cases, it may happen that \(\Gamma(v) > j\). 
        Another possibility is that, at some point, the Mealy machine is in state \((j', i)\) for \(j' \in \{2, \ldots, k\}\) and \(i \in \{1, \ldots, \WindowLength\}\), and the token is in a vertex \(u \in H_{j'}\). 
        If \(u\) belongs to \(\PlayerMain\) and she moves the token from \(u\) to a vertex \(u'\) that is outside \(H_{j'}\), then in the next turn, the first kind of transition of \(\Delta_{\Adversary}^{\TwoP}\) occurs, that is, \(\Delta_{\Adversary}^{\TwoP}((1, \WindowLength), u')\).
        In this case, it is possible that \(\Gamma(u') > 1\). 
        For such scenarios, the first component of the state of the Mealy machine updates to \(\Gamma(v)\) to be an indicator of the smallest subgame that \(v\) belongs to. 
        The second component remains equal to \(\WindowLength\). 
        \item If \(q = (j, \WindowLength)\) for \(j \in \{1, \ldots, k\}\) and \(v \in W \intersection H_j\), then the strategy was in attractor mode, but switches to window-open mode in this step. 
        Again, if the game begins, or if \(\PlayerAdversary\) manages to keep the window open for \(\WindowLength\) steps, or if \(\PlayerMain\) moves the token out of a subgame to a \(W\)-vertex, it may be the case that \(\Gamma(v) > j\).  
        Thus, the first component of the state of the Mealy machine updates to \(\Gamma(v)\) to correctly reflect the smallest subgame that \(\PlayerAdversary\) begins to keep the window open in. 
        The second component decreases by one, since \(\PlayerAdversary\) has begun to keep the window open from \(v\) and only needs to keep the window open for \(\WindowLength -1\) steps after this. 
        \item If \(q = (j, i)\) for \(j \in \{1, \ldots, k\}\) and \(i \in \{2, \ldots, \WindowLength - 1\}\), and \(v \in H_j\), then the strategy is in window-open mode. 
        The state of the Mealy machine changes to \((j, i - 1)\).
        The first component of the state of the Mealy machine keeps track of the smallest subgame in which \(\PlayerAdversary\) began to keep the window open in order to decide the optimal successor vertex \(D_i^j(\Vertex)\), so it remains unchanged.
        The second component decreases to \(i-1\), and we now describe why.
        If \(v\) belongs to \(\PlayerAdversary\), and since \(H_j\) is a trap for \(\PlayerAdversary\), the successor \(v'\) of \(v\) chosen by \(\PlayerAdversary\) also belongs to \(H_j\), and \(\PlayerAdversary\) chooses the successor that lets him keep the window open for \(i-1\) more steps after \(v'\).
        On the other hand, if \(v\) belongs to \(\PlayerMain\), then the successor \(v'\) of \(v\) chosen by \(\PlayerMain\) may or may not belong to \(H_j\). 
        If~\(v'\) belongs to \(H_j\), then \(\PlayerAdversary\) can still keep the window open for \(i-1\) more steps after~\(v'\), since \(\PlayerAdversary\) has been playing optimally for the objective \(\overline{\GW(\WindowLength)}\) restricted to the subgame \(\GameH_j\).
        However, if \(\PlayerMain\) moves the token out of \(\GameH_j\), then the window may have closed, but we let the state of the Mealy machine change to \((j, i -1)\) regardless.
        In the next turn, the Mealy machine will be in state \((j, i -1)\), but \(v'\) will not belong to \(H_j\). Thus, by the first kind of transition of \(\Delta_{\Adversary}^{\TwoP}\), after the next step, the state of the Mealy machine will change to \(\Delta_{\Adversary}^{\TwoP}((1, \WindowLength), v')\). 
        The Mealy machine will behave as if its state had changed to \((1, \WindowLength)\) (instead of \((j, i -1)\)) after reading \(v\). 
        \item If \(q = (j, 1)\) for \(j \in \{1, \ldots, k\}\) and \(v \in H_j\), then the strategy is in window-open mode. 
        After this step, \(\PlayerAdversary\) has successfully kept the window open for \(\WindowLength\) steps, and the strategy switches to attractor mode. 
        The second component resets to \(\WindowLength\) to indicate that the strategy switches to attractor mode. 
        The first component of the state of the Mealy machine resets to 1, the way it was at the beginning of the game. 
        Since \(\Gamma(v) \ge 1\), the first component will correctly change to \(\Gamma(v)\) in the next step by the second or third type of transition in the definition of \(\Delta_\Adversary^{\TwoP}(q, v)\).
    \end{itemize}
    \noindent 
    The output function \(\delta_\Adversary^{\TwoP} \colon Q_\Adversary^{\TwoP} \times \Vertices \to \Vertices \union \{\epsilon \}\) is defined as follows: 
    \[
        \delta_\Adversary^{\TwoP}(q, v) = 
        \begin{cases}
            \epsilon        & q \in Q_\Adversary^{\TwoP}, v \in \VerticesMain\\
            \delta_\Adversary^{\TwoP}((1, \WindowLength), v)  & q = (j, i), v \in \VerticesAdversary \intersection H_1 \setminus H_j, \text{for all } i \in \{1, \ldots , \WindowLength\}, j \in \{2, \ldots, k\}\\
            A^{\Gamma(v)}(v) & q = (j, \WindowLength), v \in A \intersection \VerticesAdversary \intersection H_j, \text{for all } j \in \{1, \ldots , k\}, \\
            D^{\Gamma(v)}_\WindowLength(v) & q = (j, \WindowLength), v \in W \intersection \VerticesAdversary \intersection H_j, \text{for all } j \in \{1, \ldots , k\}, \\
            D_i^{j}(v)  & q = (j, i), v \in \VerticesAdversary \intersection H_j,  \text{for all } i \in \{1, \ldots, \WindowLength - 1\}, j \in \{1, \ldots, k\}\\
        \end{cases}
    \]
    Suppose the Mealy machine is in state \(q \in Q_\Adversary^{\TwoP}\) and the token is in vertex \(v \in \Vertices\). 
    We describe the definition of \(\delta_\Adversary^{\TwoP}(q, v)\). 
    \begin{itemize}
        \item If \(v \in \VerticesMain\), then the Mealy machine outputs \(\epsilon\) since a strategy of \(\PlayerAdversary\) is not defined for prefixes ending with a \(\PlayerMain\) vertex. 
        \item  If \(q = (j, i)\) for some \(i \in \{1, \ldots, \WindowLength\}\) and \(j \in \{2, \ldots, k\}\), but \(v \in H_1 \setminus H_j\), then as described in the definition of the transition function \(\Delta_\Adversary^{\TwoP}(q, v)\), this means that \(\PlayerMain\) must have moved the token out of the subgame \(\GameH_j\) when the last vertex before \(v\) was read.
        The Mealy machine behaves as if it was in state \((1, \WindowLength)\), and outputs \(\delta_\Adversary^{\TwoP}((1, \WindowLength), v)\). 
        \item If  \(q = (j, \WindowLength)\) and \(v \in A\intersection \VerticesAdversary \intersection H_j\) for \(j \in \{1, \ldots, k\}\), then the strategy is in attractor mode. \(\PlayerAdversary\) must move the token by following the attractor strategy to reach \(W_{\Gamma(v)}\). The vertex given by the attractor strategy is \(A^{\Gamma(v)}(v)\).
        \item If  \(q = (j, \WindowLength)\) and \(v \in W\intersection \VerticesAdversary \intersection H_j\) for \(j \in \{1, \ldots, k\}\), then the strategy switches to window-open mode. \(\PlayerAdversary\) has not started to keep to window open, but can now begin to keep the window open for \(\WindowLength\) steps. The best vertex to choose for this is given by the successor \(D_\WindowLength^{\Gamma(v)}(v)\). 
        \item Finally, if \(q = (j, i)\) and \(v \in \VerticesAdversary \intersection H_j\) for \(i \in \{1, \ldots, \WindowLength - 1\}\) and \(j \in \{1, \ldots, k\}\), then the strategy is in window-open mode. \(\PlayerAdversary\) has kept the window open for \(\WindowLength -i\) steps already, with the first vertex in the window from \(H_j \setminus H_{j+1}\). He must keep the window open for \(i\) more steps, and hence chooses \(D_i^j(v)\) as the successor vertex. 
    \end{itemize}
    \noindent 
    This concludes the construction of a Mealy machine defining a winning strategy of \(\PlayerAdversary\).
    \qed
\end{construction}

\begin{rem}
    The definitions \(\Delta_{\Adversary}^{\TwoP}(q,v)\) and \(\delta_{\Adversary}^{\TwoP}(q,v)\) are recursive definitions when \(q = (j,i)\) and \(v \in H_1 \setminus H_j\). 
    Recall that when this is the case, we have \(\Delta_{\Adversary}^{\TwoP}(q, v) = \Delta_{\Adversary}^{\TwoP}((1, \WindowLength), v)\)
    and \(\delta_{\Adversary}^{\TwoP}(q, v) = \delta_{\Adversary}^{\TwoP}((1, \WindowLength), v)\). 
    The output and the transition of the Mealy machine \(M_{\Adversary}^{\TwoP}\) from the state \((j, i)\) are as if the Mealy machine is actually in state \((1, \WindowLength)\). 
    If the Mealy machine is in state \((1, \WindowLength)\), then since \(j = 1\), it is not possible for \(v\) to belong to \(H_1 \setminus H_1\) because it is an empty set.
    Thus, when \(j= 1\), the transition and output functions do not make recursive calls to themselves. 
    Hence, the depth of recursion in both \(\Delta_{\Adversary}^{\TwoP}\) and \(\delta_{\Adversary}^{\TwoP}\) is never greater than~1.
\end{rem}

\begin{rem}%
\label{rem:memoryPlayer2}
    Note that the memory required for \(\PlayerAdversary\) to play optimally is \(k \cdot \WindowLength\), where \(k\) is the number of recursive calls to the {\sf GoodWin} algorithm.
    This gives a tighter bound than that claimed in~\cite{CDRR15} since $k \le |V|$.
\end{rem}

\begin{rem}
    Every play \(\Play\) that is consistent with the strategy constructed in Construction~\ref{con:twop_mealy_machine_adversary} has infinitely many open windows of length \(\WindowLength\), and therefore satisfies the \(\FWMPLBar\) objective.
    If \(\PlayerAdversary\) follows this strategy, then the outcome is always winning for him, irrespective of how \(\PlayerMain\)'s choices are made (whether deterministic or randomized).
    The proof of Theorem~\ref{thm:memoryPlayer2} thus continues to hold even if the strategy of \(\PlayerMain\) is not deterministic. 
    Since the strategy constructed in Construction~\ref{con:twop_mealy_machine_adversary} is a deterministic strategy,
    we have that deterministic strategies suffice for the \(\FWMPLBar\) objective for \(\PlayerAdversary\), and memory of size \(\abs{\Vertices} \cdot \WindowLength\) suffices.
\end{rem}

\paragraph*{Lower bound on memory requirement for Player 2.}
In~\cite{CDRR15}, it was shown that memoryless strategies do not suffice for \(\PlayerAdversary\). 
We improve upon this lower bound.
Given a window length \(\WindowLength \ge 2\), for every $k \ge 1$, we construct a graph \(\{\Game_{k,\WindowLength}\}\) with $2k+\WindowLength-1$ vertices such that every winning strategy of \(\PlayerAdversary\) in \(\{\Game_{k,\WindowLength}\}\) requires at least $k+1$ memory.

\begin{thm}%
\label{thm:player-two-memory-lower-bound-g-k-l}
    There exists a family of non-stochastic games \(\{\Game_{k,\WindowLength}\}_{k \ge 1, \WindowLength \ge 2}\) with objective \(\FWMPL\) for \(\PlayerMain\) and edge weights \(-1,0,+1\) such that every winning strategy of \(\PlayerAdversary\) requires at least \(\frac{1}{2}(\abs{\Vertices} - \WindowLength + 1) + 1\) memory, where $|V| = 2k+\WindowLength -1$.
\end{thm}
\begin{proof}
    Let \(A = \{a_1, \ldots, a_k\}\), \(B = \{b_1, \ldots, b_k\}\), and  \(C = \{c_1, \ldots, c_{\WindowLength - 1}\}\) be pairwise disjoint sets.
    The vertices of \(\Game_{k,\WindowLength}\) are \(A \union B \union C\) with \(\VerticesMain = A \union C\) and \(\VerticesAdversary = B\).
    Now we list the edges in \(\Game_{k,\WindowLength}\):
    \begin{enumerate}
        \item For all \(p \in \{1, \ldots, k\}\) and \(r \in \{1, \ldots, k\}\) such that \(p \le r\), we have an edge \((a_p, b_r)\) with payoff \(-1\).
        \item For all \(p \in \{2, \ldots, k\}\), we have an edge \((a_p, a_{p-1})\) with payoff \(+1\). 
        \item For all \(p \in \{2, \ldots, k\}\), we have an edge \((a_p, b_{p-1})\) with payoff \(+1\). 
        \item For all \(p \in \{1, \ldots, k\}\), we have an edge \((b_p, c_{\WindowLength -1})\) with payoff \(0\). 
        \item For all \(p \in \{1, \ldots, k\}\), we have an edge \((b_p, a_{p})\) with payoff \(+1\). 
        \item For all \(p \in \{2, \ldots, \WindowLength -1\}\), we have an edge \((c_p, c_{p-1})\) with payoff \(0\). 
        \item We have edges \((c_1, a_k)\) and \((c_1, b_k)\) with payoff \(+1\) each.
    \end{enumerate}
    Figure~\ref{fig:family-of-graphs-player-adversary-memory-lower-bound} shows the game \(\Game_{4, 3}\).
    
    \begin{figure}[t]
        \centering
        \scalebox{0.8}{
            \begin{tikzpicture}[
                first/.append style={fill=blue!20, xshift=5mm},
                second/.append style={fill=red!20, yshift=-2mm},
                third/.append style={fill=green!20, xshift=5mm},
                bigpositive/.append style={blue, thick},
                bignegative/.append style={red, thick},
                smallpositive/.append style={black, thick},
            ]
                \node[state, first] (a1) {\(a_1\)};
                \node[square, draw, second, below of=a1] (b1) {\(b_1\)};
                \node[state, first, right of=a1] (a2) {\(a_2\)};
                \node[square, draw, second, below of=a2] (b2) {\(b_2\)};
                \node[state, first, right of=a2] (a3) {\(a_3\)};
                \node[square, draw, second, below of=a3] (b3) {\(b_3\)};
                \node[state, first, right of=a3] (a4) {\(a_4\)};
                \node[square, draw, second, below of=a4] (b4) {\(b_4\)};
                \node[first, fill=none, right of=a4] (a5) {};
                \node[second, fill=none, below of=a5] (b5) {};
                \node[state, third, right of=a4] (c1) {\(c_1\)};
                \node[state, third, right of=c1] (c2) {\(c_2\)};
                \draw 
                      (a1) edge[bignegative]  (b2)
                      (a1) edge[bignegative]  (b3)
                      (a1) edge[bignegative]  (b4)
                    
                      (a2) edge[bignegative]  (b3)
                      (a2) edge[bignegative]  (b4)
                      
                      (a3) edge[bignegative]  (b4)

                      (a2) edge[bigpositive] (a1)
                      (a3) edge[bigpositive] (a2)
                      (a4) edge[bigpositive] (a3)
                      
                      (a2) edge[bigpositive] (b1)
                      (a3) edge[bigpositive] (b2)
                      (a4) edge[bigpositive] (b3)
                      
                      (c2) edge[smallpositive] (c1)
        
                      (c1) edge[bigpositive] (a4)
                      (c1) edge[bigpositive] (b4)
                      
                      (b1) edge[smallpositive, bend right=48] (c2)
                      (b2) edge[smallpositive, bend right=44] (c2)
                      (b3) edge[smallpositive, bend right=38] (c2)
                      (b4) edge[smallpositive, bend right=30] (c2)
                    ;
                \draw[transform canvas={xshift=-1mm}]
                      (b1) edge[bigpositive] (a1)
                      (b2) edge[bigpositive] (a2)
                      (b3) edge[bigpositive] (a3)
                      (b4) edge[bigpositive] (a4)
                      ;
                \draw[transform canvas={xshift=1mm}]
                      (a1) edge[bignegative]  (b1)
                      (a2) edge[bignegative]  (b2)
                      (a3) edge[bignegative]  (b3)
                      (a4) edge[bignegative]  (b4)
                      ;
            \end{tikzpicture}
        }
        \vspace*{-8mm}
        \caption{The game \(\Game_{4,3}\) with parameter \(k = 4\) and window length \(\WindowLength = 3\). 
        Red edges have payoff~\(-1\), black edges have payoff~\(0\), and blue edges have payoff~\(+1\). 
        Memory of size at least \(k = 4\) is needed to define a winning strategy for the \(\FWMPLBar\) objective for \(\PlayerAdversary\) in this game.}
        \label{fig:family-of-graphs-player-adversary-memory-lower-bound}
    \end{figure}
    
    Observe that the only open windows of length \(\WindowLength\) in the game \(\Game_{k,\WindowLength}\) are sequences of the form \(a_p b_{r} c_{\WindowLength -1} \cdots c_1\) for all \(p \le r\). 
    Also note that \(\PlayerAdversary\) has a winning strategy that wins starting from every vertex in the game:
    \begin{itemize}
        \item 
        If the token is in \(C\), then it eventually reaches \(c_1\). 
        From \(c_1\), the token can move to either~\(a_k\) or \(b_k\). In the latter case, \(\PlayerAdversary\) moves the token from \(b_k\) to \(a_k\). 
        In both cases, the token reaches \(A\).
        \item 
        If the token is in \(A\), then it cannot remain in \(A\) forever; it must eventually move to~\(B\). 
        Moreover, \(\PlayerAdversary\) can ensure that the token eventually moves from \(A\) to \(B\) along an edge with negative payoff. 
        Suppose whenever the token moves from \(A\) to \(B\), it moves an edge with positive payoff, i.e., from \(a_p \in A\) to \(b_{p-1} \in B\) for some \(p \in \{2, \ldots, k\}\).
        Then, \(\PlayerAdversary\) moves the token from \(b_{p-1}\) back to \(a_{p-1} \in A\).
        The token then eventually reaches~\(a_1\) from which all out-edges have negative payoff, and must necessarily move to \(B\) along an edge with negative payoff.
        \item  
        Eventually, the token moves from \(a_p\) to \(b_r\) for some \(p \le r\). 
        In this case, \(\PlayerAdversary\) moves the token from \(b_r\) to \(c_{\WindowLength - 1}\) and eventually to \(c_1\).
        The edge \((a_p, b_r)\) has a negative payoff and the \(\WindowLength -1\)  edges on the path from \(b_r\) to \(c_1\) have payoff \(0\) each.
        Hence, the window starting at \(a_p\) remains open for \(\WindowLength\) steps. 
        \item 
        Now, the token is on \(c_1\) again and \(\PlayerAdversary\) can eventually keep the window open for \(\WindowLength\) steps again.
    \end{itemize}
    \noindent 
    In this manner, \(\PlayerAdversary\) can ensure that open windows of length \(\WindowLength\) occur infinitely often in the play.

    \textit{Good choices.}
    When the token reaches a vertex \(b_r \in B\), \(\PlayerAdversary\) can either move the token to \(a_r \in A\) or to \(c_{\WindowLength - 1}\in C\). 
    Depending on which vertex the token was on before reaching \(b_r\), one of the two choices is \emph{good} for \(\PlayerAdversary\). 
    If the token reaches \(b_r\) from the left or above, i.e., from \(a_p\) for \(p \le r\), then the edge \((a_p, b_r)\) has negative payoff. 
    In this case, it is good for \(\PlayerAdversary\) to move the token to \(c_{\WindowLength - 1} \in C\) so that the window starting at \(a_p\) remains open for~\(\WindowLength\) steps. 
    Otherwise, if the token reaches \(b_r\) from the right, i.e., from \(a_{r+1}\), then it is good for \(\PlayerAdversary\) to move the token to \(a_r\) so that an edge with negative payoff may eventually be taken. 
    
    For all \(u \in A \union \{c_1\}\), for all \(b_r \in B\) such that \((u, b_r)\) is an edge in \(\Game_{k,\WindowLength}\), we denote by \(\chi(u, b_r)\) the vertex \(a_r\) or \(c_{\WindowLength -1}\) that is good for \(\PlayerAdversary\). 
    We list the good choices in the game \(\Game_{4,3}\) in Table~\ref{tab:necessary-transitions-g-4-3}. 
    The columns are indexed by \(u \in A \union \{c_1\}\) and the rows are indexed by \(b_r \in B\).
    If the edge \((u, b_r)\) does not exist in the game, then the cell corresponding to this edge is left empty in the table. 
    
    \begin{table}[t]
        \centering
        \begin{tabular}{|l|l|l|l|l|}
            \hline
            \(a_1 b_1 \to c_{2}\) &   \(a_2 b_1 \to a_1\) &                   &                   &                   \\
            \(a_1 b_2 \to c_{2}\) &   \(a_2 b_2 \to c_{2}\) & \(a_3 b_2 \to a_2\) &                   &                   \\
            \(a_1 b_3 \to c_{2}\) &   \(a_2 b_3 \to c_{2}\) & \(a_3 b_3 \to c_{2}\) & \(a_4 b_3 \to a_3\) &                   \\
            \(a_1 b_4 \to c_{2}\) &   \(a_2 b_4 \to c_{2}\) & \(a_3 b_4 \to c_{2}\) & \(a_4 b_4 \to c_{2}\) & \(c_1 b_4 \to a_4\) \\
            \hline
        \end{tabular}
        \caption{Good choices \(\chi(u, b_r)\) for all \(u \in A \union \{c_1\}\) and \(b_r \in B\) in the game \(\Game_{4, 3}\).}
        \label{tab:necessary-transitions-g-4-3}
    \end{table}
    
    In Lemma~\ref{lem:distinct-states-g-k-l}, we show that for each column in the table, there exists a distinct memory state in every Mealy machine defining a winning strategy of \(\PlayerAdversary\).
    This gives a lower bound of \(k+1\) on the number of states of such a Mealy machine.
    Since \(\Game_{k, \WindowLength}\) has \(2k + \WindowLength - 1\) vertices, the memory requirement of a winning strategy of \(\PlayerAdversary\) is at least \(\frac{1}{2}(\abs{\Vertices} - \WindowLength+1)+1\).
    This concludes the proof of Theorem~\ref{thm:player-two-memory-lower-bound-g-k-l}.
\end{proof}

\paragraph*{Constructing Mealy machines.}
We show that every winning strategy of \(\PlayerAdversary\) in \(\Game_{k,\WindowLength}\) requires at least \(k + 1\) memory. 
Let \(\Strategy[\Adversary]^{\TwoP}\) be a winning strategy of \(\PlayerAdversary\) in \(\Game_{k,\WindowLength}\), and let \(M_{\Adversary}^{\TwoP} = (Q_{\Adversary}^{\TwoP}, q_0, \Vertices, \Vertices \union \{\epsilon\}, \Delta_{\Adversary}^{\TwoP}, \delta_{\Adversary}^{\TwoP})\) be a Mealy machine defining \(\Strategy[\Adversary]^{\TwoP}\). 

For all \(u \in A \union \{c_1\}\), let \(Q_{u}\) denote the set of all states that \(M_{\Adversary}^{\TwoP}\) could be in after reading a prefix ending in \(u\), 
i.e., \(Q_{u} = \{ q \in Q_{\Adversary}^{\TwoP} \suchthat \exists \Prefix \in \PrefixSet{} : \hat{\Delta}_{\Adversary}^{\TwoP}(q_0, \Prefix \cdot u) = q \} \). 
Lemma~\ref{lem:distinct-states-g-k-l} gives a lower bound on the number of states in \(M_{\Adversary}^{\TwoP}\).

\begin{lem}
\label{lem:distinct-states-g-k-l}
    Let \(\Strategy[\Adversary]^{\TwoP}\) be a winning strategy for \(\PlayerAdversary\) for the \(\FWMPLBar\) objective, and let \(M_{\Adversary}^{\TwoP}\) be a Mealy machine defining \(\Strategy[\Adversary]^{\TwoP}\).
    Then, for all vertices \(u \in A \union \{c_1\}\), there exists a state \(q_u \in Q_{u}\) such that for all \(b_r \in B \intersection \Edges(u)\), we have that \(\delta_{\Adversary}^{\TwoP}(q_u, b_r) = \chi(u, b_r)\). 
    Moreover, the Mealy machine \(M_{\Adversary}^{\TwoP}\) has \(k + 1\) distinct states.
\end{lem}
\begin{proof}
    We prove the contrapositive.
    Suppose there exists a vertex \(u \in A \union \{c_1\}\) such that for all \(q_u \in Q_{u}\), there exists a vertex \(b_r \in B \intersection \Edges(u)\) such that \(\delta_{\Adversary}^{\TwoP}(q_u, b_r) \ne \chi(u, b_r)\).
    Then, we show that \(\Strategy[\Adversary]^{\TwoP}\) is not winning for \(\PlayerAdversary\).
    We show this by constructing a strategy \(\Strategy[\Main]\) of \(\PlayerMain\) such that the outcome \((\Strategy[\Main], \Strategy[\Adversary]^{\TwoP})\) satisfies \(\FWMPL\) and is thus losing for \(\PlayerAdversary\).
    We have that either \(u \in A\) or \(u = c_1\). 
    \begin{itemize}
        \item  Suppose \(u = c_1\). \\
        Irrespective of which vertex the game begins from, the strategy~\(\Strategy[\Main]\) tries to eventually move the token to \(c_1\). 
        If the token never reaches \(c_1\), then this implies that the token never reaches \(c_{\WindowLength-1}\),  and therefore, this implies that every time the token reaches \(B\), \(\PlayerMain\) moves it to \(A\) and not \(C\).
        Thus, no windows remain open for \(\WindowLength\) steps.
        Otherwise, the token eventually reaches \(c_1\) after having seen at most one open window of length \(\WindowLength\). 
        After the token reaches \(c_1\) for the first time,
        we show that under the assumption that for all \(q_{c_1} \in Q_{c_1}\), there exists a vertex \(b_r \in B \intersection \Edges(c_1)\) such that \(\delta_{\Adversary}^{\TwoP}(q_{c_1}, b_r) \ne \chi(c_1, b_r)\), that subsequently there are no more open windows of size \(\WindowLength\) in the outcome. 
        
        To see this, observe that \(B \intersection E(c_1) = \{b_k\}\) and \(\chi(c_1, b_k) = a_{k}\).
        Since for all \(q_{c_1} \in Q_{c_1}\) we have that \(\delta_{\Adversary}^{\TwoP}(q_{c_1}, b_k) \ne \chi(c_1, b_k)\), we have that \(\delta_{\Adversary}^{\TwoP}(q_{c_1}, b_k) = c_{\WindowLength - 1}\).
        That is, each time the token reaches \(b_k\) from \(c_1\), the strategy \(\Strategy[\Adversary]^{\TwoP}\) moves the token to \(c_{\WindowLength -1}\).
        Hence, the token is stuck in the cycle \((b_{k} c_{\WindowLength-1} \cdots c_1)\) where every edge has nonnegative payoff, and no windows open. 
        
        \item Otherwise, suppose \(u = a_j \in A\) for some \(j \in \{1, \ldots, k\}\).\\
        Irrespective of which vertex the game begins from, the strategy~\(\Strategy[\Main]\) eventually moves the token to \(a_j\) encountering at most one open window of length \(\WindowLength\).
        Under the assumption in the contrapositive statement, we have that each time the token reaches \(a_j\), there exists a successor \(b_r\) of \(a_j\) such that \(\Strategy[\Adversary]^{\TwoP}\) does not play according to \(\chi(u, b_r)\) from \(b_r\).
        Specifically, at least one of the following holds:
        \begin{itemize}
            \item If \(\Strategy[\Main]\) moves the token from \(a_j\) to \(b_{j-1}\), then \(\Strategy[\Adversary]^{\TwoP}\) moves the token from \(b_{j-1}\) to~\(c_{\WindowLength - 1}\). 
            \item There exists \(r \ge j\) such that if \(\Strategy[\Main]\) moves the token from \(a_j\) to \(b_{r}\), then \(\Strategy[\Adversary]^{\TwoP}\) moves the token from \(b_r\) to \(a_r\).
        \end{itemize}
        \noindent 
        In particular, if \(j = 1\), then the first statement does not hold since \(b_{0}\) is not defined. Hence, in the case of \(j = 1\), the second statement holds.
        
        To see this, suppose that when the token reaches \(a_j\), the state of the Mealy machine~\(M_{\Adversary}^{\TwoP}\) becomes~\(q'\). 
        When the token moves from \(a_j\) to a successor of \(a_j\), suppose that the state of~\(M_{\Adversary}^{\TwoP}\) becomes~\(q \in Q_{a_{j}}\), i.e., we have that \(\delta_{\Adversary}^{\TwoP}(q', a_j) = q\). 
        Now, there exists \(b_r \in B \intersection \Edges(a_j)\)  such that \(\delta_{\Adversary}^{\TwoP}(q, b_r) \ne \chi(a_j, b_r) \). 
        Recall that \(B \intersection E(a_j) = \{b_{j-1}, b_j, \ldots, b_{k}\}\), and that \(\chi(a_j, b_{j-1}) = a_{j-1}\), and \(\chi(a_j, b_{r}) = c_{\WindowLength -1}\) for all \(r \ge j\). 
        Therefore, we have that at least one of the following holds: 
        \(\delta_{\Adversary}^{\TwoP}(q, b_{j-1}) \ne a_{j-1}\) or there exists \(r \ge j\) such that \(\delta_{\Adversary}^{\TwoP}(q, b_{r}) \ne c_{\WindowLength - 1}\).
        Equivalently, at least one of the following holds: \(\delta_{\Adversary}^{\TwoP}(q, b_{j-1}) = c_{\WindowLength - 1}\) or there exists \(r \ge j\) such that  \(\delta_{\Adversary}^{\TwoP}(q, b_r) = a_r\). 
       
        In the first case, the window does not open at \(a_j\). 
        When the token reaches \(c_{\WindowLength-1}\), the strategy \(\Strategy[\Main]\) eventually moves the token back to \(a_j\) without opening any new windows.
        In the second case, since the edge \((a_j, b_r)\) has payoff \(-1\), a window opens at \(a_j\). 
        However, since the edge \((b_r, a_r)\) has payoff \(+1\), this window closes in the next step and does not remain open for \(\WindowLength\) steps.
        The strategy \(\Strategy[\Main]\) eventually moves the token back to \(a_j\) along vertices in \(A\).
        
        Therefore, each time the token reaches \(a_j\), the strategy \(\Strategy[\Main]\) moves the token to a successor \(b_r\) of \(a_j\) from which \(\Strategy[\Adversary]^{\TwoP}\) does not play according to \(\chi(a_j, b_r)\). 
        This way there are subsequently no open windows of length \(\WindowLength\) in the outcome.
    \end{itemize}
    \noindent 
    This completes the proof of the contrapositive. We now show that such a Mealy machine \(M_{\Adversary}^{\TwoP}\) has at least \(k + 1\) distinct states.
    For all \(u \in A \union \{c_1\}\), let \(q_{u}\) denote a state in \(Q_{u}\) that plays in accordance with the good choices, i.e., for all \(b_r \in B \intersection \Edges(u)\), we have that \(\delta_{\Adversary}^{\TwoP}(q_{u}, b_{r}) = \chi(u, b_r)\). 
    Then, for all \(i, j\) such that \(0 \le i < j < k\), we have that \(q_{a_i}\) and \(q_{a_j}\) are distinct states
    since \(\delta_{\Adversary}^{\TwoP}(q_{a_i}, b_{j-1}) = c_{\WindowLength - 1}\) but \(\delta_{\Adversary}^{\TwoP}(q_{a_j}, b_{j-1}) = a_{j-1}\).
    This gives \(k\) distinct states in \(M_{\Adversary}^{\TwoP}\).
    In addition to this,  
    since \(\delta_{\Adversary}^{\TwoP}(q_{c_1}, b_{k}) = a_k\) but \(\delta_{\Adversary}^{\TwoP}(q_{a_j}, b_{k}) = c_{\WindowLength - 1}\) for all \(j \in \{1, \ldots, k\}\), we have that \(q_{c_1}\) is distinct from each of the \(k\) distinct states found before.
    Thus \(M_{\Adversary}^{\TwoP}\) has at least \(k+1\) distinct states.
\end{proof}

If we allow \(\PlayerAdversary\) to use randomized strategies, then the upper bound on the memory size required for \(\PlayerAdversary\) improves to memoryless strategies. 
\begin{prop}
    A memoryless randomized winning strategy exists for \(\PlayerAdversary\) for the \(\FWMPLBar\) objective.
\end{prop}
\noindent 
A memoryless randomized winning strategy for \(\PlayerAdversary\) for the \(\FWMPLBar\) objective is the following: 
Recall that the winning region of \(\PlayerAdversary\) is a trap for \(\PlayerMain\). 
In each turn, \(\PlayerAdversary\) picks an out-edge uniformly at random out of all out-edges that keep the token in the trap. 
It is always the case that with probability~\(1\), an open window of length \(\WindowLength\) will eventually occur in the play. 
Thus, following this strategy, with probability~\(1\), infinitely many open windows of length \(\WindowLength\) occur in the outcome, resulting in \(\PlayerAdversary\) winning the \(\FWMPLBar\) objective. 

Given a (deterministic) winning strategy \(\Strategy[\Adversary]^{\TwoP}\) of \(\PlayerAdversary\) for the \(\FWMPLBar\) objective, the following lemma gives an upper bound on the number of steps between consecutive open windows of length \(\WindowLength\) in any play consistent with \(\Strategy[\Adversary]^{\TwoP}\).
This lemma is used in 
Section~\ref{sec:reducing_stochastic_window_mean_payoff_games}, where we construct an almost-sure winning strategy of \(\PlayerAdversary\) for the \(\FWMPLBar\) objective.

\begin{lem}%
\label{lem:open_windows_guarantee_general}
    Let \(\Game\) be a non-stochastic game such that all vertices in \(\Game\) are winning for \(\PlayerAdversary\), that is, \(\TwoPlayerWinningRegion[\Game]{\Adversary}{\FWMPLBar} = V\). 
    Let \(\Strategy[\Adversary]^{\TwoP}\) be a finite-memory strategy of \(\PlayerAdversary\) of memory size \(\Memory\) that is winning for \(\FWMPLBar\) from all vertices in \(\Game\).
    Then, for every play \(\Play\) of \(\Game\) consistent with \(\Strategy[\Adversary]^{\TwoP}\), every infix of \(\Play\) of length \(\Memory \cdot \abs{\Vertices} \cdot \WindowLength\) contains an open window of length \(\WindowLength\). 
\end{lem}
\begin{proof}
    Since \(\Strategy[\Adversary]^{\TwoP}\) has memory of size \(\Memory\), fixing this strategy in \(\Game\) gives a one-player game~\(\Game^{\Strategy[\Adversary]^{\TwoP}}\) with \(\Memory \cdot \abs{\Vertices}\) vertices.
    Then, the claim is that every path of length \(\Memory \cdot \abs{\Vertices} \cdot \WindowLength\) in \(\Game^{\Strategy[\Adversary]^{\TwoP}}\) contains an open window of length \(\WindowLength\). 
    Suppose towards a contradiction that there exists a path of length \(\Memory \cdot \abs{\Vertices} \cdot \WindowLength\) in the one-player game that does not contain an open window of length \(\WindowLength\). 
    Since every window is closed in no more than \(\WindowLength\) steps, and the window is closed at the initial vertex to begin with, there are at least \( \Memory \cdot \abs{\Vertices} + 1\) vertices in this path where a window closes. 
    Since there are only \(\Memory \cdot \abs{\Vertices} \) vertices in \(\Game^{\Strategy[\Adversary]^{\TwoP}}\), by the pigeonhole principle, there exists a vertex \(u\) that is visited twice in this path, both times with the window closed. 
    Thus, the path contains a cycle without open windows of length \(\WindowLength\).
    Since \(\PlayerMain\) can reach this cycle and loop in it forever, it gives an outcome that is winning for \(\PlayerMain\) in \(\Game\), which is a contradiction. 
\end{proof}

\section{Reducing stochastic games to non-stochastic games}%
\label{sec:reducing_stochastic_games_to_two_player_games}
In this section, we recall a sufficient condition that allows us to solve stochastic games by solving, as a subroutine, a non-stochastic game with the same objective.
The sufficient condition was presented in~\cite{CHH09} for solving finitary Streett objectives
and can be generalized to arbitrary prefix-independent objectives.
Under this condition, the qualitative problems for stochastic games can be solved as efficiently (up to a factor of \(\abs{\Vertices}^2\)) as non-stochastic games with the same objective.
Also, it follows that the memory requirement for \(\PlayerMain\) to play optimally in  stochastic games is the same as in non-stochastic games with the same objective.

We now describe the sufficient condition that we call the sure-almost-sure property.
Given a stochastic game \(\Game\), let \(\TwoPlayerGame = \tuple{(\Vertices, \Edges), (\VerticesMain, \VerticesAdversary \union \VerticesRandom, \emptyset), \PayoffFunction}\) be the \emph{(adversarial) non-stochastic game corresponding to \(\Game\)}, obtained by changing all probabilistic vertices to \(\PlayerAdversaryDash\) vertices. 
We omit the probability function in the tuple since there are no more probabilistic vertices. 

\begin{defi}[Sure-almost-sure (\(\SAS\))~property] \label{def:prop}  
    A prefix-independent objective \(\Objective\) in a game~\(\Game\) satisfies the \(\SAS\) property
    if \(\TwoPlayerWinningRegion[\TwoPlayerGame]{\Adversary}{\ObjectiveBar} = \Vertices\) implies \(\ASWinningRegion[\Game]{\Adversary}{\ObjectiveBar} = \Vertices\), that is, if \(\PlayerAdversary\) wins the objective \(\ObjectiveBar\) from every vertex in \(\TwoPlayerGame\), then \(\PlayerAdversary\) almost-surely wins the same objective \(\ObjectiveBar\) from every vertex in \(\Game\).
\end{defi}

The definition of the \(\SAS\) property implies that if there exists a vertex from which \(\PlayerMain\) wins the objective \(\Objective\) positively in the stochastic game \(\Game\), then there exists a vertex from which \(\PlayerMain\) wins the same objective \(\Objective\) in the non-stochastic game \(\TwoPlayerGame\).
Note that every prefix-independent objective satisfies the converse of the \(\SAS\)~property since if \(\PlayerAdversary\) wins almost-surely from all vertices in \(\Game\),  then since he controls all probabilistic vertices in \(\TwoPlayerGame\), he wins from all vertices in \(\TwoPlayerGame\) by choosing optimal successors of probabilistic vertices.  

\begin{rem}
We show in Section~\ref{sec:reducing_stochastic_window_mean_payoff_games} that for all stochastic games \(\Game\), the objectives \(\FWMPL[\Game]\) and \(\BWMP[\Game]\) satisfy the \(\SAS\)~property. 
As noted earlier in Section~\ref{sec:window_mean_payoff}, the \(\FWMP_{\Game}(1)\) objective is equivalent to a \textsf{co\Buchi} objective, and thus, \textsf{co\Buchi} satisfies the \(\SAS\) property as well.
One can show that the generalized \textsf{co\Buchi} objective, that is, an objective that is a union of several \textsf{co\Buchi} objectives also satisfies the \(\SAS\) property.
In particular, objectives such as \(\BWMP[\Game]\), and finitary parity and finitary Streett objectives (as defined in~\cite{CHH09}) can be seen as countable unions of \textsf{co\Buchi} objectives, and these objectives satisfy the \(\SAS\) property.

Now, we present an example of objective that \emph{does not} satisfy the \(\SAS\)~property. 
Consider the example in \figurename~\ref{fig:property_counterexample_parity}.
The objective \(\varphi\) in this game is a B\"{u}chi objective: 
a play \(\Play\) satisfies the B\"{u}chi objective if \(\pi\) visits vertex \({\Vertex}_1\) infinitely often.
Although from every vertex, with positive probability (in fact, with probability \(1\)), a play visits \({\Vertex}_1\) infinitely often, from none of the vertices, \(\PlayerMain\) can ensure the B\"{u}chi objective in the non-stochastic game \(\TwoPlayerGame\).
\begin{figure}[t]
    \centering
    \begin{tikzpicture}
        \node[state, accepting] (v1) {\(\Vertex[1]\)};
        \node[random, draw, right of=v1, xshift=+5mm] (v2) {\(\Vertex[2]\)};
        \node[state, right of=v2, xshift=+5mm] (v3) {\(\Vertex[3]\)};
        \draw 
              (v1) edge[below, pos=0.3] (v2)
              (v2) edge[bend right, above right, pos=0.4] node{\(\EdgeProbability{.5}\)} (v1)
              (v2) edge[bend left, above left, pos=0.4] node{\(\EdgeProbability{.5}\)} (v3)
              (v3) edge[below, pos=0.3] (v2)
        ;
    \end{tikzpicture}
    \caption{B\"{u}chi objective does not satisfy the \(\SAS\)~property in this game.}
    \label{fig:property_counterexample_parity}
\end{figure}
\end{rem}

The following theorem states that if an objective \(\Objective\) satisfies the \(\SAS\)~property, then solving the positive (resp., almost-sure) satisfaction problem can be done within a linear (resp., quadratic) factor of the time needed to solve non-stochastic games with the same objective.
\begin{thm}%
\label{thm:twop_to_stochastic}
    Given \(\Game\) and \(\Objective\), suppose in every subgame \(\Game'\) of \(\Game\), the objective \(\Objective\) restricted to~\(\Game'\) satisfies the \(\SAS\)~property. 
    Let \(\TwoPlayerWin_{\Objective}(\TwoPlayerGame)\) be an algorithm computing \(\TwoPlayerWinningRegion[\TwoPlayerGame]{\Main}{\Objective}\) in \(\TwoPlayerGame\) in time~\(\mathbb{C}\).
    Then, the positive and almost-sure satisfaction of \(\Objective\) can be decided in time \(\bigO(\abs{\Vertices} \cdot (\mathbb{C} + \abs{\Edges}))\) and \(\bigO(\abs{\Vertices}^2 \cdot (\mathbb{C} +  \abs{\Edges}))\) respectively.
    
    Moreover, for positive and almost-sure satisfaction of \(\Objective\), the memory requirement for \(\PlayerMain\) to play optimally in stochastic games is no more than that for non-stochastic games.
\end{thm}
\noindent 
Theorem~\ref{thm:twop_to_stochastic} does not give bounds on memory requirement for winning \(\PlayerAdversary\) strategies for objective \(\Objective\) in the stochastic game, but we provide such bounds specifically for \(\FWMPL\) and \(\BWMP\) in Section~\ref{sec:reducing_stochastic_window_mean_payoff_games}.
The proof of the theorem appears shortly after Corollary~\ref{cor:quantitative_satisfaction_complexity}.

Finally, we look at the quantitative decision problem.
The quantitative satisfaction for \(\Objective\) can be decided in \(\NP^{B} \)~(\cite[Theorem~6]{CHH09}), where \(B\) is an oracle deciding positive and almost-sure satisfaction problems for~\(\Objective\).
It is not difficult to see that the quantitative satisfaction for~\(\Objective\) can be decided in \(\NP^{B} \cap \coNP^{B}\).
Moreover, as stated in~\cite[Definition 2]{CHH09}, the vertices of a stochastic game can be partitioned into classes from which  \(\PlayerMain\) wins \(\Objective\) with the same maximal probability.
From~\cite[Lemma 7]{CHH09}, a strategy of \(\PlayerMain\) that is almost-sure winning in every class for the objective \(\Objective \union \ReachObj(Z)\) for some suitable subset \(Z\) of the class is a winning strategy of \(\PlayerMain\) for the quantitative satisfaction of \(\Objective\). 
Analogously, a strategy of \(\PlayerAdversary\) that is positive winning in every class for objective \(\ObjectiveBar \intersection \SafeObj(\overline{Z})\), where \(\overline{Z}\) is the complement of \(Z\) in the class, is a winning strategy for \(\PlayerAdversary\).
Thus, the memory requirement of winning strategies for both players for the quantitative decision problem is no greater than that for the qualitative decision problem.
\begin{cor}%
\label{cor:quantitative_satisfaction_complexity}
    Given \(\Game\) and \(\Objective\) as described in Theorem~\ref{thm:twop_to_stochastic}, 
    let \(B\) be an oracle deciding the qualitative satisfaction of \(\Objective\). 
    Then, the quantitative satisfaction of \(\Objective\) is in \(\NP^{B} \intersection \coNP^{B}\).
    Moreover, the memory requirement of optimal strategies for both players is no greater than that for the positive and almost-sure satisfaction of \(\Objective\). 
\end{cor}
\noindent 
Now, we describe an algorithm \(\PosWin_{\Objective}\) to compute the positive winning region of \(\PlayerMain\) in \(\Game\) with objective \(\Objective\).
The algorithm uses \(\TwoPlayerWin_{\Objective}\) as a subroutine.
Then, 
we describe an algorithm \(\ASWin_{\Objective}\) that uses \(\PosWin_{\Objective}\) as a subroutine to compute the almost-sure winning region for \(\PlayerMain\) for the objective~\(\Objective\).
The algorithms and their correctness proof are the same as in the case of finitary Streett objectives described in~\cite{CHH09}.

\begin{proof}[Proof of Theorem~\ref{thm:twop_to_stochastic}]
    We recall the recursive
    procedures (in Algorithm~\ref{alg:poswin} and Algorithm~\ref{alg:aswin}) to compute the positive and the almost-sure winning regions for \(\PlayerMain\) in stochastic games with an objective that satisfies the \(\SAS\)~property. 
    The algorithms are similar to the case of finitary Streett objectives~\cite{CHH09}, which satisfy the \(\SAS\)~property.
    Note that, because of determinacy, the positive winning region \(\PosWinningRegion[\Game]{\Main}{\Objective}\)
    for \(\PlayerMain\) is the complement of the almost-sure winning region  \(\ASWinningRegion[\Game]{\Adversary}{\ObjectiveBar}\) for \(\PlayerAdversary\).
    
    \begin{algorithm}[t]
        \caption{\(\PosWin_{\Objective}(\Game)\)}\label{alg:poswin}
        \begin{algorithmic}[1]
            \Require{\(\Game = \tuple{(\Vertices, \Edges), (\VerticesMain, \VerticesAdversary, \VerticesRandom), \ProbabilityFunction, \PayoffFunction}\), the stochastic game}
            \Ensure{The set of vertices from which \(\PlayerMain\) positively wins objective~\(\Objective\) in~\(\Game\)}
            \State \(\WinningSet{\Main} \gets \TwoPlayerWin_{\Objective}(\TwoPlayerGame)\) \label{alg_line:poswin_2P}
            \If{\(\WinningSet{\Main} = \emptyset\)} 
                \State \Return \(\emptyset\)\label{alg_line:poswin_return_empty}
            \Else
                \State \(\PosAttractorSet{\Main} \gets \PosAttr{\Main}{\WinningSet{\Main}}\)\label{alg_line:poswin_attractor_computation}
                \State \Return \(\PosAttractorSet{\Main} \union  \PosWin_{\Objective}(\Game \restriction (\Vertices \setminus \PosAttractorSet{\Main}))\) \label{alg_line:poswin_recursive_call}
            \EndIf
        \end{algorithmic}
    \end{algorithm}
    
    \begin{algorithm}[t]
        \caption{\(\ASWin_{\Objective}(\Game)\)}\label{alg:aswin}
        \begin{algorithmic}[1]
            \Require{\(\Game = \tuple{(\Vertices, \Edges), (\VerticesMain, \VerticesAdversary, \VerticesRandom),  \ProbabilityFunction, \PayoffFunction}\), the stochastic game}
            \Ensure{The set of vertices in \(\Vertices\) from which \(\PlayerMain\) almost-surely wins \(\Objective\) in~\(\Game\)}
            \State \(W_\Adversary \gets \Vertices \setminus \PosWin_{\Objective}(\Game)\)\label{alg_line:aswin_winning_set_computation}
            \If{\(\WinningSet{\Adversary} = \emptyset\)}
                \State \Return \(\Vertices\)\label{alg_line:aswin_return_all_vertices}
            \Else
                \State \(\PosAttractorSet{\Adversary} \gets \PosAttr{\Adversary}{\WinningSet{\Adversary}}\)\label{alg_line:aswin_attractor_computation}
                \State \Return \(\ASWin_{\Objective}(\Game \restriction (\Vertices \setminus \PosAttractorSet{\Adversary}))\)\label{alg_line:aswin_recursive_call}
            \EndIf
        \end{algorithmic}
    \end{algorithm}
    
    The depth of recursive calls in Algorithm~\ref{alg:poswin}
    is bounded by \(\abs{\Vertices}\), the number of vertices in \(\Game\), as the argument in the recursive call (Line~\ref{alg_line:poswin_recursive_call}) has strictly 
    fewer vertices than \(\abs{\Vertices}\), since $\PosAttractorSet{\Main} \neq \emptyset$. 
    The \(\PlayerMainDash\)~positive attractor is computed in time \(\bigO( \abs{\Edges})\), and suppose \(\TwoPlayerWin_{\Objective}\) runs in time \(\mathbb{C}\). 
    These subroutines are executed at most \(\abs{\Vertices}\) times, once in every depth of the recursive call. 
    Thus, the total running time of \(\PosWin_{\Objective}\) is \(\bigO(\abs{\Vertices} \cdot (\mathbb{C} + \abs{\Edges}))\).
    
    Let \(\WinningSet[i]{\Main}\) and \(\PosAttractorSet[i]{\Main}\) denote the sets \(\WinningSet{\Main}\) and \(\PosAttractorSet{\Main}\) computed in the recursive call of depth~\(i\) respectively. 
    Recall that the sets \(\PosAttractorSet[i]{\Main}\) form a partition of the positive winning region \(\PosWinningRegion[\Game]{\Main}{\Objective}\) for \(\PlayerMain\), and that for all \(i\), we have that \(\WinningSet[i]{\Main} \subseteq \PosAttractorSet[i]{\Main}\).
    Let \(\Strategy[\Main]^{\TwoP}\) be a winning strategy of \(\PlayerMain\) in the non-stochastic game~\(\TwoPlayerGame\). 
    We construct a positive-winning strategy \(\Strategy[\Main]^{\Pos}\) for \(\PlayerMain\) in the stochastic game as follows.
    Given a prefix \(\Prefix \in \PrefixSet[\Main]{\Game}\), we determine the value of \(i\) for which \(\Last{\Prefix} \in \PosAttractorSet[i]{\Main}\).
    Then, if \(\Last{\Prefix} \in \WinningSet[i]{\Main}\), then \(\Strategy[\Main]^{\Pos}\) plays like \(\Strategy[\Main]^{\TwoP}\), that is \(\Strategy[\Main]^{\Pos}(\Prefix) = \Strategy[\Main]^{\TwoP}(\Prefix)\);  
    otherwise, \(\Last{\Prefix} \in \PosAttractorSet[i]{\Main} \setminus \WinningSet[i]{\Main}\), and let \(\Strategy[\Main]^{\Pos}(\Prefix) = \Strategy[\Main]^{\mathrm{Attr}}(\Prefix)\) 
    where \(\Strategy[\Main]^{\mathrm{Attr}}\)  is a positive-attractor strategy to \(\WinningSet[i]{\Main}\) (which is memoryless, i.e., \(\Strategy[\Main]^{\mathrm{Attr}}(\Prefix) =\Strategy[\Main]^{\mathrm{Attr}}(\Last{\Prefix}) \)).
    Then, \(\Strategy[\Main]^{\Pos}\) is a positive-winning strategy for \(\PlayerMain\) from all vertices in \(\PosWin_{\Objective}(\Game)\). 
    
    The depth of recursive calls in Algorithm~\ref{alg:aswin} is also bounded by \(\abs{\Vertices}\).
    The set \(\WinningSet{\Adversary}\) from which \(\PlayerAdversary\) wins almost-surely for objective \(\ObjectiveBar\) is computed in time \(\bigO(\abs{\Vertices} \cdot (\mathbb{C} + \abs{\Edges}))\), and the \(\PlayerAdversaryDash\)~positive attractor is computed in time \(\bigO( \abs{\Edges})\). 
    This leads to a total running time \(\bigO(\abs{\Vertices}^2 \cdot (\mathbb{C} +  \abs{\Edges}))\).
\end{proof}

The following lemma, which is a special case of Theorem~1 in~\cite{Cha07} where it has been proved for concurrent stochastic games, allows us to use results from the computation of the positive winning region for \(\PlayerMain\) in \(\Game\) to obtain the almost-sure winning region for \(\PlayerMain\) in \(\Game\). 
In Theorem~1 in~\cite{Cha07}, it has been shown that for a prefix-independent objective, if there exists a vertex from which \(\PlayerMain\) wins positively, then there exists a vertex from which \(\PlayerMain\) wins almost-surely.
Since prefix-independent objectives are closed under complementation, the theorem also holds for \(\PlayerAdversary\).
Considering the theorem for \(\PlayerAdversary\), and taking the contrapositive, we have the following lemma for the special case of turn-based zero-sum stochastic games.
In particular, in Algorithm~\ref{alg:aswin}, since \(\WinningSet{\Adversary} = \emptyset\) denotes that all vertices are positively winning for \(\PlayerMain\), this gives that all vertices are almost-surely winning for \(\PlayerMain\) by Lemma~\ref{lem:all_positively_implies_all_almost_surely}.

\begin{lem}%
\label{lem:all_positively_implies_all_almost_surely}\cite[Theorem~1]{Cha07}
    If \(\PlayerMain\) positively wins a stochastic game \(\Game\) with a prefix-independent objective \(\Objective\) from every vertex in \(\Vertices\), then \(\PlayerMain\) almost~surely wins \(\Game\) with objective \(\Objective\) from every vertex in \(\Vertices\), that is, if \(\PosWinningRegion[\Game]{\Main}{\Objective} = \Vertices\), 
    then \(\ASWinningRegion[\Game]{\Main}{\Objective} = \Vertices\). 
\end{lem}
\noindent 
Let \(\Strategy[\Main]^{\Pos}\) be the strategy of \(\PlayerMain\) as described in the \(\PosWin_{\Objective}\) algorithm.
An almost-sure winning strategy \(\Strategy[\Main]^{\AS}\) of \(\PlayerMain\) for objective \(\Objective\) is the same as \(\Strategy[\Main]^{\Pos}\); if \(\PlayerMain\) follows the strategy \(\Strategy[\Main]^{\Pos}\), then she almost-surely satisfies \(\Objective\) from all vertices in \(\ASWin_{\Objective}(\Game)\).
For both positive and almost-sure winning, \(\PlayerMain\) does not require any additional memory in the stochastic game compared to the non-stochastic game.\footnote{
    If deterministic strategies suffice for \(\PlayerMain\) in non-stochastic games to win an objective \(\Objective\) satisfying the \(\SAS\) property, then deterministic strategies also suffice for \(\PlayerMain\) for the positive and almost-sure winning strategies of the same objective \(\Objective\) in stochastic games.
}

\section{Reducing stochastic window mean-payoff games: A special case}
\label{sec:reducing_stochastic_window_mean_payoff_games}

In this section, we show that for all stochastic games \(\Game\) and for all \(\WindowLength \ge 1\), the  fixed window mean-payoff objective \(\FWMPL[\Game]\) and the bounded window mean-payoff objective \(\BWMP[\Game]\), which are prefix independent objectives, satisfy the \(\SAS\)~property of Definition~\ref{def:prop}.
Thus, by Theorem~\ref{thm:twop_to_stochastic}, we obtain bounds on the complexity and memory requirements of \(\PlayerMain\) for positive satisfaction and almost-sure satisfaction of these objectives.
The algorithms to compute the positive and the almost-sure winning regions of \(\PlayerMain\) for \(\FWMPL\)  (resp., \(\BWMP\)) objective can be obtained by instantiating Algorithms~\ref{alg:poswin} and~\ref{alg:aswin} respectively with \(\Objective\) equal to \(\FWMPL\) (resp., \(\BWMP\)).
We also show that for both these objectives, the memory requirements of \(\PlayerAdversary\) to play optimally for positive and almost-sure winning in stochastic games is no more than that of the non-stochastic games.

\subsection{Fixed window mean-payoff objective}%
\label{sub:fwmp-sas}

We show that the \(\SAS\)~property holds for the objective \(\FWMPL\) for all stochastic games \(\Game\) and for all \(\WindowLength \ge 1\).

\begin{lem}%
\label{lem:fwmp-sas}
    For all stochastic games \(\Game\) and for all \(\WindowLength \ge 1\), the objective \(\FWMPL\) satisfies the \(\SAS\)~property.
\end{lem}
\begin{proof}
    We need to show that 
    if \(\TwoPlayerWinningRegion[\TwoPlayerGame]{\Adversary}{\FWMPLBar} = \Vertices\), then  \(\ASWinningRegion[\Game]{\Adversary}{\FWMPLBar} = \Vertices\). 
    
    If \(\TwoPlayerWinningRegion[\TwoPlayerGame]{\Adversary}{\FWMPLBar} = \Vertices\), 
    then from Theorem~\ref{thm:memoryPlayer2}, there exists a finite-memory strategy \(\Strategy[\Adversary]^{\TwoP}\) (say, with memory \(\Memory\)) of \(\PlayerAdversary\) that is winning for objective \(\FWMPLBar\) from every vertex in \(\TwoPlayerGame\).
    Given such a strategy, we construct below a strategy \(\Strategy[\Adversary]^{\AS}\) of \(\PlayerAdversary\) in the stochastic game \(\Game\) that is almost-sure winning for \(\FWMPLBar\) from every vertex in \(\Game\).
    
    In \(\TwoPlayerGame\), \(\PlayerAdversary\) chooses the successor of vertices in \(\VerticesAdversary \union \VerticesRandom\) according to the strategy~\(\Strategy[\Adversary]^{\TwoP}\). 
    Since \(\Strategy[\Adversary]^{\TwoP}\) is a winning strategy, \(\PlayerAdversary\) can satisfy the \(\FWMPLBar\) objective irrespective of \(\PlayerMain\)'s strategy.
    In the stochastic game \(\Game\), however, \(\PlayerAdversary\) has less control.
    He can only choose successors for vertices in~\(\VerticesAdversary\), while the successors for vertices in~\(\VerticesRandom\) are chosen according to the probability function~\(\ProbabilityFunction\) that is specified in the game. 
    It is possible that for a probabilistic vertex, the successor chosen by the distribution is not what \(\PlayerAdversary\) would have chosen, resulting in a potentially worse outcome for him.
    We use the fact that \(\FWMPLBar\) is a \Buchi-like objective (\(\PlayerAdversary\) would like to \emph{always eventually} see an open window  of length~\(\WindowLength\)), and that \(\Strategy[\Adversary]^{\TwoP}\) is winning from every vertex to show that despite having control over fewer vertices, \(\PlayerAdversary\) has a strategy \(\Strategy[\Adversary]^{\AS}\) that is almost-sure winning for the \(\FWMPLBar\) objective from every vertex in the stochastic game~\(\Game\).
    
    Let \(\Play = v_0 v_1 \cdots\) be an outcome in the stochastic game \(\Game\) when \(\PlayerAdversary\) follows the strategy \(\Strategy[\Adversary]^{\TwoP}\), i.e., for all \(v_i\) in \(\Play\) such that \(v_i \in \VerticesAdversary\), we have that \(v_{i+1} = \Strategy[\Adversary]^{\TwoP}(v_0v_1\cdots v_i)\). 
    For all probabilistic vertices \(v_j \in \VerticesRandom\) in the play \(\Play\), if the successor vertex of \(v_j\) chosen by the probability distribution is not equal to the successor vertex \(\Strategy[\Adversary]^{\TwoP}(v_0 v_1 \cdots v_j)\) of \(v_j\) given by the strategy \(\Strategy[\Adversary]^{\TwoP}\), 
    i.e., if \(v_{j+1} \ne \Strategy[\Adversary]^{\TwoP}(v_0 v_1 \cdots v_j)\), then we say that a \emph{deviation} from the strategy \(\Strategy[\Adversary]^{\TwoP}\) occurs in~\(\Play\) at \(v_j\).
    Note that the prefix \(v_0v_1\cdots v_j v_{j+1}\) with the deviation does never appears in any play in \(\TwoPlayerGame\) that is consistent with \(\Strategy[\Adversary]^{\TwoP}\).
    
    Some deviations may cause the outcome to be losing for \(\PlayerAdversary\). 
    Therefore, starting with the strategy \(\Strategy[\Adversary]^{\TwoP}\), we construct a strategy \(\Strategy[\Adversary]^{\AS}\) that mimics \(\Strategy[\Adversary]^{\TwoP}\) as long as no such deviations occur, and \emph{resets} otherwise, i.e., the strategy forgets the prefix of the play before the deviation. 
    We call the strategy \(\Strategy[\Adversary]^{\AS}\) a \emph{reset strategy}.
    We see in Construction~\ref{con:reset-strategy-mealy-machine-adversary}, given a Mealy machine \(M_\Adversary^{\TwoP}\) that defines \(\Strategy[\Adversary]^{\TwoP}\), how to construct a Mealy machine \(M_{\Adversary}^{\AS}\) that defines the reset strategy~\(\Strategy[\Adversary]^{\AS}\).
    In the construction, we show that the memory size of \(\Strategy[\Adversary]^{\AS}\) is no more than that of \(\Strategy[\Adversary]^{\TwoP}\). 
    Therefore, all games with objective \(\FWMPL\) satisfy the \(\SAS\) property.
\end{proof}

In Example~\ref{exa:reset-strategy-example}, we see an example of a stochastic game \(\Game\) along with a strategy \(\Strategy[\Adversary]^{\TwoP}\) that is winning for \(\PlayerAdversary\) from all vertices in the adversarial game \(\TwoPlayerGame\).
We show that this strategy \(\Strategy[\Adversary]^{\TwoP}\) need not be almost-sure winning for \(\PlayerAdversary\) in the stochastic game \(\Game\) and then give an intuition on how to use \(\Strategy[\Adversary]^{\TwoP}\) to construct a reset strategy \(\Strategy[\Adversary]^{\AS}\) that is almost-sure winning from all vertices in \(\Game\).
In Construction~\ref{con:reset-strategy-mealy-machine-adversary}, we formally show how to obtain a Mealy machine that defines \(\Strategy[\Adversary]^{\AS}\) from a Mealy machine that defines \(\Strategy[\Adversary]^{\TwoP}\) without adding any new states. 

\begin{exa}%
\label{exa:reset-strategy-example}
    \figurename~\ref{fig:reset-strategy-running-example-game} shows a stochastic game \(\Game\) with objective \(\overline{\FWMP(3)}\) for \(\PlayerAdversary\).
    The edges \((v_2, v_4)\) and \((v_3, v_5)\) have negative payoffs and all other edges have zero payoff.
    For each probabilistic vertex \(v \in \VerticesRandom\), we have that the probability function \(\ProbabilityFunction(v)\) is a uniform distribution, 
    i.e., we have:
    \(\ProbabilityFunction(v_2)(v_4) = \ProbabilityFunction(v_2)(v_5) = \frac{1}{2}\), 
    \(\ProbabilityFunction(v_3)(v_4) = \ProbabilityFunction(v_3)(v_5) = \frac{1}{2}\), 
    and 
    \(\ProbabilityFunction(v_6)(v_7) =  \ProbabilityFunction(v_6)(v_8) = \frac{1}{2}\). 
    
    \figurename~\ref{fig:reset-strategy-example-mealy-machine-input} shows a Mealy machine \(M_{\Adversary}^{\TwoP}\) defining a strategy \(\Strategy[\Adversary]^{\TwoP}\) that is winning for \(\overline{\FWMP(3)}\) from all vertices in the adversarial game \(\TwoPlayerGame\). 
    In figures, for states \(q_i\), \(q_j\) of the Mealy machine and for vertices \(v\), \(v'\) of the game, an edge from state \(q_i\) to \(q_j\) with label \(v / v'\) denotes that the next state of the Mealy machine is \(\Delta(q_i, v) = q_j\) and the next vertex is \(\delta(q_i, v) = v'\).
    To see that \(\Strategy[\Adversary]^{\TwoP}\) is a winning strategy, note that each time the token reaches \(v_1\), we have that \(\PlayerMain\) may move the token from \(v_1\) to either \(v_2\) or \(v_3\). 
    The strategy \(\Strategy[\Adversary]^{\TwoP}\) moves the token 
    from \(v_2\) and \(v_3\) to \(v_4\) and \(v_5\) respectively, ensuring that a window opens.
    Then, when the token reaches \(v_6\), the strategy always moves the token to \(v_8\) and then back to \(v_1\).
    The Mealy machine never moves the token to \(v_7\).
    If the game begins in \(v_7\), then the token is moved to \(v_8\) and then to \(v_1\) and the token never goes to \(v_7\) after that.
    Each time the token reaches \(v_1\), the token must move to \(v_2\) or \(v_3\). 
    Since there are no edges with positive payoff, the window starting at \(v_2\) or \(v_3\) never closes, and in particular, remains open for \(3\) steps. 
    Since the token reaches \(v_1\) infinitely often, by following this strategy, \(\PlayerAdversary\) ensures for all strategies of \(\PlayerMain\), the outcome contains infinitely many open windows of length \(3\) and thus, the strategy \(\Strategy[\Adversary]^{\TwoP}\) is winning for \(\PlayerAdversary\) for \(\overline{\FWMP(3)}\) in \(\TwoPlayerGame\). 
    
    \begin{figure}[t]
        \centering
        \scalebox{0.9}{
        \begin{tikzpicture}
            \node[state] (v1) {\(v_1\)};
            \node[random, draw, above right of=v1, yshift=-.4cm] (v2) {\(v_2\)};
            \node[random, draw, below right of=v1, yshift=.4cm] (v3) {\(v_3\)};
            \node[square, draw, right of=v2] (v4) {\(v_4\)};
            \node[square, draw, right of=v3] (v5) {\(v_5\)};
            \node[random, draw, below right of=v4, yshift=.3cm] (v6) {\(v_6\)};
            \node[square, draw, right of=v6] (v7) {\(v_7\)};
            \node[square, draw, below of=v6] (v8) {\(v_8\)};
            \draw 
                  (v1) edge[below right] node[yshift=.1cm]{} (v2)
                  (v1) edge[above right] node[yshift=-.1cm]{} (v3)
                  (v2) edge[above] node{\(-1\)} (v4)
                  (v2) edge[below] node[below right, yshift=-.2cm]{} (v5)
                  (v3) edge[below] node[left]{} (v4)
                  (v3) edge[below] node{\(-1\)} (v5)
                  (v4) edge[below] node[above]{} (v6)
                  (v5) edge[below] node{} (v6)
                  (v6) edge[below] node{} (v7)
                  (v6) edge[right] node{} (v8)
                  (v7) edge[below] node{} (v8)
                  (v8) edge[below, bend left=40] node{} (v1)
                  (v8) edge[loop right] node{} (v8)
                ;
        \end{tikzpicture}
        }
        \caption{The game \(\Game\) with objective \(\overline{\FWMP(3)}\) for \(\PlayerAdversary\) from Example~\ref{exa:reset-strategy-example}. All edges except \((v_2, v_4)\) and \((v_3, v_5)\) have payoff \(0\).
        For all probabilistic vertices \(v \in \VerticesRandom\), the probability function \(\ProbabilityFunction(v)\) is a uniform distribution over the out-neighbours \(E(v)\) of \(v\).
        }
        \label{fig:reset-strategy-running-example-game}
    \end{figure}
    \begin{figure}[t]
        \centering
        \scalebox{0.8}{
        \begin{tikzpicture}
            \node[state, initial] (q0) {\(q_0\)};
            \node[state, right of=q0] (q1) {\(q_1\)};
            \node[state, right of=q1, xshift=+15mm] (q2) {\(q_2\)};
            \node[state, right of=q2, xshift=+15mm] (q3) {\(q_3\)};
            \node[state, right of=q3] (q4) {\(q_4\)};
            \node[state, right of=q4] (q5) {\(q_5\)};
            \draw 
                  (q0) edge[above] node{\(v_1 / \epsilon\)} (q1)
                  (q0) edge[loop above] node{\(v_8/v_1\)} (q0)
                  (q0) edge[above, bend right=29, align=center, pos=0.55] node {\(v_2/v_4\), \(v_3/v_5\)} (q2)
                  (q0) edge[above, bend right=33, align=center, pos=0.55] node {\(v_4/v_6\), \(v_5/v_6\)} (q3)
                  (q0) edge[above, bend right=41, align=center, pos=0.55] node {\(v_6/v_8\), \(v_7/v_8\)} (q4)
                  (q1) edge[above, align=center] node {\(v_2/v_4\), \(v_3/v_5\)} (q2)
                  (q2) edge[above, align=center] node {\(v_4/v_6\), \(v_5/v_6\)} (q3)
                  (q3) edge[above] node{\(v_6/ v_8\)} (q4)
                  (q4) edge[above] node{\(v_7/v_8\)} (q5)
                  (q4) edge[below, bend right=37, pos=0.5] node{\(v_8/v_1\)} (q0)
                  (q5) edge[loop right] node{\(v_8/v_8\)} (q5)
                ;
        \end{tikzpicture}
        }
        \caption{Mealy machine \(M_{\Adversary}^{\TwoP}\) defining a strategy \(\Strategy[\Adversary]^{\TwoP}\) that is winning from all vertices in \(\TwoPlayerGame\) for \(\overline{\FWMP(3)}\).}
        \label{fig:reset-strategy-example-mealy-machine-input}
    \end{figure}
    
    Observe that the strategy \(\Strategy[\Adversary]^{\TwoP}\) is not almost-sure winning for \(\PlayerAdversary\) from any vertex in the stochastic game \(\Game\). 
    If \(\PlayerAdversary\) follows the strategy \(\Strategy[\Adversary]^{\TwoP}\) in \(\Game\), then the probability that he wins \(\overline{\FWMP(3)}\) is less than 1. 
    This is because when the token reaches \(v_6\), then with probability \(\frac{1}{2}\), it moves to \(v_7\). Once that happens, the state of the Mealy machine changes to~\(q_5\) and the strategy moves the token to \(v_8\) and keeps it there forever.
    No new windows open, and the outcome is losing for \(\overline{\FWMP(3)}\).
    Hence, if \(\PlayerAdversary\) follows the strategy \(\Strategy[\Adversary]^{\TwoP}\) in~\(\Game\), then with positive probability, he does not win \(\overline{\FWMP(3)}\). 
    
    This is not an issue in the non-stochastic game \(\TwoPlayerGame\) since the \(q_4 \xrightarrow{v_7} q_5\) transition is never taken in \(\TwoPlayerGame\) as long as \(\PlayerAdversary\) plays according to \(\Strategy[\Adversary]^{\TwoP}\).
    If \(\PlayerAdversary\) plays according to \(\Strategy[\Adversary]^{\TwoP}\), then in any transition that changes the state of the Mealy machine to \(q_4\), the token moves to~\(v_8\) by that transition. 
    There does not exist any prefix \(\Prefix\) in \(\TwoPlayerGame\) that is consistent with \(\Strategy[\Adversary]^{\TwoP}\) that causes the Mealy machine to take the \(q_4 \xrightarrow{v_7} q_5\) transition. 
    We call such transitions that cannot be taken in \(M_{\Adversary}^{\TwoP}\) in \(\TwoPlayerGame\) when \(\PlayerAdversary\) plays according to \(\Strategy[\Adversary]^{\TwoP}\) \emph{unreachable}. 
    We see that there are no reachable transitions that change the state of \(M_{\Adversary}^{\TwoP}\) to \(q_5\), and thus, the outgoing transition \(q_5 \xrightarrow{v_8} q_5\) from \(q_5\) is also unreachable.
    One can verify that all transitions in \(M_{\Adversary}^{\TwoP}\) other than the two mentioned above are reachable in \(\TwoPlayerGame\).
    
    If \(\PlayerAdversary\) follows the strategy \(\Strategy[\Adversary]^{\TwoP}\) in the stochastic game \(\Game\), then a deviation may occur at a vertex in \(\VerticesRandom\) and the Mealy machine \(M_{\Adversary}^{\TwoP}\) may follow a transition that is otherwise unreachable in \(\TwoPlayerGame\).
    If an unreachable transition is taken, then we cannot guarantee that the output of the Mealy machine will result in a play that is winning for \(\PlayerAdversary\).
    For example, when the token moves from \(v_6\) to \(v_7\), the Mealy machine takes the unreachable \(q_4 \xrightarrow{v_7} q_5\) transition, which as we saw above results in an outcome that is not winning for the \(\overline{\FWMP(3)}\) objective.
    
    Since we do not know how the Mealy machine \(M_{\Adversary}^{\TwoP}\) behaves on taking unreachable transitions, we design the Mealy machine \(M_{\Adversary}^{\AS}\) that defines the almost-sure winning strategy in \(\Game\)
    to reset instead of taking an unreachable transition. 
    For instance, we define the transition in \(M_{\Adversary}^{\AS}\) from state~\(q_4\) on reading vertex \(v_7\) to be as if the game started from \(v_7\).
    The Mealy machine \(M_{\Adversary}^{\TwoP}\) on reading \(v_7\) from the initial state, outputs \(v_8\) and changes its state to \(q_4\). 
    Therefore, we want the same behaviour in \(M_{\Adversary}^{\AS}\) from~\(q_4\), i.e.,  on reading vertex~\(v_7\) from~\(v_4\), the Mealy machine outputs~\(v_8\) and update its state to~\(q_4\). 
    We add the necessary reset transitions in this manner for all states. 
    For every state~\(q\), for every vertex \(v\), if there is a reachable transition \(q \xrightarrow{v} q'\) from~\(q\) on reading~\(v\), then we retain the same transition in \(M_{\Adversary}^{\AS}\), i.e., with the same output and same next state.
    If there is no reachable transition from~\(q\) on reading~\(v\), then we add a reset transition. 
    We go to the same state and output the same vertex that would be output from~\(q_0\) on reading~\(v\). 
    (For all vertices in \(\VerticesRandom\), we change the output of the Mealy machine to \(\epsilon\) since \(\PlayerAdversary\) does not control these vertices in the stochastic game \(\Game\).)
    This gives us a complete Mealy machine. 
    
    \begin{figure}[t]
        \centering
        \scalebox{0.8}{
        \begin{tikzpicture}
            \node[state, initial] (q0) {\(q_0\)};
            \node[state, right of=q0] (q1) {\(q_1\)};
            \node[state, right of=q1, xshift=+15mm] (q2) {\(q_2\)};
            \node[state, right of=q2, xshift=+15mm] (q3) {\(q_3\)};
            \node[state, right of=q3, xshift=+15mm] (q4) {\(q_4\)};
            \draw 
                  (q0) edge[above] node{\(v_1 / \epsilon\)} (q1)
                  (q0) edge[loop above] node{\(v_8/v_1\)} (q0)
                  (q0) edge[above, bend right=27, align=center, pos=0.6] node {\(v_2/\epsilon\), \(v_3/\epsilon\)} (q2)
                  (q0) edge[above, bend right=30, align=center, pos=0.6] node {\(v_4/v_6\), \(v_5/v_6\)} (q3)
                  (q0) edge[above, bend right=36, align=center, pos=0.6] node {\(v_6/\epsilon\), \(v_7/v_8\)} (q4)
                  (q1) edge[above, align=center] node {\(v_2/\epsilon\), \(v_3/\epsilon\)} (q2)
                  (q2) edge[above, align=center] node {\(v_4/v_6\), \(v_5/v_6\)} (q3)
                  (q3) edge[below] node{\(v_6/ \epsilon\)} (q4)
                  (q4) edge[below, bend right=42, pos=0.75] node{\(v_8/v_1\)} (q0)
                  (q4) edge[below, bend right=41, pos=0.7] node{\(v_1/\epsilon\)} (q1)
                  (q4) edge[below, bend right=41] node{\(v_2/\epsilon, v_3/\epsilon\)} (q2)
                  (q4) edge[below, bend right=45, pos=0.5] node[yshift=-0.3mm]{\(v_4/v_6\), \(v_5/v_6\)} (q3)
                  (q4) edge[loop below] node{\(v_6/\epsilon\), \(v_7/v_8\)} (q4)
                ;
        \end{tikzpicture}
        }
        \caption{Part of the Mealy machine \(M_{\Adversary}^{\AS}\) defining a reset strategy that is almost-sure winning from all vertices in \(\Game\).
        Reset transitions out of \(q_1\), \(q_2\), and \(q_3\) have been omitted from the figure.}
        \label{fig:reset-strategy-example-mealy-machine-output}
    \end{figure}
    
    Finally, after defining the new transitions, we see that there are no transitions that lead to~\(q_5\). 
    It is an unreachable state and we delete it. 
    \figurename~\ref{fig:reset-strategy-example-mealy-machine-output} shows some of the transitions in the Mealy machine \(M_{\Adversary}^{\AS}\) obtained after the resetting. 
    The figure excludes all unreachable transitions, and shows reset transitions out of \(q_4\). 
    There exist reset transitions out of states \(q_1\), \(q_2\) and \(q_3\) as well, but we omit them in the figure for the sake of clarity. 
    \qed
\end{exa}

Now, we formally state a procedure to construct an almost-sure winning strategy in \(\Game\) from a given winning strategy in \(\TwoPlayerGame\) and show the correctness of this procedure.

\begin{construction}[Reset strategy]%
\label{con:reset-strategy-mealy-machine-adversary}
    Let \(\Strategy[\Adversary]^{\TwoP}\) be a strategy of \(\PlayerAdversary\) that is winning for \(\FWMPLBar\) from every vertex in \(\TwoPlayerGame\), and let \(M_{\Adversary}^{\TwoP} = (Q_\Adversary^{\TwoP}, q_0,  \Vertices, \Vertices \union \{\epsilon\}, \Delta_{\Adversary}^{\TwoP}, \delta_{\Adversary}^{\TwoP})\) be a Mealy machine that defines the strategy \(\Strategy[\Adversary]^{\TwoP}\), 
    where the set of states is \(Q_\Adversary^{\TwoP}\),
    the initial state is \(q_0\),
    the input alphabet is \(\Vertices\), 
    the output alphabet is \(\Vertices \union \{\epsilon\}\),
    the transition function is \(\Delta_\Adversary^{\TwoP} \colon Q_\Adversary^{\TwoP} \times \Vertices \to Q_\Adversary^{\TwoP}\), and the output function is \(\delta_\Adversary^{\TwoP} \colon Q_\Adversary^{\TwoP} \times \Vertices \to \Vertices \union \{\epsilon\}\). 
    Since the strategy \(\Strategy[\Adversary]^{\TwoP}\) is winning for \(\FWMPLBar\) in \(\TwoPlayerGame\) irrespective of the initial vertex of the game,  the transition and output functions are defined from the initial state of the Mealy machine \(M_{\Adversary}^{\TwoP}\) for all vertices in the game, that is, \(\Delta_{\Adversary}^{\TwoP}(q_0, v)\)  and \(\delta_{\Adversary}^{\TwoP}(q_0, v)\) are defined for all \(v \in \Vertices\). 
    
    From this, we give a construction of a Mealy machine \(M_\Adversary^{\AS} = (Q_\Adversary, q_0, \Vertices, \Vertices \union \{\epsilon\}, \Delta_\Adversary, \delta_\Adversary)\) with \(Q_{\Adversary} \subseteq Q_{\Adversary}^{\TwoP}\) that defines the reset strategy \(\Strategy[\Adversary]^{\AS}\) for \(\PlayerAdversary\) in the stochastic game \(\Game\). 
    We begin with all the states of \(M_{\Adversary}^{\TwoP}\) and then over the course of the construction, delete some states that are not needed.
    The initial state of \(M_{\Adversary}^{\AS}\) is \(q_0\), the same as the initial state of \(M_{\Adversary}^{\TwoP}\).
    The input and output alphabets of \(M_{\Adversary}^{\AS}\) are the same as that of \(M_{\Adversary}^{\TwoP}\). 
    It remains to define the transition function \(\Delta_{\Adversary}^{\TwoP}\) and the update function \(\delta_{\Adversary}^{\TwoP}\). 
    
    We begin by computing all the transitions in \(M_{\Adversary}^{\TwoP}\) that are reachable from the initial state~\(q_0\).
    For all states \(q_1, q_2 \in Q_{\Adversary}^{\TwoP}\) and vertices \(v \in \Vertices\),  the transition  \(q_1 \xrightarrow{v} q_2 \) is \emph{reachable} in~\(M_{\Adversary}^{\TwoP}\) from~\(q_0\) if there exists a prefix \(\Prefix \cdot v\) in \(\TwoPlayerGame\) consistent with \(\Strategy[\Adversary]^{\TwoP}\) such that \(\hat{\Delta}_{\Adversary}^{\TwoP}(q_0, \Prefix) = q_1\) and \(\Delta_{\Adversary}^{\TwoP}(q_1, v) = q_2\).
    
    We have that for all \(v \in \Vertices\), the transition 
    \(q_0 \xrightarrow{v} \Delta_{\Adversary}^{\TwoP}(q_0, v)\)
    is reachable. 
    Moreover, for all transitions \(q \xrightarrow{v} q'\) that are reachable, we have the following:
    \begin{itemize}
        \item  if \(v \in \VerticesAdversary\), then 
        the transition from \(q'\) on input 
    \(\delta_{\Adversary}^{\TwoP}(q, v)\) is also reachable;
        \item if \(v \in \VerticesMain\), then for all vertices \(v' \in \Edges(v)\), the transition from \(q\) on input \(v'\) is also reachable.
    \end{itemize}
    \noindent 
    Since we do not know how \(M_{\Adversary}^{\TwoP}\) behaves along unreachable transitions, we exclude unreachable transitions in \(M_{\Adversary}^{\AS}\) and add reset transitions which we define now.
    
    For all \(q \in Q_{\Adversary}^{\TwoP}\) and all \(v \in \Vertices\), we define the transition function~\(\Delta_{\Adversary}\):
    \[
        \Delta_{\Adversary}(q, v) = 
        \begin{cases}
            \Delta_{\Adversary}^{\TwoP}(q, v)  & \text{if there exists } q' \in Q_{\Adversary}^{\TwoP} \text{ such that } q \xrightarrow{v} q' \text{ is reachable,} \\ 
            \Delta_{\Adversary}^{\TwoP}(q_0, v) & \text{otherwise.}
        \end{cases}
    \]
    The Mealy machine on a reset transition behaves in the way it would on reading \(v\) if it were in the initial state \(q_0\), that is, if the game began from \(v\). 
    This effectively resets the state of the Mealy machine. 
    For all \(q \in Q_{\Adversary}^{\TwoP}\) and all \(v \in \Vertices\), we define the output function \(\delta_{\Adversary}\):
    \[
        \delta_{\Adversary}(q, v) = 
        \begin{cases}
            \epsilon  & \text{if } v \in \VerticesRandom \union \VerticesMain, \\ 
            \delta_{\Adversary}^{\TwoP}(q, v)  & \text{if } v \in \VerticesAdversary \text{ and there exists } q' \in Q_{\Adversary}^{\TwoP} \text{ such that } q \xrightarrow{v} q' \text{ is reachable,} \\ 
            \delta_{\Adversary}^{\TwoP}(q_0, v) & \text{otherwise.}
        \end{cases}
    \]
    If \(v \in \VerticesRandom \union \VerticesMain\), then we have \(\delta_{\Adversary}(q, v)\) equal to \(\epsilon\) since \(\Strategy[\Adversary]^{\AS}\) is a \(\PlayerAdversaryDash\) strategy in the stochastic game \(\Game\) and is not defined for prefixes ending in vertices from \(\VerticesRandom \union \VerticesMain\).
    Otherwise, if \(v \in \VerticesAdversary\), then \(\delta_{\Adversary}(q, v)\) is defined in a similar manner as \(\Delta_{\Adversary}(q, v)\). 
    
    A state \(q'\) is unreachable if there does not exist \(q \in Q_{\Adversary}^{\TwoP}\) and \(v \in \Vertices\) such that the transition \(q \xrightarrow{v} q'\) is reachable.
    We delete the unreachable states so we have \(Q_{\Adversary} \subseteq Q_{\Adversary}^{\TwoP}\).
    Since the unreachable states do not have incoming reachable transitions, this does not delete any transition that is reachable. 
    Since the set of states in \(M_{\Adversary}^{\AS}\) is a subset of the set of states in \(M_{\Adversary}^{\TwoP}\), the memory size of \(\Strategy[\Adversary]^{\AS}\) is no greater than the memory size of \(\Strategy[\Adversary]^{\TwoP}\).
    This completes the construction of the Mealy machine \(M_{\Adversary}^{\AS}\) defining the reset strategy \(\Strategy[\Adversary]^{\AS}\).
    
    We now show that this strategy is almost-sure winning for \(\PlayerAdversary\) from all vertices. 
    For all \(q \in Q_{\Adversary}\) and \(v \in \Vertices\), the output function \(\Delta_{\Adversary}(q, v)\) and the transition function \(\delta_{\Adversary}(q, v)\) are defined, and hence, \(M_{\Adversary}^{\AS}\) is a complete Mealy machine. 
    Moreover, from the way \(M_{\Adversary}^{\AS}\) is constructed, we see that every transition in \(M_{\Adversary}^{\AS}\) is either reachable in \(M_{\Adversary}^{\TwoP}\) or is a reset transition.
    
    Suppose \(\PlayerAdversary\) plays in the stochastic game \(\Game\) according to \(M_{\Adversary}^{\AS}\) and this results in the prefix \(\Prefix \cdot v\) for some \(v \in \Vertices\).
    Let \(\Prefix' \cdot v\) be the infix obtained from \(\Prefix \cdot v\) by removing all vertices until the last occurrence of a reset transition in  \(M_{\Adversary}^{\AS}\). 
    In particular, if no reset transition occurs in \(M_{\Adversary}^{\AS}\) on reading \(\Prefix \cdot v\), then \(\Prefix' \cdot v\) is equal to \(\Prefix \cdot v\). 
    From the definition of resetting, we have that starting from the initial state \(q_0\) of \(M_{\Adversary}^{\AS}\), both \(\Prefix' \cdot v\) and \(\Prefix \cdot v\) take \(M_{\Adversary}^{\AS}\) to the same state, i.e., \(\hat{\Delta}_{\Adversary}(q_0, \Prefix' \cdot v) = \hat{\Delta}_{\Adversary}(q_0, \Prefix \cdot v) \).
    
    Since no reset transitions occur in \(M_{\Adversary}^{\AS}\) on reading \(\Prefix' \cdot v\), all transitions that occur are reachable if the Mealy machine \(M_{\Adversary}^{\TwoP}\) is used for the stochastic game \(\Game\).
    Therefore, on reading \(\Prefix' \cdot v\), the sequence of states visited in~\(M_{\Adversary}^{\TwoP}\) is the same as the sequence of states visited in \(M_{\Adversary}^{\AS}\). 
    In particular, the state of \(M_{\Adversary}^{\TwoP}\) on reading \(\Prefix' \cdot v\) is the same as the state of \(M_{\Adversary}^{\AS}\) on reading \(\Prefix' \cdot v\), which is also the same as the state of \(M_{\Adversary}^{\AS}\) on reading \(\Prefix \cdot v\).
    Thus, we have that \(\hat{\Delta}^{\TwoP}_{\Adversary}(q_0, \Prefix' \cdot v) = \hat{\Delta}_{\Adversary}(q_0, \Prefix \cdot v) \).
    Note that \(\Prefix' \cdot v\)  may contain deviations that \(M_{\Adversary}^{\AS}\) does not reset on.
    For instance, in Example~\ref{exa:reset-strategy-example}, if the token is on \(v_2\), then with positive probability, it moves to \(v_5\). 
    This is a deviation as \(M_{\Adversary}^{\TwoP}\) never moves the token from~\(v_2\) to~\(v_5\). 
    However, in doing so, the Mealy machine \(M_{\Adversary}^{\AS}\) does not follow any unreachable transition and does not reset. 
    Note that both \(M_{\Adversary}^{\AS}\) and \(M_{\Adversary}^{\TwoP}\), on reading the prefix \(v_1 v_2 v_5\) with a deviation, reach the same state \(q_3\).
    
    Given \(\Prefix' \cdot v\),  
    there exists a finite path \(\Prefix'' \cdot v\) of vertices without any deviations such that \(\hat{\Delta}_{\Adversary}^{\TwoP}(q_0, \Prefix' \cdot v) = \hat{\Delta}_{\Adversary}^{\TwoP}(q_0, \Prefix'' \cdot v) \).
    This is because the transition from state \(\hat{\Delta}_{\Adversary}^{\TwoP}(q_0, \Prefix') \) on input \(v\) is reachable in \(M_{\Adversary}^{\TwoP}\).
    Corresponding to the prefix \(v_1 v_2 v_5\) in Example~\ref{exa:reset-strategy-example}, we have that \(v_1 v_3 v_5\) is a finite path without deviations that takes \(M_{\Adversary}^{\AS}\) to the same state as \(v_1 v_2 v_5\).
    
    Thus, for every prefix \(\Prefix \cdot v\) of an outcome of~\(\Game\), there exists a finite path \(\Prefix'' \cdot v\) without deviations such that  \(\hat{\Delta}_{\Adversary}(q_0, \Prefix \cdot v) = \hat{\Delta}_{\Adversary}^{\TwoP}(q_0, \Prefix'' \cdot v) \).
    As long as no deviations occur,  the sequence of vertices seen after \(\Prefix \cdot v\) is the same irrespective of whether \(\PlayerAdversary\) uses the strategy \(\Strategy[\Adversary]^{\AS}\) or \(\Strategy[\Adversary]^{\TwoP}\).
    If a play in \(\Game\) continues for \(\Memory \cdot \abs{\Vertices} \cdot \WindowLength\) steps without deviating, then by Lemma~\ref{lem:open_windows_guarantee_general}, it contains an open window of length \(\WindowLength\). 
    From any point in the play, the probability that \(\Strategy[\Adversary]^{\AS}\) successfully copies \(\Strategy[\Adversary]^{\TwoP}\) for \(i\) steps   (that is, no deviations occur) is at least \(p^{i}\), where \(p\) is the minimum probability over all the edges in \(\Game\).
    It follows that from every point in the play, the probability that an open window of length \(\WindowLength\) occurs in the next \(\Memory \cdot \abs{\Vertices} \cdot \WindowLength\) steps is at least \(p^{\Memory \cdot \abs{\Vertices} \cdot \WindowLength}\). 
    Therefore, from every position in the play, the probability that an open window of length \(\WindowLength\) occurs eventually is at least \(\sum_{i \ge 0} (1 - p^{\Memory \cdot \abs{\Vertices} \cdot \WindowLength})^i \cdot p^{\Memory \cdot \abs{\Vertices} \cdot \WindowLength} = 1\).
    Thus, with probability~1, infinitely many open windows of length \(\WindowLength\) occur in the outcome, and the outcome satisfies \(\FWMPLBar\).
    Thus, all vertices in \(\Game\) are almost-sure winning for \(\PlayerAdversary\) for \(\FWMPLBar\).
    This concludes the construction of a reset strategy that is almost-sure winning for \(\PlayerAdversary\) from all vertices in the stochastic game.
    \qed
\end{construction}

We now construct a strategy \(\Strategy[\Adversary]^{\Pos}\) of \(\PlayerAdversary\) that is positive winning from all vertices in \(\PosWinningRegion[\Game]{\Adversary}{\FWMPLBar}\).
Let \(\WinningSet[i]{\Adversary}\) and \(\PosAttractorSet[i]{\Adversary}\) denote the sets \(\WinningSet{\Adversary}\) and \(\PosAttractorSet{\Adversary}\) computed in the \(i^{\text{th}}\)~recursive call of \(\ASWin_{\FWMPL}\) algorithm respectively. 
Here, \(\ASWin_{\FWMPL}\) is the algorithm obtained by instantiating \(\Objective\) to \(\FWMPL\) in Algorithm~\ref{alg:aswin}.
If the token is in \(\bigcup_i \WinningSet[i]{\Adversary}\), then \(\Strategy[\Adversary]^{\Pos}\) mimics \(\Strategy[\Adversary]^{\AS}\);  if the token is in \(\bigcup_i \PosAttractorSet[i]{\Adversary} \setminus \WinningSet[i]{\Adversary}\), then \(\Strategy[\Adversary]^{\Pos}\) is a positive-attractor strategy to \(\WinningSet[i]{\Adversary}\) which is memoryless.
Then, \(\Strategy[\Adversary]^{\Pos}\) is a positive winning strategy for \(\PlayerAdversary\) from all vertices in \( \PosWinningRegion[\Game]{\Adversary}{\FWMPL}\).

We have shown that for two-player stochastic games with \(\FWMPL\) objective, the memory requirements of optimal strategies of both players is no greater than that for non-stochastic games with the same objective.
\begin{rem}%
\label{rem:deterministic-strategy-reset-strategy}
    All plays consistent with the reset strategy of \(\PlayerAdversary\) described in Construction~\ref{con:reset-strategy-mealy-machine-adversary} are winning for \(\PlayerAdversary\).
    Thus, the reset strategy continues to be almost-sure winning even when \(\PlayerMain\) uses randomized strategies. 
    Since the reset strategy is a deterministic strategy, we have that deterministic strategies suffice for \(\PlayerAdversary\) for the positive and almost-sure winning of the \(\FWMPLBar\) objective.
\end{rem}
From~\cite{CDRR15}, we have that the satisfaction problem for the \(\FWMPL\) objective in non-stochastic games is in \(\PTime\).
Thus, from Theorem~\ref{thm:twop_to_stochastic}, Corollary~\ref{cor:quantitative_satisfaction_complexity}, and Lemma~\ref{lem:fwmp-sas}, we have the following.

\begin{thm}%
\label{thm:posfwmp_algorithm_correctness}
    Given a stochastic game \(\Game\), a window length \(\WindowLength \ge 1\), and a threshold \(p \in [0,1]\), for \(\FWMPL[\Game]\), the 
    positive and almost-sure satisfaction problems for \(\PlayerMain\) 
    are in \(\PTime\), and the quantitative satisfaction problem is in \(\NP \cap \coNP\).
    Moreover for optimal strategies, memory of size \(\WindowLength\) is sufficient for \(\PlayerMain\) and memory of size \(\abs{V} \cdot \WindowLength\) 
    is sufficient for \(\PlayerAdversary\).
\end{thm}

\begin{exa}
    Consider the stochastic game \(\Game\) shown in \figurename~\ref{fig:posfwmp_example_stochastic},
    and objective \(\FWMPL\) with window length $\WindowLength = 2$. 
    \begin{figure}[t]
        \begin{subfigure}[b]{0.45\textwidth}
            \centering
            \begin{tikzpicture}
                \node[random, draw] (v1) {\(\Vertex[1]\)};
                \node[square, draw, left of=v1] (v2) {\(\Vertex[2]\)};
                \node[state, right of=v1] (v3) {\(\Vertex[3]\)};
                \draw 
                      (v1) edge[bend right, above] node{\(\EdgeValues{0}{.2}\)} (v2)
                      (v1) edge[auto] node{\(\EdgeValues{0}{.8}\)} (v3)
                      (v2) edge[bend right, below] node{\(\EdgeValues{-1}{}\)}  (v1)
                      (v2) edge[loop left] node{\(\EdgeValues{0}{}\)} (v2)
                      (v3) edge[loop right] node{\(\EdgeValues{0}{}\)}(v3)   
                ;
            \end{tikzpicture}
            \caption{A stochastic game \(\Game\) with three vertices.}
            \label{fig:posfwmp_example_stochastic}
        \end{subfigure}
        \hfill
        \begin{subfigure}[b]{0.45\textwidth}
            \centering
            \begin{tikzpicture}
                \node[square, draw] (v1) {\(\Vertex[1]\)};
                \node[square, draw, left of=v1] (v2) {\(\Vertex[2]\)};
                \node[state, right of=v1] (v3) {\(\Vertex[3]\)};
                \draw 
                      (v1) edge[bend right, above] node{\(\EdgeValues{0}{}\)} (v2)
                      (v1) edge[auto] node{\(\EdgeValues{0}{}\)} (v3)
                      (v2) edge[bend right, below] node{\(\EdgeValues{-1}{}\)}  (v1)
                      (v2) edge[loop left] node{\(\EdgeValues{0}{}\)} (v2)
                      (v3) edge[loop right] node{\(\EdgeValues{0}{}\)}(v3)   
                ;
            \end{tikzpicture}
            \caption{The non-stochastic game corresponding to \(\Game\).}
            \label{fig:poswmp1_example_adversarial}
        \end{subfigure}
        \centering
        \caption{For all \(\WindowLength \ge 1\), \(\PlayerMain\) can positively satisfy \(\FWMPL\) from every vertex in \(\Game\).}
        \label{fig:poswmp1_example}
    \end{figure}
    It is easy to see that all vertices are positively (even almost-surely) 
    winning for \(\PlayerMain\) in \(\Game\).
    We compute the positive winning region as follows. First, consider the
    non-stochastic game \(\TwoPlayerGame\) (\figurename~\ref{fig:poswmp1_example_adversarial}). The winning region for \(\PlayerMain\)
    in \(\TwoPlayerGame\) is $\{\Vertex[3]\}$, and we thus have that the
    \(\PlayerMain\) positive attractor of $\{\Vertex[3]\}$, which is $\{\Vertex[1],\Vertex[3]\}$ is positively winning. The complement of the positive attractor induces the subgame
    with a single vertex~$\Vertex[2]$, that can be solved recursively to get that \(\PlayerMain\)
    positively (even almost-surely) wins from $\Vertex[2]$. 
    Therefore, we conclude that \(\PlayerMain\) is positive winning from every vertex.
    \qed
\end{exa}

\begin{exa}
    Consider the stochastic game \(\Game\) shown in \figurename~\ref{fig:poswmp2_example}, and objective \(\FWMPL\) with window length \(\WindowLength = 2\). 
    \begin{figure}[t]
        \centering
        \begin{tikzpicture}
            \node[state] (v1) {\(\Vertex[1]\)};
            \node[random, draw, right of=v1] (v2) {\(\Vertex[2]\)};
            \node[square, draw, right of=v2, xshift=+5mm] (v3) {\(\Vertex[3]\)};
            \node[state, right of=v3, xshift=+5mm] (v4) {\(\Vertex[4]\)};
            \draw 
                  (v1) edge[loop left] node{\(\EdgeValues{-1}{}\)} (v1)
                  (v1) edge[auto] node{\(\EdgeValues{0}{}\)} (v2)
                  (v2) edge[auto] node{\(\EdgeValues{0}{.1}\)} (v3)
                  (v2) edge[bend right, below left, pos=0.2] node{\(\EdgeValues{0}{.9}\)} (v4)
                  (v3) edge[loop right] node{\(\EdgeValues{0}{}\)} (v3)
                  (v4) edge[loop right] node{\(\EdgeValues{-1}{}\)} (v4)
            ;
        \end{tikzpicture}
        \caption{\(\PlayerMain\) almost surely wins in \(\Game\) for the objective \(\FWMPObj{2}\) from \(\{\Vertex[3]\}\), while \(\PlayerAdversary\) positively wins from \(\{\Vertex[1], \Vertex[2], \Vertex[4]\}\).}
        \label{fig:poswmp2_example}
    \end{figure}
    We compute the almost-sure winning region for \(\PlayerMain\) by first computing the positive winning region for \(\PlayerAdversary\), which we do as follows. 
    Using \(\PosWin_{\FWMPL}\), the positive winning region for \(\PlayerMain\) in \(\Game\) is \(\{\Vertex[1], \Vertex[2], \Vertex[3]\}\). 
    The complement of this set, \(\{\Vertex[4]\}\), is the almost-sure winning region for \(\PlayerAdversary\) in \(\Game\). 
    The \(\PlayerAdversaryDash\) positive attractor of \(\{\Vertex[4]\}\) is \(\{\Vertex[2], \Vertex[4]\}\), and we can conclude that this set is positively winning for \(\PlayerAdversary\). The complement of the positive attractor induces the subgame with vertices \(\{\Vertex[1], \Vertex[3]\}\), which can be solved recursively to get that \(\PlayerAdversary\) positively (even almost-surely) wins from \(\{\Vertex[1]\}\) but does not win even positively from \(\{\Vertex[3]\}\). 
    Therefore, we conclude that \(\PlayerAdversary\) positively wins in the original stochastic game from \(\{\Vertex[1], \Vertex[2], \Vertex[4]\}\), and \(\PlayerMain\) almost surely wins from the complement \(\{\Vertex[3]\}\).
    \qed
\end{exa}

\subsection{Bounded window mean-payoff objective}%
\label{sub:bwmp}

We show that the \(\SAS\)~property holds for the objective \(\BWMP[\Game]\) for all stochastic games \(\Game\).

\begin{lem}%
\label{lem:bwmp-sas}
    For all stochastic games \(\Game\), the objective \(\BWMP\) satisfies the \(\SAS\)~property. 
\end{lem}

\begin{proof}
    We need to show that for all stochastic games \(\Game\), if \(\TwoPlayerWinningRegion[\TwoPlayerGame]{\Adversary}{\BWMPBar} = \Vertices\), then  \(\ASWinningRegion[\Game]{\Adversary}{\BWMPBar} = \Vertices\). 
    Since every play that satisfies \(\BWMPBar\) also satisfies \(\FWMPLBar\) for all \(\WindowLength \ge 1\), we have that \(\TwoPlayerWinningRegion[\TwoPlayerGame]{\Adversary}{\BWMPBar}  = \Vertices\) implies \(\TwoPlayerWinningRegion[\TwoPlayerGame]{\Adversary}{\FWMPLBar} = \Vertices\). 
    It follows that for each \(\WindowLength \ge 1\), \(\PlayerAdversary\) has a finite-memory strategy (say, with memory \(\Memory_{\WindowLength}\)), that is winning for the \(\FWMPLBar\) objective from all vertices in \(\TwoPlayerGame\).
    For every such strategy, we construct a reset strategy \(\Strategy[\Adversary]^{\WindowLength}\) of memory size at most \(\Memory_{\WindowLength}\) as described in the proof of Lemma~\ref{lem:fwmp-sas} that is almost-sure winning for the \(\FWMPLBar\) objective from all vertices.
    We use these strategies to construct an infinite-memory strategy \(\Strategy[\Adversary]^{\AS}\) of \(\PlayerAdversary\) that is almost-surely winning for \(\BWMPBar\) from all vertices in the stochastic game~\(\Game\). 
    
    Let \(p\) be the minimum probability over all edges in the game, and for all \(\WindowLength \ge 1\), let \(q(\WindowLength)\) denote \(p^{\Memory_{\WindowLength} \cdot \abs{\Vertices} \cdot \WindowLength}\). 
    We partition a play of the game into phases \(\Phase{1}, \Phase{2}, \ldots\) such that for all \(\WindowLength \ge 1\), the length of phase \(\Phase{\WindowLength}\) is equal to \(\Memory_{\WindowLength} \cdot \abs{\Vertices} \cdot \WindowLength \cdot \left\lceil 1/q(\WindowLength)\right\rceil\).
    We define the strategy \(\Strategy[\Adversary]^{\AS}\) as follows: 
    if the game is in phase~\(\WindowLength\), then \(\Strategy[\Adversary]^{\AS}\) is \(\Strategy[\Adversary]^{ \WindowLength}\),
    the reset strategy that is almost-sure winning for \(\FWMPLBar\) in \(\Game\). 
    
    We show that \(\Strategy[\Adversary]^{\AS}\) is almost-sure winning for \(\PlayerAdversary\) for \(\BWMPBar\) in~\(\Game\).
    Let \(E_\WindowLength\) denote the event that phase \(\Phase{\WindowLength}\) contains an open window of length \(\WindowLength\). 
    Given a play \(\Play\), if \(E_\WindowLength\) occurs in~\(\Play\) for infinitely many \(\WindowLength \ge 1\), then for every suffix of \(\Play\) and for all \(\WindowLength \ge 1\), the suffix contains an open window of length \(\WindowLength\), and \(\Play\) satisfies \(\BWMPBar\). 
    For all \(\WindowLength \ge 1\), we compute the probability that \(E_\WindowLength\) occurs in the outcome. 
    For all \(\WindowLength \ge 1\), we can divide phase \(\WindowLength\) into \(\left\lceil 1/q(\WindowLength) \right\rceil\) blocks of length \(\Memory_{\WindowLength} \cdot \abs{\Vertices} \cdot \WindowLength\) each.
    If at least one of these blocks contains an open window of length~\(\WindowLength\), then the event \(E_\WindowLength\) occurs.
    It follows from the proof of Lemma~\ref{lem:fwmp-sas} that if \(\PlayerAdversary\) follows~\(\Strategy[\Adversary]^{\WindowLength}\), then the probability that there exists an open window of length \(\WindowLength\) in the next \(\Memory_{\WindowLength} \cdot \abs{\Vertices} \cdot \WindowLength\) steps is at least \(q(\WindowLength)\).
    Hence, the probability that none of the blocks in the phase contains an open window of length \(\WindowLength\) is at most \((1-q(\WindowLength))^{\left\lceil 1/q(\WindowLength) \right\rceil}\).
    Thus, the probability that \(E_\WindowLength\) occurs in phase \(\Phase{\WindowLength}\) is at least \( 1 - (1-q(\WindowLength))^{\left\lceil 1/q(\WindowLength)\right\rceil} > 1-\frac{1}{e} \approx 0.63 > 0\).
    It follows that with probability~$1$, for infinitely many values of \(\WindowLength \ge 1\), the event \(E_{\WindowLength}\) occurs in \(\Play\). 
    
    To show that \(E_\WindowLength\) occurs for infinitely many \(\WindowLength \ge 1\) in the outcome with probability~\(1\),%
    \footnote{%
        The sum \(\sum_{\WindowLength \ge 1} \Prob(E_\WindowLength)\) diverges to infinity.  If we can show that the events \(E_\WindowLength\) are independent, then by the second Borel-Cantelli lemma \cite{Durrett2010}, this would directly imply that the probability of infinitely many of \(E_\WindowLength\) occurring is 1. 
        However, we do not know if they are independent, so we are not able to apply the Borel-Cantelli lemma.
    }
    we show an equivalent statement: the probability that \(E_{\WindowLength}\) occurs for only finitely many values of \(\WindowLength \ge 1\) in the outcome is~\(0\). 
    Let \(F\) be the set of all plays consistent with \(\Strategy[\Adversary]^{\AS}\) in which only finitely many \(E_\WindowLength\) occur. 
    We construct countably many subsets \(F_0, F_1, \ldots\) of \(F\) as follows: 
    let~\(F_0\) be the set of all plays in \(F\) in which \(E_k\) does not occur for all \(k \ge 1\);
    and for all \(j \ge 1\), 
    let \(F_j\) consist of all plays in \(F\) in which \(E_j\) occurs, but for all \(k > j\), the events \(E_k\) do not occur (and for \(i < j\), the event \(E_i\) may or may not occur).
    Observe that \(\bigcup_{j \ge 0} F_j = F\) and \(F_i \intersection F_j = \emptyset\) for all \(i \ne j\). 
    
    For all \(k \ge 1\), the probability that \(E_k\) does not occur is at most \((1-q(k))^{\left\lceil 1/q(k) \right\rceil}\) which is at most 0.37, irrespective of whether any other \(E_j\)'s occur or not (again, this is because the probability that a block contains an open window of length \(\WindowLength\) is at least \(q(\WindowLength)\), independent of what happens in the rest of the play).
    For all \(j \ge 0\), in the event \(F_j\), for all \(k > j\), we have that \(E_k\) does not occur. 
    Since each \(E_k\) does not occur with probability at most 0.37, the probability of \(F_j\) is at most \(\prod_{k > j} (0.37)\), which is~\( 0\). 
    The event that finitely many \(E_\WindowLength\)'s occur is the countable union of disjoint events \(\bigcup_{j \ge 0} F_j\). 
    Since the probability measure of each~\(F_j\) is zero, and a countable sum of zero measure events has zero measure~\cite{Rudin87}, this implies that finitely many \(E_{\WindowLength}\) occur with probability zero. 
    Thus, the probability that \(E_{\WindowLength}\) occurs for infinitely many \(i \ge 1\) is~1.

    Hence, the objective \(\BWMPBar\) is satisfied with probability~\(1\) from all vertices in the stochastic game \(\Game\),
    and we have that \(\ASWinningRegion[\Game]{\Adversary}{\BWMPBar} = \Vertices\).
\end{proof}

Note that Lemma~2 in~\cite{CHH09} is similar to Lemma~\ref{lem:bwmp-sas} but refers to a different objective (finitary Streett instead of \(\BWMP\)).
However, the proofs have the following differences.
In our proof, each phase of a play lasts for a fixed predetermined length and we show that for all \(\WindowLength \ge 1\), in phase \(\WindowLength\), the probability that an open window of length \(\WindowLength\) occurs is at least 0.37, which is independent of \(\WindowLength\).  We use this to conclude that with probability 1, for infinitely many \(\WindowLength \ge 1\), phase \(\WindowLength\) contains an open window of length \(\WindowLength\), and thus with probability 1, the play satisfies \(\BWMPBar\). 
In the proof in \cite{CHH09}, a play continues to be in the \(\WindowLength^{\text{th}}\) phase until an open window of length \(\WindowLength\) appears.
The \(\WindowLength^{\text{th}}\) phase does not end until an open window of length \(\WindowLength\) is observed in the phase. 
They show that for each phase in the play, the phase ends with probability 1, and thus with probability 1, the play contains open windows of length \(\WindowLength\) for all \(\WindowLength \ge 1\).

Note that solving a non-stochastic game with the \(\BWMP\) objective is in \(\NP \intersection \coNP\) \cite{CDRR15}.
Thus by Corollary~\ref{cor:quantitative_satisfaction_complexity}, quantitative satisfaction for \(\BWMP\) is in \(\NP^{\NP \intersection \coNP} \intersection \coNP^{\NP \intersection \coNP}\). 
From \cite{Sch83}, we have that \(\NP^{\NP \intersection \coNP} = \NP\) and \(\coNP^{\NP \intersection \coNP} = \coNP\). 
To see this, suppose for an alphabet \(\Sigma\), if \(L \subseteq \Sigma^* \) is a language in \(\NP \intersection \coNP\), then for all \(x \in \Sigma^*\), there either exists a short witness for \(x\) belonging to \(L\) or a short witness for \(x\) not belonging to \(L\).
A nondeterministic Turing machine can guess one of these witnesses and verify in polynomial time whether \(x \in L\) or \(x \notin L\).
Hence, an \(\NP \intersection \coNP\) oracle can be simulated by an \(\NP\) machine, and we have \(\NP^{\NP \intersection \coNP} = \NP\). 
For an oracle \(\mathcal{A}\), a language \(L\) belongs to \(\coNP^{\mathcal{A}}\) if and only if its complement \(\overline{L}\) belongs to \(\NP^{\mathcal{A}}\). 
Since the quantitative satisfaction problem belongs to \(\coNP^{\mathcal{A}}\), its complement belongs to \(\NP^{\mathcal{A}}\), which is \(\NP\) for \(\mathcal{A}\) in \(\NP \intersection \coNP\), and thus, the value problem belongs to \(\coNP\).
Therefore, quantitative satisfaction of \(\BWMP\) is in \(\NP \intersection \coNP\).

Moreover, from~\cite{CDRR15}, \(\PlayerMain\) has a memoryless strategy and \(\PlayerAdversary\) needs infinite memory to play optimally in non-stochastic games with \(\BWMP\) objective.
From the proof of Lemma~\ref{lem:bwmp-sas}, 
by using the strategy \(\Strategy[\Adversary]^{\AS}\), \(\PlayerAdversary\) almost-surely wins \(\BWMPBar\) from all vertices in \(\ASWinningRegion[\Game]{\Adversary}{\BWMPBar}\).
We can construct a positive winning strategy \(\Strategy[\Adversary]^{\Pos}\) for \(\PlayerAdversary\) from all vertices in \(\PosWinningRegion[\Game]{\Adversary}{\BWMPBar}\) in a similar manner as done for the positive winning strategy for \(\FWMPLBar\) in Section~\ref{sub:fwmp-sas}.
Using similar reasoning as in Remark~\ref{rem:deterministic-strategy-reset-strategy} in Section~\ref{sub:fwmp-sas}, it follows that deterministic strategies suffice for \(\PlayerAdversary\) for the positive and almost-sure satisfaction of the \(\BWMPBar\) objective.
We summarize the results in the following theorem:
\begin{thm}
\label{thm:quantitative_bwmp}
    Given a stochastic game \(\Game\) and a threshold \(p \in [0,1]\), for \(\BWMP[\Game]\), the positive, almost-sure, and quantitative satisfaction for \(\PlayerMain\) are in \(\NP \cap \coNP\).
    Moreover, a memoryless strategy suffices for \(\PlayerMain\), while \(\PlayerAdversary\) requires an infinite memory strategy to play optimally.
\end{thm}

\begin{rem}
Solving non-stochastic games with \(\BWMP\) objective is at least as hard as solving traditional mean-payoff games~\cite{CDRR15}. 
Since non-stochastic games are a special case of stochastic games, it follows that solving the positive, almost-sure, and quantitative satisfaction problems for stochastic games with \(\BWMP\) objective is at least as hard as solving traditional mean-payoff games. 
\end{rem}

\section*{Acknowledgment}
We thank Mickael Randour for pointing out reference~\cite{BHR16a}, making us aware of the bugs in the algorithms of~\cite{CDRR15}  and the correct version of these algorithms as presented here.

\bibliographystyle{alphaurl}
\bibliography{mybib}

\end{document}